\documentclass{article}

\def\noheaderplainsetup{

\topmargin=0pt \headheight=0pt \headsep=0pt  \oddsidemargin=0pt \evensidemargin=0pt  \textheight=8.9truein \textwidth=6.2truein}

\noheaderplainsetup

\usepackage{amsfonts}

\begin{document}

%     MISC.:

\newcommand{\code}[1]{\ulcorner #1 \urcorner}
\newcommand{\mldi}{\hspace{2pt}\mbox{\footnotesize $\vee$}\hspace{2pt}}
\newcommand{\mlci}{\hspace{2pt}\mbox{\footnotesize $\wedge$}\hspace{2pt}}
\newcommand{\emptyrun}{\langle\rangle} %empty run
\newcommand{\oo}{\bot}            %opponent, the corresponding truth value and label
\newcommand{\pp}{\top}            %proponent, the corresponding truth value and label
\newcommand{\xx}{\wp}               %player, the corresponding truth value and label
\newcommand{\legal}[2]{\mbox{\bf Lr}^{#1}_{#2}} %function telling what the legal positions are
\newcommand{\win}[2]{\mbox{\bf Wn}^{#1}_{#2}} %function telling who the winner is
 \newcommand{\one}{\mbox{\sc One}}
 \newcommand{\two}{\mbox{\sc Two}}
 \newcommand{\three}{\mbox{\sc Three}}
 \newcommand{\four}{\mbox{\sc Four}}
 \newcommand{\first}{\mbox{\sc Derivation}}
 \newcommand{\second}{\mbox{\sc Second}}
 \newcommand{\uorigin}{\mbox{\sc Org}}
 \newcommand{\image}{\mbox{\sc Img}}
 \newcommand{\limitset}{\mbox{\sc Lim}}
 \newcommand{\fif}{\mbox{\bf CL15}}
\newcommand{\col}[1]{\mbox{$#1$:}}

\newcommand{\sti}{\mbox{\raisebox{-0.02cm}
{\scriptsize $\circ$}\hspace{-0.121cm}\raisebox{0.08cm}{\tiny $.$}\hspace{-0.079cm}\raisebox{0.10cm}
{\tiny $.$}\hspace{-0.079cm}\raisebox{0.12cm}{\tiny $.$}\hspace{-0.085cm}\raisebox{0.14cm}
{\tiny $.$}\hspace{-0.079cm}\raisebox{0.16cm}{\tiny $.$}\hspace{1pt}}}
\newcommand{\costi}{\mbox{\raisebox{0.08cm}
{\scriptsize $\circ$}\hspace{-0.121cm}\raisebox{-0.01cm}{\tiny $.$}\hspace{-0.079cm}\raisebox{0.01cm}
{\tiny $.$}\hspace{-0.079cm}\raisebox{0.03cm}{\tiny $.$}\hspace{-0.085cm}\raisebox{0.05cm}
{\tiny $.$}\hspace{-0.079cm}\raisebox{0.07cm}{\tiny $.$}\hspace{1pt}}}

\newcommand{\seq}[1]{\langle #1 \rangle}           % sequence: <...>

%     OPERATORS:

\newcommand{\mla}{\mbox{{\Large $\wedge$}}}
\newcommand{\mle}{\mbox{{\Large $\vee$}}}

\newcommand{\pst}{\mbox{\raisebox{-0.01cm}{\scriptsize $\wedge$}\hspace{-4pt}\raisebox{0.16cm}{\tiny $\mid$}\hspace{2pt}}}
\newcommand{\gneg}{\neg}                  %game negation
\newcommand{\mli}{\rightarrow}                     %strong reduction
\newcommand{\cla}{\mbox{\large $\forall$}}      %blind universal quantifier
\newcommand{\cle}{\mbox{\large $\exists$}}        %blind existential quantifier
\newcommand{\mld}{\vee}    %multiplicative disjunction
\newcommand{\mlc}{\wedge}  %multiplicative conjunction
\newcommand{\ade}{\mbox{\Large $\sqcup$}}      %additive existential quantifier
\newcommand{\ada}{\mbox{\Large $\sqcap$}}      %additive universal quantifier
\newcommand{\add}{\sqcup}                      %additive disjunction
\newcommand{\adc}{\sqcap}                      %additive conjunction

\newcommand{\tlg}{\bot}               %classical \bot; trivially lost elementary game
\newcommand{\twg}{\top}               %classical \top; trivially won elementary game
\newcommand{\st}{\mbox{\raisebox{-0.05cm}{$\circ$}\hspace{-0.13cm}\raisebox{0.16cm}{\tiny $\mid$}\hspace{2pt}}}
\newcommand{\cst}{{\mbox{\raisebox{-0.05cm}{$\circ$}\hspace{-0.13cm}\raisebox{0.16cm}{\tiny $\mid$}\hspace{1pt}}}^{\aleph_0}} % countable recurrence
\newcommand{\cost}{\mbox{\raisebox{0.12cm}{$\circ$}\hspace{-0.13cm}\raisebox{0.02cm}{\tiny $\mid$}\hspace{2pt}}}
\newcommand{\ccost}{{\mbox{\raisebox{0.12cm}{$\circ$}\hspace{-0.13cm}\raisebox{0.02cm}{\tiny $\mid$}\hspace{1pt}}}^{\aleph_0}} % countable corecurrence
\newcommand{\pcost}{\mbox{\raisebox{0.12cm}{\scriptsize $\vee$}\hspace{-4pt}\raisebox{0.02cm}{\tiny $\mid$}\hspace{2pt}}}

%   NUMERATED ITEMS and ENVIRONMENTS

\newtheorem{theoremm}{Theorem}[section]
\newtheorem{conditionss}[theoremm]{Condition}
\newtheorem{thesiss}[theoremm]{Thesis}
\newtheorem{definitionn}[theoremm]{Definition}
\newtheorem{lemmaa}[theoremm]{Lemma}
\newtheorem{notationn}[theoremm]{Notation}
\newtheorem{propositionn}[theoremm]{Proposition}
\newtheorem{conventionn}[theoremm]{Convention}
\newtheorem{examplee}[theoremm]{Example}
\newtheorem{remarkk}[theoremm]{Remark}
\newtheorem{factt}[theoremm]{Fact}
\newtheorem{exercisee}[theoremm]{Exercise}
\newtheorem{questionn}[theoremm]{Open Problem}
\newtheorem{conjecturee}[theoremm]{Conjecture}

\newenvironment{exercise}{\begin{exercisee} \em}{ \end{exercisee}}
\newenvironment{definition}{\begin{definitionn} \em}{ \end{definitionn}}
\newenvironment{theorem}{\begin{theoremm}}{\end{theoremm}}
\newenvironment{lemma}{\begin{lemmaa}}{\end{lemmaa}}
\newenvironment{proposition}{\begin{propositionn} }{\end{propositionn}}
\newenvironment{convention}{\begin{conventionn} \em}{\end{conventionn}}
\newenvironment{remark}{\begin{remarkk} \em}{\end{remarkk}}
\newenvironment{proof}{ {\bf Proof.} }{\  \rule{2.5mm}{2.5mm} \vspace{.2in} }
\newenvironment{idea}{ {\bf Idea.} }{\  \rule{1.5mm}{1.5mm} \vspace{.15in} }
\newenvironment{example}{\begin{examplee} \em}{\end{examplee}}
\newenvironment{fact}{\begin{factt}}{\end{factt}}
\newenvironment{notation}{\begin{notationn} \em}{\end{notationn}}
\newenvironment{conditions}{\begin{conditionss} \em}{\end{conditionss}}
\newenvironment{question}{\begin{questionn}}{\end{questionn}}
\newenvironment{conjecture}{\begin{conjecturee}}{\end{conjecturee}}

\title{The taming of recurrences in computability logic through cirquent calculus, Part II}
\author{Giorgi Japaridze\thanks{Supported by 2010 Summer Research Fellowship from Villanova University}}
\date{}
\maketitle

\begin{abstract} This paper constructs a cirquent calculus system and proves its soundness and completeness with respect to the semantics of computability logic. The logical vocabulary of the system consists of negation $\gneg$, parallel conjunction $\mlc$, parallel disjunction $\mld$,  branching recurrence $\st$, and branching corecurrence $\cost$.  The article is published in two parts, with (the previous) Part I containing preliminaries and a soundness proof, and (the present) Part II containing a  completeness proof. 
\end{abstract}

\noindent {\em MSC}: primary: 03B47; secondary: 03B70; 68Q10; 68T27; 68T15. 

\  

\noindent {\em Keywords}: Computability logic; Cirquent calculus; Interactive computation; Game semantics; Resource semantics. 

%\tableofcontents

\section{Outline} 
Being a continuation of \cite{taming1}, this article fully relies on the terminology and notation introduced in its predecessor, with which --- or, at least, with the first six sections of which --- the reader is assumed to be already familiar, and which is necessary to have at hand for references while reading this paper. 

The purpose of the present piece of writing is to prove the completeness of $\fif$, in the form of clauses (b) and (c) of  Theorem 6 of  \cite{taming1}. For the rest of this paper, we fix an arbitrary formula $\mathbb{F}_0$ and assume that 
%\marginpar{feb17a}
\begin{equation}\label{feb17a}
\fif\not\vdash\mathbb{F}_0.
\end{equation}
Our immediate and most challenging goal, to which Sections \ref{ssmam}-\ref{notice} are devoted,  is to prove that $\mathbb{F}_{0}$ is not uniformly valid. The final  Section \ref{sss} will then relatively painlessly extend this result from uniform to multiform (in)validity. 

We are going to show that there is a {\em counterstrategy} (in fact, an effective one) $\cal E$ such that, when the environment plays according to $\cal E$,  no HPM wins $\mathbb{F}^{*}_{0}$ for an appropriately selected constant interpretation $^*$. Of course, we have never defined the concept of an environment's effective strategy. As explained in the following section, the latter, just like a machine's strategy, can be understood simply as an EPM.

\section{Machines against machines}\label{ssmam}

Here we borrow a discussion from \cite{Japtocl1}. 
 For a run $\Gamma$ and a computation branch $B$ of an EPM, we say that $B$ {\bf cospells} $\Gamma$ iff
$B$ spells $\gneg\Gamma$ ($\Gamma$ with all labels reversed) in the sense of Section 2.5 of \cite{taming1}. 
Intuitively, when an EPM $\cal E$ plays as  $\oo$ rather than $\pp$,  the run that is generated by  a given computation branch $B$ of $\cal E$ is the run cospelled rather than spelled by $B$, for the moves that $\cal E$ makes 
get the label $\oo$, and the moves that its adversary makes get the label $\pp$.

We say that an EPM $\cal E$ is {\bf fair} iff, for every valuation $e$, every $e$-computation branch  
of $\cal E$ is fair (again, ``fair branch'' in the sense of Section 2.5 of \cite{taming1}).

\begin{lemma}\label{lem}
%\marginpar{lem}
Assume $\cal E$ is a fair EPM, $\cal H$ is an HPM, and $e$ is a valuation. There are a uniquely defined 
 $e$-computation branch $B_{\cal E}$ of $\cal E$ and a uniquely defined $e$-computation branch $B_{\cal H}$ of $\cal H$
--- which we respectively call {\bf the $({\cal E},e,{\cal H})$-branch} and {\bf the $({\cal H},e,{\cal E})$-branch}
 --- such that the run spelled by $B_{\cal H}$, called  {\bf the $\cal H$ vs. $\cal E$ run on $e$}, 
 is the run cospelled by $B_{\cal E}$.\end{lemma}
 
When ${\cal H},{\cal E},e$ are as above, $\Gamma$ is the $\cal H$ vs.\hspace{-1pt} $\cal E$ run on $e$ and $A$ is a game such that $\Gamma$ is a $\pp$-won (resp. $\oo$-won) run of $e[A]$,  we say that $\cal H$ {\bf wins}
(resp. {\bf loses}) $A$ {\bf against $\cal E$ on $e$}. Simply saying ``$\cal H$  {\bf wins}
(resp. {\bf loses}) $A$ {\bf against} $\cal E$'' means that $\cal H$ wins
(resp. loses) $A$ against $\cal E$ on some valuation $e$.  

A strict proof of the above lemma can be found in \cite{Jap03} (Lemma 20.4), and we will not reproduce  
the formal proof here.  Instead, the following intuitive explanation should suffice:\vspace{7pt}

{\bf Proof idea.} Assume $\cal H$, $\cal E$, $e$ are as in Lemma \ref{lem}. The scenario that we are going to describe is the unique 
play generated when the two machines compete against each other, with $\cal H$ in the role of $\pp,$  $\cal E$ in the role of $\oo$, and $e$ spelled on the valuation tapes of both machines. 
We can visualize this play as follows.
Most of the time during the process $\cal H$ remains inactive (sleeping); it is woken up only when $\cal E$ enters a permission state, on which event $\cal H$ makes a (one single) transition to its next computation step --- that may or may not result in making a move --- and goes back into a sleep that will continue until $\cal E$ enters  a permission state again, and so on. From ${\cal E}$'s perspective, $\cal H$ acts as a patient adversary who makes one or zero move only when granted permission, just as the EPM-model assumes.  And from $\cal H$'s perspective, which, like a person in a coma,  has no sense of time during its sleep and hence can think that the wake-up events that it calls the beginning of a clock cycle happen at a constant rate, $\cal E$ acts as an adversary who can make any finite number of moves during a clock cycle (i.e. while $\cal H$ was sleeping), just as the HPM-model assumes. This scenario uniquely determines an $e$-computation branch $B_{\cal E}$ of $\cal E$ that we call
the $({\cal E},e,{\cal H})$-branch, and an $e$-computation branch $B_{\cal H}$ of $\cal H$ that we call
the $({\cal H},e,{\cal E})$-branch. What we call the $\cal H$ vs. $\cal E$ run on $e$ is the run generated 
in this play. In particular --- since we let $\cal H$ play in the role of $\pp$ --- this is the run spelled by $B_{\cal H}$. $\cal E$, who plays in the role of $\oo$, sees the same run, only it sees the labels of the moves of that run in negative colors. 
That is, $B_{\cal E}$ cospells rather than spells that run. This is exactly what Lemma \ref{lem} asserts.

\section{Enumeration games}\label{sseg}
%\marginpar{sseg}
 We continue identifying natural numbers with their decimal representations.

An {\bf enumeration game} is a game in every instance of which any natural number $a$ is a legal move by either player $\xx$ at any time, and there are no other legal moves.  Either player can thus be seen to be enumerating a set of numbers --- the numbers made by it as moves during the play. The winner in a (legal) play of a given instance of an enumeration game only depends on the two sets enumerated this way. That is, it only matters {\em what} moves have been eventually made, regardless of {\em when} (in what order) and {\em how many times} (only once or repetitively) those moves were made. 

Since the order of moves is irrelevant, it is obvious that every enumeration game is static, and hence is a legitimate value of an interpretation $^*$ on any atom.  We call an interpretation $^*$ that sends every atom $P$ to an enumeration game $P^*$ an {\bf enumeration interpretation}. From now on, we will limit our considerations only to enumeration interpretations.   Note that, under this restriction, the set of legal runs of $e[\mathbb{F}_{0}^{*}]$ does not depend on the valuation $e$ or the interpretation $^*$, because all instances of all enumeration games have the same set of legal runs. In view of this fact, we may safely talk about legal moves or runs of $e[\mathbb{F}_{0}^{*}]$ even without having yet defined the interpretation $^*$ and/or without having specified the valuation $e$. Furthermore, as was done in \cite{taming1}, in many contexts we shall terminologically and notationally identify $\mathbb{F}_{0}$  or any other formula or cirquent with the game into which it turns after a constant enumeration interpretation $^*$ is applied to it.

\section{Units, funits and prompts}\label{ssufp}
%\marginpar{ssufp}
 In what follows,  ``{\bf formula}'',  ``{\bf atom}'' and ``{\bf literal}'' always mean those of $\mathbb{F}_0$. 
Extending the usage of the prefix ``o'' introduced in Definition 3.1 of \cite{taming1}, by an {\bf osubformula} of a given formula $E$  we mean a subformula of $E$ together with a fixed occurrence of it within $E$. 
  Similarly for 
 {\bf oliterals}. Whenever we simply say ``{\bf oformula}'', unless otherwise specified or implied by the context, we always mean an osubformula of $\mathbb{F}_0$. We say that an oformula $E$ is an {\bf osuperformula} of an oformula $G$ of iff  $G$ is an osubformula of $E$. 
$E$ is a {\bf proper} osuperformula (resp. osubformula) of $G$ iff $E$ is an osuperformula (resp. osubformula) of $G$ but not vice versa. 

A {\bf politeral} (``positive oliteral'') is an oliteral that is not in the scope of $\gneg$. In other words, this is a positive occurrence of a literal. For instance, the formula $P\mlc(\gneg Q\mld\gneg Q)$ has three (rather than five) politerals, which are $P$, the first occurrence of $\gneg Q$, and the second occurrence of $\gneg Q$.

The {\bf modal depth} of an oformula $G$ is the number of the proper osuperformulas of $G$  of the form $\st E$ or $\cost E$. 

A {\bf unit} is a pair $(E,\vec{x})$, which we prefer to write as $E[\vec{x}]$, where $E$ is an oformula and $\vec{x}$ is an
array (sequence) of as many infinite bitstrings as the modal depth of    $E$ (if the latter is $0$, then $E[\vec{x}]$ simply looks like $E[\hspace{1pt}]$). Below is an inductive (re)definition of the set of units, together with the functional {\bf parenthood} relation  on units, and the {\bf projection} function that takes a unit, a run $\Omega$ and returns a run:

\begin{description}
  \item[(a)] $\mathbb{F}_0[\hspace{1pt}]$ is a unit, and the projection of $\Omega$ on $\mathbb{F}_0[\hspace{1pt}]$ is $\Omega$ itself. The unit $\mathbb{F}_0[\hspace{1pt}]$ has no parents. 
  \item[(b)] Assume $E[\vec{x}]$ is a unit, and  $\Theta$ is the projection of $\Omega$ on it. Then:
  \begin{enumerate}  
    \item Suppose $E$ is $G_0\mlc G_1$ or $G_0\mld G_1$. Then, for either $i\in\{0,1\}$, $G_i[\vec{x}]$ is a unit, the projection of $\Omega$ on $G_i[\vec{x}]$ is $\Theta^{i.}$ (see Notation 2.2 of \cite{taming1}), and the (only) parent of $G_i[\vec{x}]$ is $E[\vec{x}]$. 
    \item Suppose $E$ is $\st G$ or $\cost G$. Then, for any infinite bitstring $y$, $G[\vec{x},y]$ is a unit,  the projection of $\Omega$ on
$G[\vec{x},y]$ is $\Theta^{\preceq y}$ (see Notation 2.2 of \cite{taming1}), and the (only) parent of $G[\vec{x},y]$ is $E[\vec{x}]$. 
  \end{enumerate}
\item[(c)] Nothing is a unit  unless it can be obtained by repeated application of (a) and (b).
\end{description}

A  {\bf funit} (``f'' for ``finite'') is a pair $(E,\vec{w})$, which we prefer to write as $E[\vec{w}]$, where $E$ is an oformula and $\vec{w}$ is a sequence of as many finite bitstrings as the modal depth of    $E$. Below is an inductive (re)definition of the set of funits, together with the {\bf parenthood} relation on funits, as well as the {\bf address} function that takes a funit and returns a string of bits and periods:

\begin{description}
  \item[(a)] $\mathbb{F}_0[\hspace{1pt}]$ is a funit, and its address (as a string) is empty. This funit has no parents. 
  \item[(b)] Assume $E[\vec{w}]$ is a funit, and its address is $\alpha$. Then:
  \begin{enumerate}  
    \item Suppose $E$ is $G_0\mlc G_1$ or $G_0\mld G_1$. Then, for either $i\in\{0,1\}$, $G_i[\vec{w}]$ is a funit, and its address is $\alpha i.$ ($\alpha$ followed by the bit $i$ followed by a period).  The (only) parent of $G_i[\vec{w}]$ is $E[\vec{w}]$. 
    \item Suppose $E$ is $\st G$ or $\cost G$. Then, for any finite bitstring $u$, $G[\vec{w},u]$ is a funit, and its address is $\alpha u.$ ($\alpha$ followed by the bitstring $u$ followed by a period).  The (only) parent of $G[\vec{w},u]$ is $E[\vec{w}]$.
  \end{enumerate}
\item[(c)] Nothing is a funit  unless it can be obtained by repeated application of (a) and (b).
\end{description}

When $E[\vec{x}]=E[x_1,\ldots,x_n]$ is a unit and $E[\vec{w}]=E[w_1,\ldots,w_n]$ is a funit such that $w_1\preceq x_1,\ldots,w_n\preceq x_n$,\footnote{remember that, as we agreed earlier, $w\preceq x$ means that $w$ is a prefix of $x$.}    we say that $E[\vec{w}]$ is a {\bf funital restriction} of  $E[\vec{x}]$, and that  $E[\vec{x}]$ is a  {\bf unital extension} of $E[\vec{w}]$.

A {\bf politeral unit} (resp. {\bf politeral funit}) is a unit (resp. funit) $L[\vec{x}]$ such that $L$ is a politeral. 

We say that a unit $E[\vec{x}]$ is a {\bf child} of a unit $G[\vec{y}]$ iff $G[\vec{y}]$ is the parent of $E[\vec{x}]$. We say that $E[\vec{x}]$ is a {\bf subunit} of $G[\vec{y}]$, and that $G[\vec{y}]$ is a {\bf superunit} of $E[\vec{x}]$, iff $E[\vec{x}]=G[\vec{y}]$, or $E[\vec{x}]$ is a child of $G[\vec{y}]$, or $E[\vec{x}]$ is a child of a child of $G[\vec{y}]$, or \ldots.  Such a subunit  or superunit is said to be {\bf proper} iff $G[\vec{y}]\not=E[\vec{x}]$. Similarly for funits instead of units. 

We agree that, when $E[\vec{w}]$ is a funit  and $\beta$ is a move,  by ``{\bf making the move $\beta$ in $E[\vec{w}]$}'' we mean making the move $\alpha \beta$, where $\alpha$ is the address of 
$E[\vec{w}]$. When $E[\vec{x}]$ is a unit  and $\beta$ is a move,  by ``{\bf making the move $\beta$ in $E[\vec{x}]$}'' we mean making the move $\beta$ in some funital restriction of $E[\vec{x}]$. If here $\beta$ is a decimal numeral, we say that it is a {\bf numeric move}. Note that numeric moves can be legally made only in politeral (f)units.

Intuitively, a legal play/run $\Omega$ of $\mathbb{F}_0$ comprises parallel plays in all of the politeral units. Namely, the run that is taking place in a politeral unit $L[\vec{x}]$ is nothing but the projection of $\Omega$ on $L[\vec{x}]$. Every move of $\Omega$ has the form $\alpha a$, where $a$ is a decimal numeral, and $\alpha$ is the address of some politeral funit $L[\vec{w}]$. 
 The intuitive and technical effect of such a move  is making the numeric move $a$ in all  politeral units $L[\vec{x}]$ that happen to be unital extensions of $L[\vec{w}]$. 

We define the {\bf height} $h$ of a funit $E[w_1,\ldots,w_n]$ as the length of a longest bitstring among $w_1,\ldots,w_1$ (here, if $n=0$, we let $h=0$). We say that a funit $E[w_1,\ldots,w_n]$ is {\bf regular} iff all of the bitstrings $w_1,\ldots,w_n$ are of the same length. 

Let $\Phi$ be a position. By a {\bf $\Phi$-active funit} we mean a funit $G[\vec{u}]$ such that $\Phi$ contains some move made by either player in $G[\vec{u}]$ (that is, $\Phi$ contains some labmove of the form $\xx\alpha\beta$, where $\xx\in\{\top,\bot\}$ and $\alpha$ is the address of $G[\vec{u}]$).
And by a 
{\bf $\Phi$-prompt} we mean any of the (finitely many) regular politeral funits $L[\vec{w}]$ such that 
\begin{itemize}
  \item either $\vec{w}$ is empty ($L$ is not in the scope of any occurrences of $\st$ and/or $\cost$ within $\mathbb{F}_0$, that is), 
  \item or $L[\vec{w}]$ is of height $h$, where $h$ is the smallest number exceeding the heights of all $\Phi$-active funits.  
 \end{itemize}

\section{The counterstrategy $\cal E$}\label{sscs}
%\marginpar{sscs}

Technically, the counterstrategy  $\cal E$ that we are going to construct in this section is an EPM, meaning that, in a computation branch of $\cal E$, the moves made by $\cal E$ appear (as they should according to our definition of an EPM given in Section 2.5 of \cite{taming1}) on the run tape of $\cal E$ with the label $\pp$, and the moves made by its adversary with the label $\oo$. However, we will be eventually interested in the runs cospelled rather than spelled by computation branches of $\cal E$. Correspondingly, when describing or analyzing the work of $\cal E$, it is beneficial for our intuition to directly talk about the positions ``cospelled'' rather than spelled on the run tape of the machine at various times. Namely, we agree that, in the context of any step of any play (computation branch) of $\cal E$, by the {\bf current position} we mean $\gneg \Psi$, where $\Psi$ is the position spelled on the run tape of $\cal E$. 

We define $\cal E$ as an EPM that creates a variable $i$, initializes its value to $1$, and then keeps performing the following routine forever;
  at 
every step of (our description of)  the latter,  $\Phi$ refers to the {\em then-current} position of the play:

\begin{quote} ROUTINE: Let   
$L_1[\vec{w}_1],\ldots,L_n[\vec{w}_n]$ be the lexicographic list of all $\Phi$-prompts.  
\begin{enumerate}
  \item Do the following while $i\leq n$:
  \begin{enumerate}
    \item Make the numeric move $a$ in $L_i[\vec{w}_i]$, where $a$ is the smallest natural number never made as a numeric move in whatever politeral unit by either player so far in the play (in $\Phi$, that is). Then increment $i$ to $i+1$.\footnote{Note that every iteration of this WHILE loop changes the value of $\Phi$.}
\end{enumerate}
\item Once $i$ becomes $n+1$, reset it to $1$, grant permission, and repeat ROUTINE.
\end{enumerate}
\end{quote}

Note that $\cal E$ is a fair EPM, so that Lemma \ref{lem} applies. This is so because $\cal E$ repeats ROUTINE infinitely many times, and each repetition grants permission. 

Also take a note of the fact that  
%\marginpar{feb23a}
\begin{equation}\label{feb23a}
\mbox{\em $\cal E$ always makes ``fresh'' numeric moves},
\end{equation}
in the sense that it never makes a numeric move (in whatever politeral funit) that has already been made in the play (in whatever politeral funit) by either player.  

We now pick  an arbitrary valuation $e$, an arbitrary HPM $\cal H$ and denote by $\Omega$ the $\cal H$ vs. $\cal E$ run on $e$, i.e., the run cospelled by the $({\cal E},e,{\cal H})$-branch. 
We fix these parameters $e$, $\cal H$, $\Omega$  until further notice (in Section \ref{sss}) and agree that our discourse is always in the context of these particular $e$, $\cal H$ and $\Omega$. So, for instance, when we say ``$\cal E$ made the move $\alpha$ at time $t$'', it is to be understood as that such a move was made by $\cal E$ in the $({\cal E},e,{\cal H})$-branch on the $t$'th clock cycle. It is important to point out once again that, in such contexts, $\cal E$ will be viewed as the ``author'' of $\bot$-labeled moves (in $\Omega$), and its adversary $\cal H$ as the ``author'' of $\top$-labeled moves, because this is exactly how the corresponding moves are labeled in $\Omega$, which is the run cospelled (rather than spelled) by the $({\cal E},e,{\cal H})$-branch that we are considering.   

Notice that $\cal E$ never makes illegal moves. We may safely pretend that neither does its adversary, for otherwise $\cal E$ automatically wins and the case is trivial. So, we adopt the assumption that 
%\marginpar{feb23b}
\begin{equation}\label{feb23b}
\mbox{\em $\Omega$ is a legal run of $\mathbb{F}_0$,}
\end{equation}
which, in view of our conventions of Section \ref{sseg}, precisely means that, for any constant enumeration interpretation $^*$, \ $\Omega$ is a legal run of $\mathbb{F}_{0}^{*}$ --- or, equivalently, that for any (not necessarily constant) enumeration interpretation $^*$, \ $\Omega$ is a legal run of $e[\mathbb{F}_{0}^{*}]$. 

\section{Unit trees, resolutions and driving}\label{ss17n}
%\marginpar{ss17n}

We have been using capital Latin letters as metavariables for formulas, oformulas and cirquents. In what follows, we use the same metavariables for units as well.
When $E$ is used to denote a unit $G[\vec{x}]$, the latter (as an expression) is said to be the {\bf expanded form} of the former.  For a unit  $E$, by the {\bf $\mathbb{F}_{0}$-origin} of $E$, symbolically 
\[\tilde{E},\]
we shall mean the oformula $G$  such that the expanded form of $E$ is $G[\vec{x}]$ for some $\vec{x}$. When $\tilde{E}=G$ and $G$ has the form $\st H$, we may refer to $E$ (resp. $G$) as a {\bf $\st$-unit} (resp. {\bf $\st$-oformula}). Similarly in the case of other connectives instead of $\st$. 

In the sequel we may refer to a unit as $\st E$ to indicate that it is a $\st$-unit. In such a case, the $\mathbb{F}_0$-origin of $\st E$ will be denoted by $\st \tilde{E}$.  Similarly in the case of other connectives instead of $\st$.

By a {\bf unit tree} we mean a nonempty set $S$ of units  such that,  whenever a unit $E$  is in $S$, we have:
\begin{itemize}
  \item all superunits of $E$ are in $S$;
  \item if $E$ is a $\mlc$- or $\mld$-unit, then both of its children are in $S$;
  \item if $E$ is a $\st$- or $\cost$-unit, then at least one of its children is in $S$.
\end{itemize}  
Note that every unit tree is indeed a tree of units formed by the parenthood relation, where the {\bf root} is always the unit $\mathbb{F}_0[\hspace{1pt}]$.  

For unit trees we use the same metavariables as for formulas, oformulas, cirquents or units. The unit tree consisting of all units (the ``biggest'' unit tree) we denote by  
\[\mathbb{F}_1.\]
It may be helpful for one's intuition to think of $\mathbb{F}_1$ as (the parse tree of) a ``formula'' of finite height but infinite --- in fact, uncountable --- width. Such a ``formula'' is similar to $\mathbb{F}_0$, with the only difference that, while every $\st$- or $\cost$-osubformula of $\mathbb{F}_0$ has a single child, a corresponding ``osubformula'' (unit) of $\mathbb{F}_1$ has uncountably many children instead --- one child per  infinite bitstring.

When $F$ is a unit tree and $E$ is an element of $F$ other than the root, by {\bf trimming} $F$ at $E$ we mean deleting from $F$ all subunits of $E$ (including $E$ itself). 

 A {\bf resolution}  is a function $\cal A$ that sends each unit $\st E$  either to the value ${\cal A}(\st E)=\mbox{{\em Unresolved}}$,  or to an infinite
 bitstring ${\cal A}(\st E)=x$. In the former case we say that $\st E$ is {\bf $\cal A$-unresolved}; in the latter case we say that 
$\st E$ is {\bf $\cal A$-resolved} and, where $\st E=\st \tilde{E}[\vec{y}]$, call the unit $\tilde{E}[\vec{y},x]$ the {\bf $\cal A$-resolvent} of $\st E$. 
When $\cal A$ is fixed, we may omit it and simply say ``resolved'', ``resolvent'' etc. 

Intuitively, a resolution $\cal A$ selects a single child --- namely, the resolvent --- for each $\cal A$-resolved $\st$-unit, and does nothing for any other units.  In accordance with this intuition, when $\cal A$ is a resolution, we will be using the expression 
 \[\mathbb{F}_{1}^{\cal A}\] 
to denote the unit tree that is the result of trimming $\mathbb{F}_1$ at every child of every $\cal A$-resolved $\st$-unit except the child which is the $\cal A$-resolvent of that unit. 

The {\bf trivial resolution}, which we denote by 
\(\emptyset,\)
is the one that returns the value {\em Unresolved} for every $\st$-unit.  Thus, $\mathbb{F}_{1}=\mathbb{F}_{1}^{\emptyset}$. 

On the other extreme, a resolution $\cal A$ such that every $\st$-unit of $\mathbb{F}_{1}^{\cal A}$ is $\cal A$-resolved is said to be {\bf total}. 

We say that a resolution ${\cal A}_2$ is an {\bf extension} of a resolution ${\cal A}_1$ iff every ${\cal A}_1$-resolved $\st$-unit is also ${\cal A}_2$-resolved (but not necessarily vice versa), and the  ${\cal A}_2$-resolvent of each such unit is the same as its ${\cal A}_1$-resolvent. 

We say that two resolutions ${\cal A}_1$ and ${\cal A}_2$ are {\bf consistent} iff they agree on every $\st$-unit that is both ${\cal A}_1$-resolved and ${\cal A}_2$-resolved.  Note that if one resolution is an extension of another, then the two resolutions are consistent. 

Let $S$ be a set of pairwise consistent resolutions. Then the {\bf union}  of the resolutions from $S$
is the ``smallest'' common extension $\cal R$ of all elements of $S$. Precisely, for any unit $\st E$, $\st E$ is $\cal R$-resolved 
 iff it is ${\cal A}$-resolved for some ${\cal A}\in S$, in which case ${\cal R}(\st E)={\cal A}(\st E)$.    Note that if the set $S$ is empty, such a  resolution $\cal R$ is trivial.

The {\bf smallest common superunit} of units $E$ and $G$ is the unit $H$ such that:
\begin{enumerate}
  \item $H$ is a  superunit of both $E$ and $G$, i.e., is a {\bf common superunit} of $E$ and $G$,  and
  \item $H$ is a subunit of every common superunit of $E$ and $G$. 
\end{enumerate}  

Note that any two units have a unique smallest common superunit. 

\begin{definition}\label{dec10a}
%\marginpar{dec10a}
Below $E,G,H$ are arbitrary units,  and $\cal A$ is an arbitrary resolution. 
\begin{enumerate}
  \item We say that $E$ {\bf drives $G$ through $H$}  iff  the following two conditions are satisfied: 
  \begin{itemize} 
    \item $H$ is the smallest common superunit of $E$ and $G$.
    \item $E$ has no proper $\cost$-superunits that happen to be subunits of $H$. 
  \end{itemize}
\item We say that $E$ {\bf $\cal A$-strictly drives $G$ through $H$}  iff $E$ drives $G$ through $H$ and, in addition, the following  condition is satisfied for any unit $\st F$:
\begin{itemize}
  \item  If $\st F$ is a proper superunit of $E$ that happens to be a subunit of $H$, then $\st F$ is $\cal A$-resolved and $E$ is a subunit of its  $\cal A$-resolvent.
\end{itemize} 
  \item When we simply say ``$E$ drives  $G$'',
 it is to be understood as ``$E$ drives $G$ through $H$ for some $H$''. Similarly for ``$E$  $\cal A$-strictly drives $G$''.
\end{enumerate}
\end{definition}

\begin{lemma}\label{jan3a}
%\marginpar{jan3a}
For every resolution $\cal A$ and units $E$ and $G$,  if $G$ is a subunit of $E$, then $E$ $\cal A$-strictly drives $G$ through $E$. 
\end{lemma}

\begin{proof} Consider any resolution $\cal A$ and units $E,G$. If $G$ is a subunit of $E$, then $E$ is the smallest common superunit of $E$ and $G$. Of course, $E$ has no proper superunits that happen to be subunits of $E$. So, $E$ $\cal A$-strictly drives $G$ through $E$.  \end{proof}

\begin{lemma}\label{jan2a} \ 
%\marginpar{jan2a}

1.
The relation ``drives''  is transitive. That is, for any units $G_1,G_2,G_3$, if $G_1$ drives $G_2$  and $G_2$ drives $G_3$, then $G_1$ drives  $G_3$. 

2. Similarly for ``$\cal A$-strictly drives'' (for whatever resolution $\cal A$).
\end{lemma}

\begin{proof} Here we will only prove clause 1. Clause 2 can be verified in a rather similar way, and we leave this job (if necessary) to the reader. 

Assume $G_1,G_2,G_3$ are units such that $G_1$ drives $G_2$ and $G_2$ drives $G_3$. 
Let $H_{12}$ be the smallest common superunit of $G_1$ and $G_2$, and $H_{23}$ be the smallest common superunit of $G_2$ and $G_3$. Thus, $H_{12}$ and $H_{23}$ have a common subunit (namely, $G_2$). Therefore, since $\mathbb{F}_1$ is a tree,  $H_{12}$ is either a superunit or a subunit of $H_{23}$. 

If $H_{12}$ is a superunit of $H_{23}$, then $H_{12}$ is clearly the smallest common superunit of $G_1$ and $G_3$. We also know that, among the subunits of  $H_{12}$, $G_1$ has no proper $\cost$-superunits. Thus, $G_1$ drives $G_3$ through $H_{12}$.

If $H_{12}$ is a subunit of $H_{23}$, then  $H_{23}$ is clearly the smallest common superunit of $G_1$ and $G_3$. Note that, among the subunits of  
$H_{23}$, $H_{12}$ has no  proper $\cost$-superunits, because then so would have $G_2$, meaning that $G_2$ does not drive $G_3$. At the same time, among the subunits of  $H_{12}$, $G_1$ has no proper $\cost$-superunits. Thus, $G_1$ has no proper $\cost$-superunits that happen to be subunits of $H_{23}$, meaning that $G_1$ drives  $G_3$ through $H_{23}$.
\end{proof}

\begin{lemma}\label{nov8d}
%\marginpar{nov8d}
For every resolution $\cal A$ and unit $G$, the set of the units that $\cal A$-strictly drive $G$ is finite.  Furthermore, all units that $\cal A$-strictly drive $G$ have different $\mathbb{F}_0$-origins. 
\end{lemma}

\begin{proof} Consider an arbitrary resolution $\cal A$ and an arbitrary unit $G$. 
Seeing $\mathbb{F}_{1}$ as a downward-growing tree of units induced by the parenthood relation, every unit $E$ that $\cal A$-strictly drives $G$ can be reached from $G$ by first going up the tree to the smallest common superunit $H$ of $E$ and $G$, and then descending to $E$, where $\cost$-units or $\cal A$-unresolved $\st$-units may occur (may be passed through) only during the upward journey, and where, whenever an $\cal A$-resolved $\st$-unit is passed through on the downward journey, the descend from there always happens to the $\cal A$-resolvent of that unit.  There are only finitely many upward journeys from $G$ (as many as the number of all superunits of $G$), and from each upward journey there are only finitely many downward journeys that satisfy the above conditions. 
 Furthermore, to every upward-downward journey (sequence of units) $\vec{U}$ of the above kind obviously corresponds a unique upward-downward journey in the parse tree of $\mathbb{F}_0$ --- namely, the result of replacing every unit of $\vec{U}$ by its $\mathbb{F}_0$-origin. With a moment's thought, this can be seen to imply that the $\mathbb{F}_0$-origin of the last unit of $\vec{U}$ is unique. 
\end{proof}

\begin{lemma}\label{may1b}
%\marginpar{may1b}
For any units $E,G$ and resolutions ${\cal A},{\cal A}'$, if $E$ $\cal A$-strictly drives $G$ and ${\cal A}'$ is an extension of $\cal A$, then $E$ also ${\cal A}'$-strictly drives $G$. 
\end{lemma}

\begin{proof} Rather immediate from the relevant definitions. \end{proof}

\section{Visibility}

We say that two politeral units $L$ and $M$ are {\bf opposite} to each other iff $\tilde{L}=\gneg\tilde{M}$ (which is the same as to say that $\tilde{M}=\gneg \tilde{L}$) and, for both $\xx\in\{\top,\bot\}$, the set of all $\xx$-labeled moves found in the projection of $\Omega$ (the run fixed in Section \ref{sscs}) on $L$ coincides with the set of all $\gneg \xx$-labeled moves found in the projection of $\Omega$ on $M$. 

\begin{lemma}\label{feb24a}
%\marginpar{feb24a}
Every politeral unit has at most one opposite politeral unit.
\end{lemma}

\begin{proof} As observed in (\ref{feb23a}), $\cal E$ always makes fresh numeric moves. And it does so in regular funits of ever increasing heights.\footnote{For the exception of funits of the form $E[\hspace{1pt}]$, of course, whose height is always $0$.} Therefore its is obvious that, if $M_1$ and $M_2$ are two different politeral units (and hence have no common regular funital restrictions of heights greater than a certain bound $h$), there is a numeric move (in fact, infinitely many numeric moves) made by $\cal E$ in $M_1$ but not in $M_2$ (and vice versa, of course). That is, the set of the $\oo$-labeled moves found in the projection of $\Omega$ on $M_1$ is different from the set of the $\oo$-labeled moves found in the projection of $\Omega$ on $M_2$. For this reason, $M_1$ and $M_2$ cannot be simultaneously opposite to the same unit. To summarize, no politeral unit $L$ can have two different opposite (to $L$) politeral units $M_1$ and $M_2$.
\end{proof}

\begin{definition}\label{dec10b}
%\marginpar{dec10b}
Let $\cal A$ be a resolution.  
\begin{enumerate}
\item A {\bf  visibility chain} in $\cal A$ is a nonempty sequence 
%\marginpar{jan21a}
\begin{equation}\label{jan21a}
L_1,\ M_1,\ \ \ldots,\ L_{n},\ M_{n}
\end{equation}
  of politeral units such that, for every $i$ with $1\leq i\leq n$, the following holds:
\begin{enumerate} 
  \item $L_{i}$ and $M_{i}$ are opposite to each other.
  \item If $i<n$, then $M_{i}$ \ $\cal A$-strictly drives $L_{i+1}$.
\end{enumerate}
\item  $L_1$ is said to be the {\bf head} of chain (\ref{jan21a}), and $M_n$ is sad to be the {\bf tail}. 
\item The {\bf type} of chain (\ref{jan21a}) is the sequence 
\[\tilde{L}_1,\ \tilde{M}_1,\ \ldots,\ \tilde{L}_{n}, \ \tilde{M}_{n}\]
 resulting from (\ref{jan21a}) by replacing all units with their $\mathbb{F}_0$-origins.  

\end{enumerate}
\end{definition}

In what follows, by a {\bf visibility chaintype} we mean any nonempty, even-length finite sequence  of politerals of $\mathbb{F}_0$. 
Thus, the type of every visibility chain is a visibility chaintype. Not every visibility chaintype can be the type of an actual visibility chain though, and the reason for introducing the term ``visibility chaintype'' is merely technical.   
 
\begin{lemma}\label{jan26b}
%\marginpar{jan26c}
There are only countably many visibility chaintypes.
\end{lemma}

\begin{proof} This is so because $\mathbb{F}_0$ only has finitely many politerals, and every visibility chaintype is a finite sequence of such politerals. 
\end{proof}

\begin{lemma}\label{nov8cc}
%\marginpar{nov8cc}
For every visibility chaintype $\vec{T}$, resolution $\cal A$ and politeral unit $J$, there is at most one $J$-tailed visibility chain in $\cal A$ of type $\vec{T}$.  
\end{lemma}

\begin{proof} Consider any resolution $\cal A$, politeral unit $J$  and visibility chaintype  
\[\vec{T} \ =\ X_1, \ Y_1,\ \ldots, X_n,\ Y_n.\]  
 Assume the following are two $J$-tailed visibility chains in $\cal A$  of type $\vec{T}$: 
\[\vec{N}\ =\ L_1,\ M_1,\ \ldots, L_n,\ M_n;\]
\[\vec{N}'\ =\ L'_1,\ M'_1,\ \ldots, L'_n,\ M'_n.\]
We want to show that $\vec{N}=\vec{N}'$, that is, show that, for every $i\in\{1,\ldots,n\}$, $M_i=M'_i$ and $L_i=L'_i$. 

Since both visibility chains are $J$-tailed, we have $M_n=M'_n=J$. This, in view of Lemma \ref{feb24a}, immediately implies that we also have $L_n=L'_n$, because $L_n$ (resp. $L'_n$) and $M_n$ (resp. $M'_n$)  are opposite.   

Now consider the case $i=n-1$ (the cases $i=n-2$, $i=n-3$, \ldots will be handled similarly). Both $M_{n-1}$ and $M'_{n-1}$ $0$-drive $L_n=L'_n$ in $\cal A$, and we  also have $\tilde{M}_{n-1}=\tilde{M}'_{n-1}=Y_{n-1}$. Therefore, by the ``Furthermore'' clause of Lemma \ref{nov8d},  $M_{n-1}=M'_{n-1}$. This, in turn, as in the previous case, immediately implies that  $L_{n-1}=L'_{n-1}$. 
\end{proof}

\begin{lemma}\label{nov8c}
%\marginpar{nov8c}
For every resolution $\cal A$ and politeral unit $J$, the following set is countable:
\[\mbox{\{$I$ \ $|$ \ $I$ is a politeral unit such that there is a visibility chain in $\cal A$ with head $I$ and tail $J$\}}.\] 
\end{lemma}

\begin{proof} Immediately from Lemmas \ref{jan26b} and \ref{nov8cc}.
\end{proof}

\begin{definition}\label{dec10c}
%\marginpar{dec10c}
Let $\cal A$ be a resolution, and $E$ and  $G$ any units. We say that $E$ is {\bf visible} to $G$ in $\cal A$ iff  either $E$ $\cal A$-strictly drives $G$  or there is a visibility chain $\vec{N}$ in $\cal A$ such that $E$ $\cal A$-strictly drives the head of $\vec{N}$, and the tail of $\vec{N}$ $\cal A$-strictly drives $G$.
\end{definition}

\begin{lemma}\label{nov8cf}
%\marginpar{nov8cf}
For every resolution $\cal A$ and unit $G$, the set of the units that are visible to $G$  in $\cal A$ is countable. 
\end{lemma}

\begin{proof} Immediately from Lemmas \ref{nov8d} and \ref{nov8c}. 
\end{proof}

\begin{lemma}\label{jan6cc}
%\marginpar{jan6cc}
Assume $\cal A$ is a resolution, $E$ is a unit and $G$ is a subunit of $E$. Then:  

1. $E$ is visible to $G$ in $\cal A$.

2. For any unit $H$, if $G$ is visible to $H$ in $\cal A$, then so is $E$. 
\end{lemma}

\begin{proof} Assume the conditions of the lemma.  Below, ``visible'' means ``visible in $\cal A$''. Similarly for  ``visibility chain''. 

Clause 1 is immediate from Lemma \ref{jan3a}. 
For clause 2, assume $H$ is a unit  such that $G$ is visible to $H$. If the reason for this visibility is that $G$ $\cal A$-strictly drives $H$, then the same reason makes $E$ also visible to $H$, because, by Lemma \ref{jan3a} and clause 2 of Lemma \ref{jan2a}, $E$ $\cal A$-strictly drives $H$. Now assume the reason for $G$ being visible to $H$ is that there is a visibility chain $\vec{N}$ such that $G$ $\cal A$-strictly drives the head of $\vec{N}$, and the tail of $\vec{N}$ $\cal A$-strictly drives $H$. Then, again by Lemma \ref{jan3a} and clause 2 of Lemma \ref{jan2a}, $E$ $\cal A$-strictly drives the head of $\vec{N}$. Hence $E$ is visible to $H$.  
\end{proof}

\begin{lemma}\label{may1a}
%\marginpar{may1a}
If a unit $E$ is visible to a unit $G$ in a resolution $\cal A$, then $E$ remains visible to $G$ in all extensions of $\cal A$ as well. 
\end{lemma}

\begin{proof} Immediately from Lemma \ref{may1b}. \end{proof}

\section{Domination}

\begin{definition}\label{dec10d}
%\marginpar{dec10d}
Let $F$ be a unit tree, and $\st E,G\in F$.  
A {\bf $\st E$-over-$G$ domination chain} in $F$ is a nonempty sequence
\[
 L_{1},\ M_{1},\ X_1,\ \ldots,\ L_{n},\ M_{n},\ X_n\]
of units where, with $i$ ranging over $\{1,\ldots,n\}$ and $G$ renamed into $L_{n+1}$ for convenience, the following conditions are satisfied: 
\begin{enumerate}
  \item $L_{i}$ and $M_i$ are opposite politeral units of $F$.
  \item $M_i$ drives $L_{i+1}$ through $X_i$. 
  \item $M_{i}$ does not drive any $L_j$ with $i+2\leq j\leq n+1$.
  \item $M_{i}$  is not a subunit of $\st E$.
  \item $L_{1}$ is a subunit of $\st E$.
\end{enumerate}  
The {\bf type} of such a domination chain is the sequence 
\[
 \tilde{L}_{1},\ \tilde{M}_{1},\ \tilde{X}_1,\ \ldots,\ \tilde{L}_{n},\ \tilde{M}_{n},\ \tilde{X}_n.\]

\end{definition}

When we simply say ``a domination chain'', it is to be understood as ``a $\st E$-over-$G$ domination chain in $F$ for some (whatever) 
unit tree $F$ and units $\st E,G\in F$''. 

By a {\bf domination chaintype} we shall mean any sequence  of oformulas of $\mathbb{F}_0$ of length $3n$ for some positive integer $n$. 

\begin{definition}\label{may1c}
%\marginpar{may1c}
Where $F$ is a unit tree and $\st E$ and $G$ are units,  we say that $\st E$  {\bf dominates}  $G$ in $F$ iff $\st E,G\in F$ and either $G$ is a proper subunit of $\st E$ or there is a $\st E$-over-$G$ domination chain  in $F$.
\end{definition}

\begin{lemma}\label{dec29a}
%\marginpar{dec29a}
For any unit tree $F$ and  units $\st E$, $G$, $H$ of $F$, if $\st E$ dominates $G$ in $F$ and $H$ is a subunit of $G$, then $\st E$ also dominates $H$ in $F$.  
\end{lemma}

\begin{proof} Assume the conditions of the lemma. Below ``dominates'' should be understood as ``dominates in $F$''. Similarly for ``domination chain''.
  
If the reason for $\st E$ dominating  $G$ is that $G$ is a proper subunit of $\st E$, then the same reason makes $\st E$ dominate $H$. Now assume  the reason for $\st E$ dominating  $G$ is that there is  a $\st E$-over-$G$ domination chain 
%\marginpar{jan8a}
\begin{equation}\label{jan8a}
L_{1},\ M_{1},\ X_1,\ \ldots,\ L_{n},\ M_{n},\ X_n.
\end{equation} 
So, $M_{n}$ drives $G$. In view of Lemma \ref{jan3a}, we also know that $G$ drives $H$.  Hence, by clause 1 of Lemma \ref{jan2a},  $M_{n}$   drives $H$. Thus, at least one of the units $M_1,\ldots,M_n$ drives $H$. Let $i$ be the smallest integer among $1,\ldots,n$ such  that $M_i$ drives $H$, and let $Y$ be the smallest common superunit of $M_i$ and $H$. Then, obviously, the sequence 
\[L_{1},\ M_{1},\ X_1,\ \ldots,\ L_{i-1},\ M_{i-1},\ X_{i-1}, \ L_{i},\ M_{i}, \ Y\]
is a 
 $\st E$-over-$H$ domination chain, so that $\st E$ dominates $H$.
\end{proof}

\begin{lemma}\label{nov10c}
%\marginpar{nov10c}
Any two opposite politeral units of any unit tree $F$ are dominated (in $F$) by exactly the same $\st$-units of $F$. 
\end{lemma}

\begin{proof} Consider any unit tree $F$, any two opposite politeral units $L$ and $M$ of $F$, any unit $\st E$ of $F$ and assume that $\st E$ dominates $L$, with ``dominates'' here and below meaning ``dominates in $F$'', and similarly for ``domination chain''. Our goal to show that $\st E$ also dominates $M$.

If $M$ is a subunit of $\st E$, then it is so in the proper sense, so $\st E$ dominates $M$. Otherwise, if $L$ is a subunit of $\st E$, then obviously the three-element sequence $L,M,M$ is a $\st E$-over-$M$ domination chain and, again, $\st E$ dominates $M$.  Now, for the rest of this proof, assume neither $M$ nor $L$ is a subunit of $\st E$.

Let  \[ L_{1},\ M_{1},\ X_1, \  \ldots, \ L_{n},\ M_{n},\ X_n\] be a $\st E$-over-$L$ domination chain. Thus, $M_{n}$  drives $L$. If one of the 
units $M_1,\ldots,M_n$ drives $M$, then, where $M_i$ is the leftmost of such units and $Y$ is the smallest common superunit of $M_i$ and $M$, the following sequence is obviously a $\st E$-over-$M$ domination chain:
\[L_{1},\ M_{1}, \ X_1,\  \ldots, \ L_{i-1},\ M_{i-1}, \ X_{i-1},  \ L_{i},\ M_{i},\ Y.\]
And if none of the 
units $M_1,\ldots,M_n$ drives $M$, then
the sequence 
\[L_{1},\ M_{1},\ X_1,\ \ldots,\ L_{n},\ M_{n},\ X_n,\ L, \ M,\ M\]
is a $\st E$-over-$M$ domination chain. In either case we thus have that  $\st E$ dominates $M$, as desired.  
\end{proof}

\begin{lemma}\label{apr6a}
%\marginpar{apr6a}
Assume $F$ is a unit tree, and $\st E,G,G'$ are units of $F$ such that $\st E$ dominates $G$ in $F$ and $G'$ is the parent of $G$.  
Then we have:

(a) If $G'$ is a $\mlc$- or $\mld$-unit, then $\st E$ dominates $G'$ in $F$.

(b) If $G'$ is a $\st$-unit,  then either $\st E=G'$ or $\st E$ dominates $G'$ in $F$.
\end{lemma}

\begin{proof} Assume the conditions of the lemma. In what follows, ``dominates'' and ``domination chain'' should be understood in the context of 
$F$. 

{\em Clause (a)}: Assume $G'$ is a $\mlc$- or $\mld$-unit.
If the reason for  $\st E$ dominating $G$ is that $G$ is a proper subunit of $\st E$, then obviously $G'$ is also a proper subunit of $\st E$ and is thus also dominated by the latter.
Suppose now the reason for $\st E$ dominating $G$ is that  there is a $\st E$-over-$G$ domination chain 
\[L_1,\ M_1,\ X_1,\ \ldots,\ L_n,\ M_n,\ X_n.\]
Let $Y$ be $X_n$ if the latter is a superunit of $G'$, and be $G'$ otherwise. In other words, $Y$ is the smallest common superunit of $X_n$ and $G'$. As such, $Y$ is also the smallest common superunit of $M_n$ and $G'$. Further, obviously $M_n$ has no $\cost$-superunit that happens to be a subunit of $Y$. Thus, $M_n$ drives $G'$. Further, no $M_i$ with $1\leq i<n$ drives $G'$, for otherwise, in view of Lemma \ref{jan3a}, $M_i$ would also drive $G$. To summarize, $\st E$ dominates $G'$ because the following is a $\st E$-over-$G'$ domination chain: 
\[L_1,\ M_1,\ X_1,\ \ldots,\ L_n,\ M_n,\ Y.\]

{\em Clause (b)}: Assume $G'$ is a $\st$-unit.  
If the reason for $\st E$ dominating $G$ is that the latter is a proper subunit of the former, then it is clear that either $\st E=G'$, or $G'$ is a proper subunit of $\st E$ and hence is dominated by $\st E$. 
And if the reason for $\st E$ dominating $G$ is that there is a $\st E$-over-$G$ domination chain, we use an argument similar to the one employed in the proof of clause (a) and, again, find that $\st E$ dominates $G'$.  
\end{proof}

\section{Main lemma}
 
\begin{lemma}\label{nov5}
%\marginpar{nov5}
There is a total resolution 
$\cal T$ such that, with 
``dominates'' meaning ``dominates in   $\mathbb{F}_{1}^{\cal T}$'', the following conditions are satisfied:

(i) No $\st$-unit of $\mathbb{F}_{1}^{\cal T}$ dominates the root of $\mathbb{F}_{1}^{\cal T}$.\footnote{Remember that the root of $\mathbb{F}_{1}^{\cal T}$, just as the root of any other unit tree, is nothing but the unit $\mathbb{F}_0[\hspace{1pt}]$.}

(ii) The relation of domination on $\mathbb{F}_{1}^{\cal T}$ is asymmetric. That is, no two (not necessarily distinct)  $\st$-units of $\mathbb{F}_{1}^{\cal T}$ dominate  each other.  

(iii) The relation of domination on $\mathbb{F}_{1}^{\cal T}$ is transitive. That is, for any units $\st E$, $\st G$, $H$ of $\mathbb{F}_{1}^{\cal T}$, if $\st E$ dominates $\st G$ and $\st G$ dominates $H$, then $\st E$ dominates $H$. 
\end{lemma}

The rest of this fairly long section is exclusively devoted to a proof of the above lemma. Throughout that proof, we assume that $\mathbb{F}_0$ is not $\st$-free, for otherwise the lemma holds trivially --- there are no $\st$-units to dominate anything. 

\subsection{Pruning and pillaring}\label{sstod}
%\marginpar{sstod}

We agree that a {\bf pruner} means a set $S$ of finite bitstrings containing the empty string $\epsilon$ and satisfying the condition that, whenever $u$ is a  prefix of some string of $S$, there is an extension $v$ of $u$ (i.e. $u\preceq v$) such that both $v0$ and $v1$ are in $S$. The shortest of such extensions $v$ of $u$ (possibly $v=u$) is denoted by 
\[[u]^{S},\] which can simply be written as $[u]$ when $S$ is fixed in a context.

Next, when $S$ is a pruner, the {\bf limit set} of $S$, denoted by 
\[\limitset(S),\]
is the set of all infinite bitstrings $x$ such that there are infinitely many elements $w$ of $S$ with $w\preceq x$. Thus,  where $w_1,w_2,w_3,\ldots$ are all such elements of $S$ listed according to their lengths, we have \[w_1\preceq w_2\preceq w_3\preceq\ldots\preceq x.\] This explains the word ``limit'' in the name of $\limitset(S)$: every element of $\limitset(S)$ is the limit of an above kind of a series of elements of $S$.

It is obvious that, for every pruner $S$, $\limitset(S)$  has a continuum of elements: there is a straightforward one-to-one mapping from the set $2^\omega$ of all infinite bitstrings to  $\limitset(S)$. Namely, to any element $b_1b_2b_3\ldots$ (each $b_i$ is either $0$ or $1$) of $2^\omega$ corresponds the element of $\limitset(S)$ whose initial segments are \[\epsilon, \ \ [\epsilon]^Sb_1, \ \ [[\epsilon]^Sb_1]^Sb_2,\ \  [[[\epsilon]^Sb_1]^Sb_2]^Sb_3,\ \  \ldots\] 
(remember that $\epsilon$ is the empty bitstring).

Now, we define a  {\bf pruning}   as a function  ${\cal B}$  that sends every unit $\st E$  to some pruner, which (the pruner) is denoted by ${\cal B}(\st E)$. 

When $\cal A$ is a resolution and $\cal B$ is a pruning,  the expression \[\mathbb{F}_{1}^{\cal AB}\] 
 denotes the unit tree that is the result of trimming $\mathbb{F}_{1}^{\cal A}$ at every child $E[\vec{x},y]$  of every $\cal A$-unresolved unit $\st E[\vec{x}]$ such that $y\not\in \limitset({\cal B}(\st E[\vec{x}]))$. 

The {\bf trivial pruning}, denoted (just like the trivial resolution) by 
\(\emptyset,\)   
is the one that,  for each $\st$-unit, returns the pruner consisting of all finite bitstrings. Note that the limit set of such a pruner 
 is the set of all infinite bitstrings. Hence it is clear that $\mathbb{F}_{1}^{{\cal A}\emptyset}=\mathbb{F}_{1}^{\cal A}$. Thus, $\mathbb{F}_{1}=\mathbb{F}_{1}^{\emptyset\emptyset}$.

We say that two (finite or infinite) bitstrings $x$ and $y$ are {\bf comparable} iff $x\preceq y$ or $y\preceq x$; otherwise they are {\bf incomparable}.   

 A {\bf pillaring}  is a function  ${\cal C}$  that sends every 
unit $\st E$  to a nonempty set ${\cal C}(\st E)=\{w_1,\ldots,w_r\}$, where   
$w_1,\ldots,w_r$ are pairwise incomparable finite bitstrings, called the {\bf $\cal C$-pillars of} $\st E$. We additionally require that the following two conditions be satisfied:
\begin{itemize}
  \item There is a number $r$ such that no  $\st$-unit has more than $r$ $\cal C$-pillars; the smallest of such numbers $r$ is said to be the {\bf width} of the pillaring $\cal C$.
  \item  There is a number $s$ such that no $\cal C$-pillar of any $\st$-unit is longer than $s$; the smallest of such numbers $s$ is said to be the {\bf height} of the pillaring $\cal C$.
\end{itemize}

 When $\cal A$ is a resolution, $\cal B$ is a pruning 
 and $\cal C$ is a pillaring, the expression
\[\mathbb{F}_{1}^{\cal ABC}\] denotes the unit tree that is the result of trimming $\mathbb{F}_{1}^{\cal AB}$ at every child $E[\vec{x},y]$  of every $\cal A$-unresolved unit $\st E[\vec{x}]$ such that $y$ is not comparable with any of the $\cal C$-pillars of $\st E[\vec{x}]$.

We say that a pillaring $\cal C$ is {\bf concordant} with a pruning $\cal B$ iff, for every unit $\st E$, every $\cal C$-pillar of $\st E$ is an element of the pruner ${\cal B}(\st E)$.

The {\bf trivial pillaring}, denoted (just like the trivial resolution and the trivial pruning) by 
\(\emptyset,\)   
is the one that returns $\{\epsilon\}$  for each $\st$-unit. Note that, for any resolution $\cal A$ and pruning $\cal B$, $\mathbb{F}_{1}^{{\cal AB}\emptyset}=\mathbb{F}_{1}^{\cal AB}$. Thus, $\mathbb{F}_{1}=\mathbb{F}_{1}^{\emptyset\emptyset\emptyset}$.

A {\bf hightening} of a pillaring $\cal C$ is a pillaring ${\cal C}'$ such that, for any unit $\st E$, 
 where ${\cal C}(\st E)=\{w_1,\ldots,w_r\}$, we have  ${\cal C}'(\st E)=\{w'_1,\ldots,w'_r\}$, where, for each $i\in\{1,\ldots,r\}$, $w_i\preceq w'_i$. 
Thus, ${\cal C}'$ only differs from $\cal C$ in that it makes some (maybe all, maybe none) of the pillars ``taller''.  

\subsection{The $\uparrow$, $\downarrow$ and $\hat{\in}$ notations; unit incomparability; maturity}
Let $a$ be a natural number, and $E$ a unit (resp. funit). By ``{\bf numerically making the move $a$ in $E$}'' we shall mean making the numeric move $a$ in 
some politeral subunit (resp. politeral subfunit) of $E$. Of course, if $E$ is a politeral (f)unit, numerically making the move $a$ in $E$ simply 
means making the move $a$ in $E$. Otherwise, numerically making the move $a$ in $E$ means making a move $\beta$ in $E$ such that the string ``$.a$'' is a suffix of $\beta$.\footnote{Strictly speaking, this is so only provided that the move we are talking about is legal. However, it should be remembered that, as we agreed  in Section \ref{sscs}, our discourse is always $\Omega$, which, by (\ref{feb23b}), is a legal run.}

Let $a$ be a natural number, $\st E$ a unit, and $w$ a finite bitstring.  By ``{\bf numerically making the move $a$ in branch $w\hspace{-2pt}\uparrow$ of $\st E$}'' we shall mean making a move $\beta$ in $\st E$  such that the string ``$w.$'' is a prefix of $\beta$ and the string ``$.a$'' is a suffix of $\beta$.  The same terminology, but without ``$\uparrow$'', extends from finite bitstrings $w$ to infinite bitstrings $x$. Namely, where $x$ is an infinite bitstring,  by ``{\bf numerically making the move $a$ in branch $x$ of $\st E$}'' we mean numerically making the move $a$ in branch $w\hspace{-2pt}\uparrow$ of $\st E$ for some finite prefix $w$ of $x$. Observe that  numerically making the move $a$ in branch $w\hspace{-2pt}\uparrow$ of $\st E$ signifies numerically making the move $a$ in every branch $x$ of $\st E$ with $w\preceq x$. 
Also note that, where $\st \tilde{E}[\vec{y}]$ is the extended form of $\st E$,  numerically making the move $a$ in branch $x$ of $\st E$ means nothing but numerically making the move $a$ --- in the sense of the preceding paragraph --- in the child $\tilde{E}[\vec{y},x]$ of $\st E$.  

For a natural number $m$, we shall write \[\Omega_m\] for the initial segment of $\Omega$ (the run fixed in Section \ref{sscs}) that consists of the moves made at times (clock cycles) not exceeding $m$.
Next, where $E$ is a unit, by \[E\hspace{-2pt}\downarrow\hspace{-2pt} m\] we denote the regular funital restriction of $E$ of height $h$, where $h$ be the greatest of the heights of the $\Omega_m$-active funits (see Section \ref{ssufp}). An exception here is the case when the modal depth of $\tilde{E}$ is $0$ (and hence all funital restrictions of $E$ coincide with $E$ and are of height $0$); in this case,  \(E\hspace{-2pt}\downarrow\hspace{-2pt} m\) simply means $E$.

We say that two units $E$ and $G$ are {\bf incomparable} iff $E$ is neither a subunit nor a superunit of $G$.
Where $m$ is a natural number, we say that $E$ and $G$ are {\bf $\downarrow\hspace{-2pt} m$-incomparable} iff $E\hspace{-2pt}\downarrow\hspace{-2pt}m$ is neither a subfunit nor a superfunit of $G\hspace{-2pt}\downarrow\hspace{-2pt}m$. Note that this is the same as to say that the address of $E\hspace{-2pt}\downarrow\hspace{-2pt}m$ is neither a prefix  nor an extension of the address of $G\hspace{-2pt}\downarrow\hspace{-2pt}m$.

\begin{lemma}\label{mar16a} Below $m$ and $m'$ range over natural numbers.
%\marginpar{mar16a}

1. Whenever two units are $\downarrow\hspace{-2pt} m$-incomparable, they are also $\downarrow\hspace{-2pt} m'$-incomparable for every $m'\geq m$. 

2. For any incomparable units $E$ and $G$ there is an $m$ such that $E$ and $G$ are $\downarrow\hspace{-2pt} m$-incomparable. 
\end{lemma}

\begin{proof} Clause 1 is a straightforward observation. 
For clause 2, consider any incomparable units $E,G$. Let $H$ be the smallest common superunit of $E$ and $G$. 

First, assume $H$ has the form $H_1\mlc H_2$ or $H_1\mld H_2$. 
We may assume that $E$ is a subunit of $H_1$. Then $G$ is a  subunit of $H_2$, for otherwise the smallest common superunit of $E$ and $G$ would be $H_1$ --- or some proper subunit of the latter --- rather than $H$. Let $\alpha$ be the address of $H\hspace{-2pt}\downarrow\hspace{-2pt} 0$. Note that then $\alpha 0.$ is a prefix of the address of $E\hspace{-2pt}\downarrow\hspace{-2pt} 0$, and $\alpha 1.$ is a prefix of the address of $G\hspace{-2pt}\downarrow\hspace{-2pt} 0$. So, the address of $E\hspace{-2pt}\downarrow\hspace{-2pt} 0$ is neither a prefix nor an extension of the address of $G\hspace{-2pt}\downarrow\hspace{-2pt} 0$. In other words, $E$ and $G$ are $\downarrow\hspace{-2pt} 0$-incomparable.

Now assume $H$ has the form $\st H_1$ or $\cost H_1$. Note that then, in its extended form, $E$ looks like $\tilde{E}[\vec{x},y_1,\vec{z}_1]$ and $G$ looks like $\tilde{E}[\vec{x},y_2,\vec{z}_2]$, where $\vec{x}$ is a sequence of as many infinite bitstrings as the modal depth of $\tilde{H}$, and where $y_1\not=y_2$. Let $h$ be the length of the shortest bitstring that is not a common prefix of $y_1$ and $y_2$. Since $\cal E$ keeps making moves in prompts of increasing heights, by a certain time $m$, it will have made a move in a (regular) funit of height $\geq h$.  This clearly makes  $E$ and $G$\  $\downarrow \hspace{-2pt}m$-incomparable.
\end{proof}

Let $m$ be a natural number. We say that a domination chain 
\[L_1,\ M_1,\ X_1,\ \ldots,\ L_n,\ M_n,\ X_n\]
is {\bf mature} at time $m$  iff, whenever $i,j\in\{1,\ldots,n\}$, $H_i$ (resp. $H_j$) is a superunit of $M_i$ (resp. $M_j$) and $H_i$ is incomparable with $H_j$, we have that $H_i$ is also $\downarrow\hspace{-2pt} m$-incomparable with $H_j$.

\begin{lemma}\label{mar18c}
%\marginpar{mar18c}
For every domination chain $\vec{D}$ there is an integer $m$ such that, for any $m'\geq m$, $\vec{D}$ is mature at time $m'$. 
\end{lemma}

\begin{proof} Immediately from Lemma \ref{mar16a} and the obvious fact that any unit has only finitely many superunits. \end{proof} 

For a unit $E$ and natural numbers $a$ and $m$, we write 
\[\mbox{\em $\bot a\hat{\in}_m E$ (resp. $\top a\hat{\in}_m E$)} \]
to mean that $a$ has been numerically made in $E$ by $\cal E$ (resp. $\cal E$'s adversary) at time $m$. When we omit the subscript $m$ and simply write $\bot a\hat{\in} E$, we mean ``$\bot a\hat{\in}_m E$ for some $m$''. Similarly for $\top a\hat{\in} E$. 
As expected, $\hat{\not\in}$ means ``not $\hat{\in}$''. 

\begin{lemma}\label{mar16b}
%\marginpar{mar16b}
Assume  
$F$ is a unit tree, $H,\st E,G\in F$,  $m_{0}$ is a natural number,
\begin{equation}\label{mar18d}
L_1,\ M_1,\ X_1,\ \ldots,\ L_n,\ M_n,\ X_n\end{equation}
is a $\st E$-over-$G$ domination chain in $F$ mature at time $m_{0}$, and $b_1,\ldots,b_n$ are natural numbers such that, for each $k\in\{1,\ldots,n\}$, 
$\bot b_k\hat{\in}_{m_k}M_k$\vspace{2pt} (and hence also $\bot b_k\hat{\in}_{m_k}X_k$) for some $m_k\geq m_{0}$. Then there is at most one $k\in\{1,\ldots,n\}$ such that $\tilde{H}$ is an osubformula of $\tilde{X}_k$ and $\bot b_k\hat{\in} H$.
\end{lemma}

\begin{proof} Assume the conditions of the lemma.  For a contradiction, additionally assume that $1\leq k_1<k_2\leq n$, $ \tilde{H}$ is an osubformula of both $\tilde{X}_{k_1}$ and $\tilde{X}_{k_2}$, and both $\bot b_{k_1}\hat{\in}H$ and $\bot b_{k_2}\hat{\in} H$. 

We know that $\cal E$ made the numeric move $a_{k_1}$ in $M_{k_1}$ at time $m_{k_1}$. Analyzing the work of $\cal E$, this can easily be seen to imply that such a move was, in fact, made in the funital restriction $M_{k_1}\hspace{-2pt}\downarrow\hspace{-2pt}m_{k_1}$ of $M_{k_1}$. This, in turn, implies that $M_{k_1}\hspace{-2pt}\downarrow\hspace{-2pt}m_{k_1}$ is  
a subfunit of $H\hspace{-2pt}\downarrow\hspace{-2pt}m_{k_1}$ (otherwise we would not have $\bot a_{k_1}\hat{\in}H$). The latter  further implies that $M_{k_1}\hspace{-2pt}\downarrow\hspace{-2pt}m_{0}$ is a subfunit of $H\hspace{-2pt}\downarrow\hspace{-2pt}m_{0}$ (because $m_0\leq k_{1}$). It is therefore obvious that there is a unital extension $H'$ of 
$H\hspace{-2pt}\downarrow\hspace{-2pt} m_{0}$ such that $M_{k_1}$ is a subunit of $H'$. Using similar reasoning, we also find that  there is a unital extension $H''$ of $H\hspace{-2pt}\downarrow\hspace{-2pt} m_{0}$ such that $M_{k_2}$ is a subunit of $H''$. Also, of course, $\tilde{H}'=\tilde{H}$, so that $\tilde{H}'$ is an osubformula of $\tilde{X}_{k_1}$.

Since both $H'$ and $H''$ are unital extensions of $H\hspace{-2pt}\downarrow\hspace{-2pt} m_0$, we have $H'\hspace{-2pt}\downarrow\hspace{-2pt} m_0=H''\hspace{-2pt}\downarrow\hspace{-2pt} m_0$ ($=H\hspace{-2pt}\downarrow\hspace{-2pt} m_0$). This means that $H'$ and $H''$ are not $\downarrow\hspace{-2pt} m_0$-incomparable. Hence, due to the maturity of (\ref{mar18d}) at time $m_0$, $H'$ and $H''$ are simply the same. Thus, $H'$ is a superunit of both $M_{k_1}$ and $M_{k_2}$. 

Since both $H'$ and $X_{k_1}$ are superunits of $M_{k_1}$, $H'$ and $X_{k_1}$  have to be comparable. The latter, in view of the fact that  $\tilde{H}'$  is an osubformula of $\tilde{X}_{k_1}$, obviously  means that 
$H'$ is a subunit of $X_{k_1}$. But then, since $H'$ is a superunit of both $M_{k_1}$ and $M_{k_2}$, either $H'$ or one of its proper subunits is the smallest common superunit $Y$ of $M_{k_1}$ and $M_{k_2}$, where no proper $\st$-superunit of $M_{k_1}$ is a subunit of $Y$. 
So, $M_{k_1}$ drives $M_{k_2}$ through $Y$. In turn, since (\ref{mar18d}) is a $\st E$-over-$G$ domination chain, $M_{k_2}$ drives $L_{{k_2}+1}$ (if $k_2<n$) or $G$ (if $k_2=n$). Thus, by clause 1 of Lemma \ref{jan2a}, $M_{k_1}$ drives $L_{{k_2}+1}$ 
(if $k_2<n$) or $G$ (if $k_2=n$). This, however, is impossible according to the definition of domination chain. 
\end{proof}

\subsection{The operation \one}\label{ssone}
%\marginpar{ssone}

We say that a tuple \(({\cal A}, {\cal B}, {\cal C},\vec{T}, \st E,G)\) is {\bf $\one$-appropriate} iff 
$\cal A$ is a resolution, $\cal B$ is a pruning,  $\cal C$ is a pillaring concordant with $\cal B$, $\vec{T}$ is a domination chaintype,  $\st E,G\in\mathbb{F}_{1}^{\cal ABC}$, and 
 ${\cal C}(\st E)$ is a singleton.

For the rest of this section, we fix an arbitrary \one-appropriate tuple \(({\cal A}, {\cal B}, {\cal C},\vec{T}, \st E,G)\). The operation \one\ that we are going to describe takes such a tuple and returns a pillaring, denoted by
\[\one({\cal A}, {\cal B}, {\cal C},\vec{T}, \st E,G).\]

\subsubsection{The construction of \one}

Let us rename the pillaring $\cal C$ into ${\cal C}^{-1}$ (yes, the superscript is ``minus $1$''). 
We do our construction ---  henceforth referred to as the ``{\bf construction of \one}'' --- of the pillaring 
$\one({\cal A}, {\cal B}, {\cal C}, \vec{T}, \st E, G)$ in consecutive steps, starting from step $\#0$. 
Each step $\# i$ returns a 
pillaring ${\cal C}^{i}$ such that ${\cal C}^i$ is a hightening of ${\cal C}^{i-1}$, and ${\cal C}^{i}(\st E)$ is (remains) a singleton.  At each step $\# i$, the construction of \one\ is either {\bf complete} (terminates)  or  {\bf incomplete}. If it is complete, then step $\#i$ is the last step, and the pillaring ${\cal C}^i$ is declared to be  the sought value of $\one({\cal A}, {\cal B},{\cal C}, \vec{T}, \st E, G)$. Otherwise, we proceed to the next, $\#(i+1)$ step.\vspace{10pt}

{\bf Step $\# i$}: Let  $v_i$ be the (unique) ${\cal C}^{i-1}$-pillar of $\st E$, and let $S={\cal B}(\st E)$. 
Let ${\cal D}^i$ be the pillaring  such that 
${\cal D}^{i}(\st E)=\{[v_i]^S0\}$ and, on all other $\st$-units, ${\cal D}^{i}$ agrees with ${\cal C}^{i-1}$.

If there is no $\st E$-over-$G$ domination chain of type $\vec{T}$ in $\mathbb{F}_{1}^{{\cal ABD}^i}$, then we declare the construction to be {\em complete}, and let ${\cal C}^i={\cal D}^i$ ($= \one({\cal A}, {\cal B},{\cal C}, \vec{T}, \st E, G)$).

Otherwise we declare the construction to be {\em incomplete}, and pick --- according to some fixed choice function --- 
a $\st E$-over-$G$ domination chain
%\marginpar{mar15a}
\begin{equation}\label{mar15a} 
L_{1}^{i},\ M_{1}^{i},\ X_{1}^{i},\ \ldots,\ L_{n}^{i},\ M_{n}^{i},\ X_{n}^{i}
\end{equation}
(fix it!) of type $\vec{T}$ in    $\mathbb{F}_{1}^{{\cal ABD}^i}$.

Let $m^{i}_{0}$ be the smallest natural number satisfying the following conditions:
\begin{description}
  \item[(i)] the chain (\ref{mar15a}) is mature at time $m^{i}_{0}$;\footnote{Lemma \ref{mar18c} guarantees that $m_{0}^{i}$ can always be chosen so as to satisfy this condition.} 
  \item[(ii)] there is a $\Omega_ {{m}^{i}_{0}}$-active funit whose height exceeds the height of the pillaring ${\cal D}^i$;\footnote{Since $\cal E$ makes numeric moves in funits of ever increasing heights, $m_{0}^{i}$ can always be chosen so as to satisfy this condition. Of course, our earlier assumption that $\mathbb{F}_0$ is not $\sti$-free (see the paragraph following Lemma \ref{nov5}) is relevant here.}
  \item[(iii)] unless $i=0$, we have $m^{i}_{0}> m^{i-1}_{n}$, where $m_{n}^{i-1}$ is the value obtained at step $\#(i-1)$ in the same way as we are going to obtain the value $m_{n}^{i}$ at the present step shortly. 
\end{description}

For each $k\in\{1,\ldots,n\}$, we now define the natural numbers $l^{i}_{k},a^{i}_{k},t^{i}_{k},m^{i}_{k},b^{i}_{k}$ as follows:
\begin{itemize}
  \item $l_{k}^{i}$ is the smallest number  
with $l_{k}^{i}>m^{i}_{k-1}$ such that $\cal E$ moved in $L_{k}^{i}$ at time $l_{k}^{i}$, and  
$a_{k}^{i}$ is the numeric move made by $\cal E$ in $L_{k}^{i}$ at that time  (that is, $\bot a_{k}^{i}\hat{\in}_{l_{k}^{i}}L_{k}^{i}$);
\item $t_{k}^{i}$\vspace{2pt} is the time with $t_{k}^{i}>l_{k}^{i}$ at which 
 the adversary of $\cal E$ made the same numeric move $a_{k}^{i}$ in $M_{k}^{i}$ (that is, $\top a_{k}^{i}\hat{\in}_{t_{k}^{i}}M_{k}^{i}$);\footnote{$\cal E$'s adversary  indeed must have made the move $a_{k}^{i}$ in $M_{k}^{i}$, for otherwise $L_{k}^{i}$ and $M_{k}^{i}$ would not be opposite; however, this event could only have occurred after $\cal E$ made the move $a_{k}^{i}$ because,  as we remember from (\ref{feb23a}), $\cal E$ always chooses fresh numbers for its numeric moves.}
\item $m_{k}^{i}$ is the smallest 
number with $m_{k}^{i}> t^{i}_{k}$ such that $\cal E$ moved in $M_{k}^{i}$ at time $m_{k}^{i}$, and 
$b_{k}^{i}$ is the numeric move made by $\cal E$ in $M_{k}^{i}$ at that time  (that is, $\bot b_{k}^{i}\hat{\in}_{m_{k}^{i}}M_{k}^{i}$).
\end{itemize}

We now let ${\cal C}^{i}$ be the (unique) pillaring  satisfying the following conditions:
\begin{enumerate}
  \item ${\cal C}^{i}(\st E)=\{[v_i]^S1\}$. 
  \item Consider any unit $\st H$ other than $\st E$. Let $w_1,\ldots,w_p$ be the ${\cal C}^{i-1}$-pillars of 
$\st H$. Then the ${\cal C}^{i}$-pillars of $\st H$ are $w'_1,\ldots,w'_p$ where, for each $j\in\{1,\ldots,p\}$, $w'_j$ is defined as follows: 
\begin{enumerate}
  \item  Assume, for one of $k\in\{1,\ldots,n\}$, $\bot b_{k}^{i}\hat{\in}\st H$  and $\st \tilde{H}$ is an osubformula of $\tilde{X}_{k}^{i}$ (in view of Lemma \ref{mar16b}, such a $k$ is unique). Further assume that the move $b_{k}^{i}$  made by $\cal E$ at time $m^{i}_{k}$ was numerically  made in branch $u\hspace{-2pt}\uparrow$ of $\st H$. Then $w'_j=[u]^S1$. 
  \item Otherwise $w'_j=w_j$.
\end{enumerate}
  \end{enumerate}

This completes our description of step $\# i$.\vspace{5pt}

\begin{lemma}\label{mar27c}
%\marginpar{mar27c}
Where all parameters are as in the description of Step $\# i$, the pillaring ${\cal C}^i$ is a hightening of the pillaring ${\cal C}^{i-1}$ and, as long as the former is concordant with $\cal B$, so is the latter.  

Consequently (by induction on $i$), ${\cal C}^i$ is a hightening of $\cal C$ and is (remains) concordant with $\cal B$. 
\end{lemma}

\begin{proof} Rather obvious in view of our choice of ${\cal C}^i$ and the condition (ii) that $m^{i}_{0}$ is required to satisfy.  
\end{proof}

Of course, the value of $\one$ remains undefined on $({\cal A}, {\cal B}, {\cal C},\vec{T}, \st E,G)$ if the construction of $\one$ never terminates.  However, the following lemma, proven in Section \ref{ssjan}, rules out this possibility, so that $\one$ is defined on all $\one$-appropriate tuples:   

\begin{lemma}\label{jan6a}
%\marginpar{jan6a}
There is a number $i$  such that the construction of {\em \one}\ is complete at step $\# i$. 
\end{lemma}

\subsubsection{Termination of the construction of \one}\label{ssjan}
%\marginpar{ssjan}
This section is exclusively devoted to a proof of Lemma \ref{jan6a}. A fixed $\one$-appropriate tuple  we are dealing with is the same \(({\cal A}, {\cal B}, {\cal C},\vec{T}, \st E,G)\) as before. 
So are the parameters $v_i,S,L_{k}^{i},M_{k}^{i},X_{k}^{i},a^{i}_{j},b^{i}_{j},l^{i}_{j},t^{i}_{j},m^{i}_{j}$ from the construction of $\one$. Remember that the length of $\vec{T}$ is $3n$, and that $S={\cal B}(\st E)$. 

We additionally fix a sufficiently large integer $r$ such that 
%\marginpar{jan5a,jan5b}
\begin{equation}\label{jan5a}
\mbox{\em the width of $\cal C$ (and hence, by Lemma \ref{mar27c}, of each ${\cal C}^{i}$)  does not exceed $r$;}
\end{equation}
\begin{equation}\label{jan5b}
\mbox{\em the modal depth of no osubformula of $\mathbb{F}_0$ exceeds $r$.}
\end{equation}

%Remember that $S={\cal B}(\st E)$. Below, when we write $[z]$, it is to be understood as $[z]^S$.

For a contradiction, deny the assertion of Lemma \ref{jan6a}. In particular, assume that none of the steps $\#0$ through $\# r^{rn}$ of the construction of \one\ is final.

For a finite set $K$, we write $|K|$ to denote the cardinality of $K$.

\begin{lemma}\label{dec20a}
%\marginpar{dec20a} 
Assume $k\in\{1,\ldots,n\}$,  $\mathbb{S}$ is a nonempty subset of $\{0,\ldots,r^{rn}\}$, $i_0$ is the greatest element of $\mathbb{S}$, and the following condition is satisfied:
%\marginpar{jan23a}
\begin{equation}\label{jan23a} 
\mbox{For any  $i,j\in\mathbb{S}$ with $j<i$, $\bot a_{k}^{j}\hat{\not\in}L_{k}^{i}$.}
\end{equation}
Further, let
\[\mathfrak{S}_0\ =\ \{j\ |\ j\in \mathbb{S},\ \bot b_{k}^{j}\hat{\in}X_{k}^{i_0}\}.\] 
Then $|\mathfrak{S}_0|\leq r^r$. 
\end{lemma}

\begin{proof} 
Let $k$, $\mathbb{S}$, $i_0$, $\mathfrak{S}_0$ be as in the conditions of the lemma. In addition, we rename $X_{k}^{i_0}$  into $Y_0$ and assume that the extended form of the latter is $\tilde{Y}_{0}[\vec{x}_{0}^{i_0}]$.  Thus, 
\[X_{k}^{i_0}=Y_0=\tilde{Y}_0[\vec{x}_{0}^{i_0}].\]
For a contradiction, assume  
%\marginpar{jan24a}
\begin{equation}\label{jan24a}
|\mathfrak{S}_0|>r^r.
\end{equation}

Let $d$ be the number of the $\st$-superunits of $M_{k}^{i_0}$ that happen to be subunits of 
$Y_0$. Namely, let $\st Y_1,\ \ldots,\st Y_d$ be all such superunits, with  $\st Y_1$  being a proper superunit of $\st Y_2$, $\st Y_2$  a proper superunit of $\st Y_3$, and so on. At the same time, since $M_{k}^{i_0}$ drives $L_{k+1}^{i_0}$ (if $k<n$) or $G$ (if $k=n$) through $Y_0$, 
$M_{k}^{i_0}$ has 
no $\cost$-superunit that happens to be a subunit of  $Y_0$. This means that, for some infinite bitstrings $x_{1}^{i_0},\ldots,x_{d}^{i_0}$ (fix them), we have 
\[\begin{array}{rcl}
\st Y_1 & = & \st \tilde{Y}_1[\vec{x}^{i_0}],\\
 \st Y_2& = & \st \tilde{Y}_2[\vec{x}^{i_0},x_{1}^{i_0}],\\
 \st Y_3 & = &\st \tilde{Y}_3[\vec{x}^{i_0},x_{1}^{i_0},x_{2}^{i_0}],\\
 & \cdots &  \\
  \st Y_d & = &\st \tilde{Y}_d[\vec{x}^{i_0},x_{1}^{i_0},\ldots,x_{d-1}^{i_0}],\\ 
M_{k}^{i_0} & = & \tilde{M}_{k}^{i_0}[\vec{x}^{i_0},x_{1}^{i_0},\ldots,x_{d}^{i_0}].
\end{array}
\]

Similarly, for any $j\in\mathfrak{S}_0$, let $\vec{x}_{0}^{j},x_{1}^{j},\ldots,x_{d}^{j}$ be such that $X_{k}^{j}=\tilde{Y}_{0}[\vec{x}_{0}^{j}]$, the latter is a  superunit of  $\st \tilde{Y}_1[\vec{x}^{j}_{0}]$, the latter is a proper superunit of  
$\st \tilde{Y}_2[\vec{x}^{j}_{0},x_{1}^{j}]$, the latter  is a proper
superunit of $\st \tilde{Y}_3[\vec{x}^{j}_{0},x_{1}^{j},x_{2}^{j}]$,   \ldots, 
$\st \tilde{Y}_{d-1}[\vec{x}^{j}_{0},x_{1}^{j},\ldots,x_{d-2}^{j}]$ is a proper superunit of  $\st \tilde{Y}_{d}[\vec{x}^{j}_{0},x_{1}^{j},\ldots,x_{d-1}^{j}]$, and the latter is a superunit of $M_{k}^{j}=\tilde{M}_{k}^{j}[\vec{x}^{j}_{0},x_{1}^{j},\ldots,x_{d}^{j}]$.

According to our definition of $\mathfrak{S}_0$, we have: 
\[
\mbox{\em for any $j\in\mathfrak{S}_0$,  \ $\bot {b}^{j}_{k} \hat{\in}\tilde{Y}_{0}[\vec{x}^{i_0}_{0}]$.}\]

 For each $j\in\mathfrak{S}_0$,  let $w^{j}_{1}$  be the finite bitstring such that the corresponding move 
$b_{k}^{j}$  was numerically made (at time $m_{k}^{j}$) by $\cal E$ in branch $ w^{j}_{1}\hspace{-2pt}\uparrow$ of $\st \tilde{Y}_1[\vec{x}_{0}^{i_0}]$. Our choice of the pillaring ${\cal C}^{j}$ in (step $\#j$ of) the construction of $\one$ guarantees that each $w^{j}_{1}$ is a prefix of one of the ${\cal C}^{j}$-pillars of $\st \tilde{Y}_1[\vec{x}_{0}^{i_0}]$.\footnote{A relevant fact here is that $\sti \tilde{Y}_{1}[\vec{x}_{0}^{i_0}]\not=\sti E$. The latter is guaranteed by condition 4 of Definition \ref{dec10d}.} And since, by Lemma \ref{mar27c}, every subsequent pillaring (namely, the pillaring ${\cal C}^{i_0}$) is a hightening of the previous ones (namely, of the pillaring ${\cal C}^{j}$), we find that each $w^{j}_{1}$ is a prefix of one of the ${\cal C}^{i_0}$-pillars of $\st \tilde{Y}_1[\vec{x}_{0}^{i_0}]$. 
Next, by (\ref{jan5a}), we know that the number of such pillars does not exceed $r$. Thus, by the 
pigeonhole principle, there is a  ${\cal C}^{i_0}_{k}$-pillar $p_1$ of $\st Y_1[\vec{x}_{0}^{i_0}]$ and a subset $\mathfrak{S}_{1}$ of $\mathfrak{S}_0$ with 
\(|\mathfrak{S}_{1}|\geq \frac{|\mathfrak{S}_0|}{r}\) such that, for every $j\in\mathfrak{S}_{1}$, $w^{j}_{1}\preceq p_1$. But, by (\ref{jan24a}), 
$|\mathfrak{S}_0|>r^r$. So, \(|\mathfrak{S}_{1}|> r^{r-1}\). In view of (\ref{jan5b}), we also have $d\leq r$. Consequently,  
\[|\mathfrak{S}_1|> r^{d-1}.\]
Let $i_1$ be the greatest element of $\mathfrak{S}_{1}$. Consider an arbitrary $j\in \mathfrak{S}_{1}$. 
 Since  $b_{k}^{j}\hat{\in}\tilde{Y}_0[\vec{x}^{i_0}_{0}]$ and $\cal E$ makes every numeric move only once, the move $b_{k}^{j}$ was numerically made by $\cal E$ at time $m_{k}^{j}$  in both $\tilde{Y}_0[\vec{x}^{j}_{0}]$ and $\tilde{Y}_0[\vec{x}^{i_0}_{0}]$.  As observed in a similar situation in the proof of Lemma \ref{mar16b}, the move $b_{k}^{j}$ was, in fact, numerically made by $\cal E$ in the funital restrictions $\tilde{Y}_0[\vec{x}^{j}_{0}]\hspace{-2pt}\downarrow \hspace{-2pt}m_{k}^{j}$ and $\tilde{Y}_0[\vec{x}^{i_0}_{0}]\hspace{-2pt}\downarrow \hspace{-2pt}m_{k}^{j}$ of $\tilde{Y}_0[\vec{x}^{j}_{0}]$ and $\tilde{Y}_0[\vec{x}^{i_0}_{0}]$, respectively. Clearly this implies that 
 $\tilde{Y}_0[\vec{x}^{j}_{0}]\hspace{-2pt}\downarrow \hspace{-2pt}m_{k}^{j}=\tilde{Y}_0[\vec{x}^{i_0}_{0}]\hspace{-2pt}\downarrow \hspace{-2pt}m_{k}^{j}$. Applying the same reason to $i_1$ in the role  of $j$, we also find  $\tilde{Y}_0[\vec{x}^{i_1}_{0}]\hspace{-2pt}\downarrow \hspace{-2pt}m_{k}^{i_1}=\tilde{Y}_0[\vec{x}^{i_0}_{0}]\hspace{-2pt}\downarrow \hspace{-2pt}m_{k}^{i_1}$ and hence, as $j\leq i_1$,  $\tilde{Y}_0[\vec{x}^{i_1}_{0}]\hspace{-2pt}\downarrow \hspace{-2pt}m_{k}^{j}=\tilde{Y}_0[\vec{x}^{i_0}_{0}]\hspace{-2pt}\downarrow \hspace{-2pt}m_{k}^{j}$. Thus,  $\tilde{Y}_0[\vec{x}^{j}_{0}]\hspace{-2pt}\downarrow \hspace{-2pt}m_{k}^{j}=\tilde{Y}_0[\vec{x}^{i_1}_{0}]\hspace{-2pt}\downarrow \hspace{-2pt}m_{k}^{j}$. Since $j<i_1$, we also have $w^{j}_{1}\preceq w^{i_1}_{1}\preceq x^{i_1}_{1}$ (this can be seen from the fact that the heights of funits in which $\cal E$ makes moves keep increasing). In view of this observation, the equation 
$\tilde{Y}_0[\vec{x}^{j}_{0}]\hspace{-2pt}\downarrow \hspace{-2pt}m_{k}^{j}=\tilde{Y}_0[\vec{x}^{i_1}_{0}]\hspace{-2pt}\downarrow \hspace{-2pt}m_{k}^{j}$ obviously implies $\tilde{Y}_1[\vec{x}^{j}_{0},x^{j}_{1}]\hspace{-2pt}\downarrow \hspace{-2pt}m_{k}^{j}=\tilde{Y}_1[\vec{x}^{i_1}_{0},x^{i_1}_{1}]\hspace{-2pt}\downarrow \hspace{-2pt}m_{k}^{j}$. But the move $b_{k}^{j}$ was numerically made by $\cal E$ in $\tilde{Y}_1[\vec{x}^{j}_{0},x^{j}_{1}]$ at time $m_{k}^{j}$. Hence it was also numerically made (at the same time) in $\tilde{Y}_1[\vec{x}^{i_1}_{0},x^{i_1}_{1}]$. To summarize, we have:
\[
\mbox{\em  for any $j\in\mathfrak{S}_1$, \ $\bot b^{j}_{k} \hat{\in}\tilde{Y}_{1}[\vec{x}^{i_1}_{0},x^{i_1}_{1}]$.}\]

For each $j\in\mathfrak{S}_1$,  let $w^{j}_{2}$  be the finite bitstring such that the corresponding move 
$b_{k}^{j}$  was numerically made by $\cal E$ in branch $w^{j}_{2}\hspace{-2pt}\uparrow $ of $\st \tilde{Y}_2[\vec{x}_{0}^{i_1},x_{1}^{i_1}]$. As in the previous case, we find that each $w^{j}_{2}$ has to be a prefix of one of the ${\cal C}^{i_1}$-pillars of $\st \tilde{Y}_2[\vec{x}_{0}^{i_1},x_{1}^{i_1}]$. And, since 
the number of such pillars does not exceed $r$,  as in the previous case,  
 there is a  ${\cal C}^{i_1}$-pillar $p_2$ of $\st \tilde{Y}_2[\vec{x}_{0}^{i_1},x_{1}^{i_1}]$ and a subset $\mathfrak{S}_{2}$ of $\mathfrak{S}_1$ with 
\[|\mathfrak{S}_2|> r^{d-2}\]
 such that, for every $j\in\mathfrak{S}_{2}$, $w^{j}_{2}\preceq p_2$. 
Let $i_2$ be the greatest element of $\mathfrak{S}_{2}$. 
As in the previous case, we find that 
\[
\mbox{\em  for any $j\in\mathfrak{S}_2$, \ $\bot b^{j}_{k} \hat{\in}\tilde{Y}_{2}[\vec{x}^{i_2}_{0},x^{i_2}_{1},x^{i_2}_{2}]$.}\]

Continuing this way,
we eventually find that there is a subset $\mathfrak{S}_{d}$ of $\mathfrak{S}_{0}$ with 
%\marginpar{jan24c}
\begin{equation}\label{jan24c}
|\mathfrak{S}_{d}|>r^{d-d}=1
\end{equation}  
such that 
%\marginpar{jan15e}
\begin{equation}\label{jan15e}
\mbox{\em  for any $j\in\mathfrak{S}_d$, \ $\bot b^{j}_{k}\hat{\in}\tilde{Y}_{d}[\vec{x}^{i_d}_{0},x^{i_d}_{1},\ldots, x^{i_d}_{d}]$.}
\end{equation}

Remembering that 
$M_{k}^{i_d}=\tilde{M}_{k}[\vec{x}^{i_d}_{0},x^{i_d}_{1},\ldots, x^{i_d}_{d}]$, and that the move $b^{j}_{k}$ was made by $\cal E$ in $\tilde{M}_{k}^{j}$ (rather than in any other politeral subunits of $\tilde{Y}_{d}[\vec{x}^{j}_{0},x^{j}_{1},\ldots, x^{j}_{d}]$), with a little thought one can see that 
(\ref{jan15e}) implies the following: 
%\marginpar{jan15f}
\begin{equation}\label{jan15f}
\mbox{\em  for any $j\in\mathfrak{S}_d$, \ $\bot b^{j}_{k}\hat{\in} M_{k}^{i_d}$.}
\end{equation}
Let us rename $i_d$ into $i$. Further, let $j$ be an element of $\mathfrak{S}_d$ different from $i$, meaning that $j<i$ (the existence of such a $j$ is guaranteed by (\ref{jan24c})).  Obviously (\ref{jan15f}) implies that $\top a^{j}_{k}\hat{\in}M_{k}^{i}$, because $\cal E$ made its move $b_{k}^{j}$ only in units in which its adversary already had made the move $a_{k}^{j}$. Hence, as $M_{k}^{i}$ and $L_{k}^{i}$ are opposite, we have 
$\bot a^{j}_{k}\hat{\in}L_{k}^{i}$. 
 This, however, contradicts (\ref{jan23a}).  
\end{proof}

By induction, we now establish the existence of $n+1$ sets \[\mathbb{S}_{n}\subseteq \mathbb{S}_{n-1}\subseteq\ldots\subseteq\mathbb{S}_{1}\subseteq\mathbb{S}_{0}\subseteq \{0,\ldots,r^{rn}\}\] such that, for each $k\in\{0,\ldots,n\}$, the following three conditions are satisfied: 

%\marginpar{dec3aa}
\begin{equation}\label{dec3aa}
 \mbox{\em $|\mathbb{S}_k|> r^{r(n-k)}$.}
\end{equation}
%\marginpar{dec3ab}
\begin{equation}\label{dec3ab}
\mbox{\em  If $k>0$, then, for any  $i,j\in\mathbb{S}_{k}$ with $j<i$, $\bot b_{k}^{j}\hat{\not\in} X_{k}^{i}$.}
\end{equation}
%\marginpar{dec3ac}
\begin{equation}\label{dec3ac}
\mbox{\em  If $k<n$, then, for any  $i,j\in\mathbb{S}_{k}$ with $j<i$, $\tlg a^{j}_{k+1}\hat{\not\in} L_{k+1}^{i}$.}
\end{equation}

We define $\mathbb{S}_{0}$ by 
\[\mathbb{S}_{0}= \{0,\ldots,r^{rn}\}.\]
The cardinality of this set  set is $r^{rn}+1>r^{r(n-0)}$, so condition (\ref{dec3aa}) is  satisfied for $k=0$. Condition (\ref{dec3ab}) for $k=0$ is satisfied vacuously.  As for condition (\ref{dec3ac}), consider any $i,j\in \mathbb{S}_{0}$ such that $j<i$.
 With some analysis of the construction of \one\ we can  see that $\cal E$ numerically made the move $ a^{j}_{1}$ only in branches of  $\st E$ that extend $[v_j]^S0$, while it numerically made the move $ a^{i}_{1}$  only in branches of $\st E$ that extend $[v_j]^S1$. To summarize, $\cal E$ has never numerically made both moves  $a^{i}_{1}$ and $a^{j}_{1}$ in the same child of $\st E$, and hence in the same politeral unit $L_{1}^{e}$ for whatever $e\in\mathbb{S}_0$. Taking into account that $\cal E$ however {\em did}  make the move $a_{1}^{i}$ in $L_{1}^{i}$, we now find that it could not have  made the move $a_{1}^{j}$ in $L_{1}^{i}$. That is, $a_{1}^{j}\hat{\not\in}L_{1}^{i}$, so that 
 condition (\ref{dec3ac}) is also satisfied for $k=0$.

Now consider any $k$ with $0\leq k< n$ and assume that we have already established the existence of a set $\mathbb{S}_k$ satisfying conditions  (\ref{dec3aa}), (\ref{dec3ab}) and (\ref{dec3ac}). We  want to see that there is a set $\mathbb{S}_{k+1}$ satisfying the following three conditions (the same conditions as  (\ref{dec3aa}),  (\ref{dec3ab}) and (\ref{dec3ac}) but with  $k+1$ instead of $k$):

%\marginpar{dec17a}
\begin{equation}\label{dec17a}
\mbox{\em $|\mathbb{S}_{k+1}|>r^{r(n-(k+1))}$.}
\end{equation}
%\marginpar{dec3b}
\begin{equation}\label{dec3b}
\mbox{\em  For any  $i,j\in\mathbb{S}_{k+1}$ with $j<i$, $\bot b_{k+1}^{j}\hat{\not\in} X_{k+1}^{i}$.}
\end{equation}
%\marginpar{dec17c}
\begin{equation}\label{dec17c}
\mbox{\em  If $k+1<n$, then, for any  $i,j\in\mathbb{S}_{k+1}$ with $j<i$, $\bot a_{k+2}^{j}\hat{\not\in} L_{k+2}^{i}$.}
\end{equation}

Let $e_1$ be the greatest element of $\mathbb{S}_k$, and let 
\[\mathfrak{Q}_1 \ =\ \{j\ |\ j\in \mathbb{S}_k,\  \bot b^{j}_{k+1}\hat{\in}X_{k+1}^{e_1}\}.\]  
 Next, let $e_2$ be the greatest element of $\mathbb{S}_k-\mathfrak{Q}_1$, and let 
\[\mathfrak{Q}_2 \ =\ \{j\ |\ j\in \mathbb{S}_k-\mathfrak{Q}_1,\ \bot b^{j}_{k+1}\hat{\in}X_{k+1}^{e_2}\}.\]
 Next, let $e_3$ be the greatest element of $\mathbb{S}_k-(\mathfrak{Q}_1\cup \mathfrak{Q}_2)$, and let
\[\mathfrak{Q}_3 \ =\ \{j\ |\ j\in \mathbb{S}_k-(\mathfrak{Q}_1\cup \mathfrak{Q}_2),\ \bot b^{j}_{k+1}\hat{\in}X_{k+1}^{e_3}\}.\]
We continue in the above style and define pairs $(e_4,\mathfrak{Q}_4)$, $(e_5,\mathfrak{Q}_5)$, \ldots  until we hit a pair $(e_q,\mathfrak{Q}_q)$ such that selecting the subsequent pair is no longer possible --- namely, we have 
%\marginpar{gaza}
\begin{equation}\label{gaza}
\mbox{\em $\mathbb{S}_{k}-(\mathfrak{Q}_1\cup\ldots\cup\mathfrak{Q}_q)=\emptyset$, i.e. $\mathbb{S}_{k}=\mathfrak{Q}_1\cup\ldots\cup\mathfrak{Q}_q$.}
\end{equation}

Now, we define $\mathbb{S}_{k+1}$ by
\[\mathbb{S}_{k+1}=\{e_1,\ldots,e_q\}.\]

To verify condition (\ref{dec17a}), first note that, since $k<n$ in our present case, (\ref{dec3ac}) can be rewritten as 
\[\mbox{\em  For any  $i,j\in\mathbb{S}_{k}$ with $j<i$, $\bot a_{k+1}^{j}\hat{\not\in} L_{k+1}^{i}$.}\]

The above, of course, implies the following  for any $f\in\{1,\ldots,q\}$: 
%\marginpar{jan27b}
\begin{equation}\label{jan27b}
\mbox{\em  For any  $i,j\in\mathbb{S}_{k}-(\mathfrak{Q}_1\cup\ldots\cup \mathfrak{Q}_{f-1})$
 with $j<i$,  $\bot a_{k+1}^{j}\hat{\not\in} L_{k+1}^{i}$}
\end{equation}
(here, if $f=1$, $\mathbb{S}_{k}-(\mathfrak{Q}_1\cup\ldots\cup \mathfrak{Q}_{f-1})$ is to be understood as simply $\mathbb{S}_{k}$).
Now note that (\ref{jan27b}) is nothing but the  condition (\ref{jan23a}) of Lemma \ref{dec20a}, with $\mathbb{S}_{k}-(\mathfrak{Q}_1\cup\ldots\cup \mathfrak{Q}_{f-1})$ in the role of $\mathbb{S}$ and $k+1$ in the role of $k$. Hence, by Lemma \ref{dec20a}, with $e_f$ in the role of $i_0$ and $\mathfrak{Q}_f$ in the role of $\mathfrak{S}_0$, we find that 
%\marginpar{gazb}
\begin{equation}\label{gazb}
\mbox{\em $|\mathfrak{Q}_f|\leq r^r$ (for any $f\in\{1,\ldots,q\}$).}
\end{equation} 

Thus, by (\ref{gaza}) and (\ref{gazb}), $\mathbb{S}_k$ is the union of $q$ sets (the sets $\mathfrak{Q}_1,\ldots,\mathfrak{Q}_q$), with the cardinality  of none of those sets exceeding $r^r$. Therefore 
\[q\geq \frac{|\mathbb{S}_{k}|}{r^r},\]
i.e.,
 \[|\mathbb{S}_{k+1}|\geq \frac{|\mathbb{S}_{k}|}{r^r}.\]
At the same time, from (\ref{dec3aa}), we know that $|\mathbb{S}_{k}|>r^{r(n-k)}$. Thus,
\[|\mathbb{S}_{k+1}|> \frac{r^{r(n-k)}}{r^r}=r^{r(n-k)-r}=r^{r((n-k)-1)}=r^{r(n-(k+1))}.\]
 This proves the truth of condition (\ref{dec17a}).  

Now, consider any $i,j\in \mathbb{S}_{k+1}$ with $j<i$.  From our construction of $\mathbb{S}_{k+1}$, it is immediate that 
$\bot b_{k+1}^{j}\hat{\not\in} X_{k+1}^{i}$, which takes care of condition (\ref{dec3b}). To see that condition (\ref{dec17c}) is also satisfied, assume, for a contradiction, that $k+1<n$ and  $\bot a_{k+2}^{j}\hat{\in}L_{k+2}^{i}$.  With some thought, this can be seen to imply  that 
$X_{k+1}^{j}\hspace{-2pt}\downarrow\hspace{-2pt} l_{k+2}^{j}=X_{k+1}^{i}\hspace{-2pt}\downarrow\hspace{-2pt} l_{k+2}^{j}$.  Since $m_{k+1}^{j}<l_{k+2}^{j}$, we then also have $X_{k+1}^{j}\hspace{-2pt}\downarrow\hspace{-2pt} m_{k+1}^{j}=X_{k+1}^{i}\hspace{-2pt}\downarrow\hspace{-2pt} m_{k+1}^{j}$.
So, the move $b_{k+1}^{j}$ numerically made by $\cal E$ at time $m_{k+1}^{j}$ in ($M_{k+1}^{j}$ and hence in) $X_{k+1}^{j}$ was also simultaneously numerically made in $X_{k+1}^{i}$. That is, $\bot b_{k+1}^{j}\hat{\in} X_{k+1}^{i}$. But, from the already verified condition  (\ref{dec3b}), we know that this is not the case. 
 
This completes our construction of the sets $\mathbb{S}_0,\ldots,\mathbb{S}_n$ and our proof of the fact that each such set $\mathbb{S}_k$ satisfies conditions (\ref{dec3aa}), (\ref{dec3ab})  and (\ref{dec3ac}).\vspace{8pt}

Taking $k=n$ in (\ref{dec3aa}), we find that the set $\mathbb{S}_n$ contains at least two distinct elements $i$ and $j$. Let us say  $j<i$. By condition (\ref{dec3ab}),  $\tlg b^{j}_{n}\hat{\not\in} X_{n}^{i}$. On the other hand, of course, $\tlg b^{j}_{n}\hat{\in} X_{n}^{j}$. This implies that the units $X_{n}^{j}$ and $X_{n}^{i}$ are incomparable. But note that both units are superunits of $G$. We have thus reached a contradiction, because the same unit obviously cannot have two incomparable superunits.   Lemma \ref{jan6a} is now proven.

\subsubsection{The \one-lemma}
\begin{lemma}\label{oct10b}
%\marginpar{oct10b}
Assume $({\cal A}, {\cal B},{\cal C}, \vec{T}, \st E, G)$ is a \one-appropriate tuple, and \[{\cal C}'=\one({\cal A}, {\cal B},{\cal C}, \vec{T}, \st E, G).\] Then:

1. ${\cal C}'$ is a hightening of $\cal C$, and it is  concordant with $\cal B$.

2. There is no $\st E$-over-$G$ domination chain of type $\vec{T}$ in $\mathbb{F}_{1}^{{\cal ABC}'}$. 

\end{lemma}

\begin{proof} Assume the conditions of the present lemma. 
Clause 1 immediately follows from  Lemma \ref{mar27c}. 
For clause 2, let $i$ be the number whose existence is claimed in Lemma \ref{jan6a}. Then, where ${\cal C}^i$ and ${\cal D}^i$ are as in (step $\#i$ of) the construction of \one,  we have ${\cal C}'={\cal C}^i={\cal D}^i$. But, looking back at that construction,  we see that the reason why the latter  was complete at step $\# i$ is that there was no $\st E$-over-$G$ domination chain of type $\vec{T}$ in 
$\mathbb{F}_{1}^{{\cal ABD}^i}=\mathbb{F}_{1}^{{\cal ABC}'}$.  
\end{proof}

\subsection{The operation $\one^+$} Remember that we have defined a {\em domination chaintype} as any sequence of $3n$ oformulas where $n\geq 1$. Weakening the condition $n\geq 1$ to $n\geq 0$ yields what we refer to as a {\bf generalized domination chaintype}. That is, the set of generalized domination chaintypes is that of domination chaintypes plus the empty sequence of oformulas. 

We define a {\bf $\one^+$-appropriate tuple} \(({\cal A}, {\cal B}, {\cal C},\vec{T}, \st E,G)\) exactly as we defined a $\one$-appropriate tuple, with the only difference that $\vec{T}$ can now be empty. That is, now $\vec{T}$ is required to be a generalized domination chaintype rather than a domination chaintype. 

When \(({\cal A}, {\cal B}, {\cal C},\vec{T}, \st E,G)\)  is a $\one^+$-appropriate tuple, we define the pillaring 
\[{\cal C}'=\one^+({\cal A}, {\cal B}, {\cal C},\vec{T}, \st E,G)\]
as follows:
\begin{itemize}
  \item If $\vec{T}$ is nonempty, then ${\cal C}'=\one({\cal A}, {\cal B}, {\cal C},\vec{T}, \st E,G)$.
\item If $\vec{T}$ is empty but $G$ is not a proper subunit of $\st E$, then ${\cal C}'={\cal C}$.
\item Suppose now $\vec{T}$ is empty and $G$ is a proper subunit of $\st E$. Let $\st \tilde{E}[\vec{y}]$ be the extended form of $\st E$, and   $\tilde{E}[\vec{y},x]$ be the (unique) child of $\st E$ which is a superunit of $G$. 
Also, let $v$ be the (unique) ${\cal C}$-pillar of $\st E$, and $S={\cal B}(\st E)$. Note that exactly one of the two strings $[v]^S0$ or $[v]^S1$ is a prefix of $x$. If $[v]^S0\preceq x$ (resp. $[v]^S1\preceq x)$, we define ${\cal C}'$ as the pillaring that sends $\st E$ to $\{[v]^S1\}$ (resp. $\{[v]^S0\}$) and agrees with ${\cal C}$ on all other $\st$-units. 
\end{itemize}

\begin{lemma}\label{oct10bplus}
%\marginpar{oct10bplus}
Assume $({\cal A}, {\cal B},{\cal C}, \vec{T}, \st E, G)$ is a $\one^+$-appropriate tuple, and \[{\cal C}'=\one^+({\cal A}, {\cal B},{\cal C}, \vec{T}, \st E, G).\] Then:

1. ${\cal C}'$ is a hightening of $\cal C$, and it is concordant with $\cal B$.

2. If $\vec{T}$ is nonempty, then there is no $\st E$-over-$G$ domination chain of type $\vec{T}$ in $\mathbb{F}_{1}^{{\cal ABC}'}$. 

3. If $\vec{T}$ is empty and $G$ is a proper subunit of $\st E$, then $G\not\in\mathbb{F}_{1}^{{\cal ABC}'}$.

\end{lemma}

\begin{proof} Clauses 1 and 2 are immediate from the corresponding two clauses of Lemma \ref{oct10b} and the way ${\cal C}'$ is defined. Clause 3 can be verified through straightforward analysis. 
\end{proof}

\subsection{The operation \two}
By a {\bf $\two$-appropriate tuple} we mean $({\cal A}, {\cal B}, {\cal C},\vec{T}, \st E,G,\mathfrak{T})$, where 
\(({\cal A}, {\cal B}, {\cal C},\vec{T}, \st E,G)\) is  a $\one^+$-appropriate tuple and $\mathfrak{T}$ is a set of units such that, for any unit $\st H\in\mathfrak{T}$, ${\cal C}(\st H)$ is a singleton.

Let $({\cal A}, {\cal B}, {\cal C},\vec{T}, \st E,G,\mathfrak{T})$ be a $\two$-appropriate tuple, and let 
  ${\cal C}'=\one^+ ({\cal A}, {\cal B}, {\cal C},\vec{T}, \st E,G)$. We define 
\[\two ({\cal A}, {\cal B}, {\cal C},\vec{T}, \st E,G,\mathfrak{T})\]
as the pillaring ${\cal C}''$ such that, for any unit $\st H$:
\begin{itemize}
  \item if $\st H\in \mathfrak{T}$, then  ${\cal C}''(\st H) = {\cal C}'(\st H)$; 
  \item otherwise,  where ${\cal C}'(\st H) = \{w_1,w_2,\ldots,w_r\}$,  ${\cal B}(\st H)=S$ and $[w]$ means $[w]^{S}$, we have ${\cal C}''(\st H)=\{[w_1]0,\ [w_1]1,\ [w_2]0,\ [w_2]1,\ \ldots,\ [w_r]0,\ [w_r]1\}$.  
 \end{itemize}

\begin{lemma}\label{oct10bb}
%\marginpar{oct10bb}
Assume $({\cal A}, {\cal B}, {\cal C},\vec{T}, \st E,G,\mathfrak{T})$ is  a \two-appropriate tuple,  and \[{\cal C}''=\two 
({\cal A}, {\cal B}, {\cal C},\vec{T}, \st E,G,\mathfrak{T}).\]   Then:

1. For any unit $\st H\in\mathfrak{T}$, ${\cal C}''(\st H)$ is a singleton.

2. If $\vec{T}$ is nonempty, then there is no  $\st E$-over-$G$ domination chain of type $\vec{T}$ in $\mathbb{F}_{1}^{{\cal ABC}''}$. 

3. If $\vec{T}$ is empty and $G$ is a proper subunit of $\st E$, then $G\not\in\mathbb{F}_{1}^{{\cal ABC}''}$. 

4. ${\cal C}''$ is concordant with $\cal B$. 

5. For any unit $\st H$ and  ${\cal C}''$-pillar $w$ of $\st H$, there is a ${\cal C}$-pillar $v$ of $\st H$ with $v\preceq w$.

\end{lemma}

\begin{proof} Clauses 1, 2 and 3 rather immediately follow from the corresponding three clauses of Lemma \ref{oct10bplus}. Clauses  4 and 5 easily follow from clause 1 of Lemma \ref{oct10bplus} and the way ${\cal C}''$ is defined. 
\end{proof}

\subsection{The operation \three}

By a {\bf \three-appropriate} tuple we mean a triple $({\cal A},{\cal B},\mathfrak{T})$, where $\cal A$ is a resolution, $\cal B$ is a pruning,
and $\mathfrak{T}$ is a countable set of units closed under visibility in $\cal A$ (for any units $G,E$, if $G\in \mathfrak{T}$ and $E$ is visible to $G$ in $\cal A$, then $E\in \mathfrak{T}$).

The operation \three\ takes a \three-appropriate tuple $({\cal A},{\cal B},\mathfrak{T})$  and returns a \three-appropriate tuple 
\[({\cal A}',{\cal B}',\mathfrak{T}') = \three ({\cal A},{\cal B},\mathfrak{T}),\]
 defined as follows.

If there are no $\cal A$-unresolved $\st$-units in $\mathfrak{T}$, then $({\cal A}',{\cal B}',\mathfrak{T}')=({\cal A},{\cal B},\mathfrak{T})$.  

Otherwise, let 
%\marginpar{dec28a}
\begin{equation}\label{dec28a}
(\st E_1,G_1,\vec{T}_1),\ (\st E_2,G_2,\vec{T}_2),\ (\st E_3,G_3,\vec{T}_3),\ \ldots
\end{equation}
 be a list --- generated according to some fixed choice function --- of all triples $(\st E,G,\vec{T})$, where $\st E$ is an $\cal A$-unresolved unit of $\mathfrak{T}$, $G\in\mathfrak{T}$,    and $\vec{T}$ is a generalized domination chaintype. 

We first define the pillarings ${\cal C}_0,{\cal C}_1,{\cal C}_2,\ldots$  as follows:
\begin{itemize}
  \item ${\cal C}_0=\emptyset$ (remember from Section \ref{sstod} that $\emptyset$ is the trivial pillaring);
  \item ${\cal C}_{i+1}=\two({\cal A}, {\cal B}, {\cal C}_i,\vec{T}_i, \st E_i,G_i,\mathfrak{T})$.
\end{itemize} 

It is obvious that the trivial pillaring is concordant with every resolution. So, ${\cal C}_0$ is concordant with $\cal B$. In view of clause 4 of Lemma \ref{oct10bb}, by induction on $i$, we further see that each of the above pillarings ${\cal C}_i$ is concordant with $\cal B$ and hence ${\cal C}_{i+1}$ is well defined. 

We now define the resolution ${\cal A}'$ as follows:

\begin{itemize}
  \item For every $\cal A$-unresolved element $\st E$ of $\mathfrak{T}$, ${\cal A}'(\st E)=x$, where $x$ is an element of $\limitset({\cal B}(\st E))$  --- selected according to some fixed choice function --- such that, for each $i\geq 0$, the (unique) ${\cal C}_i$-pillar of $\st E$ is a prefix of $x$. In view of clauses 1, 4 and 5 of Lemma \ref{oct10bb}, such an $x$ exists.
%\footnote{In fact, one could show that such an $x$ is unique, so no real choice is needed in ``selecting'' it. But why bother.}
  \item For all other units $\st H$, ${\cal A}'(\st H)={\cal A}(\st H)$.  
\end{itemize}
Observe that
%\marginpar{apr1a}
\begin{equation}\label{apr1a}
\mbox{\em ${\cal A}'$ is an extension of $\cal A$.}
\end{equation}

Next, we define the pruning ${\cal B}'$\footnote{That ${\cal B}$ is indeed a pruning will be verified in Lemma \ref{may13a}.}  by stipulating that, for any unit $\st E$, we have:
\begin{itemize}
  \item If $\st E\in \mathfrak{T}$, then ${\cal B}'(\st E)={\cal B}(\st E)$.
  \item Otherwise, ${\cal B}'(\st E)=\{\epsilon\}\cup{\cal C}_1(\st E)\cup {\cal C}_2(\st E)\cup{\cal C}_3(\st E)\cup\ldots$. 
\end{itemize}

Finally, we define the set $\mathfrak{T}'$ to be the closure of $\mathfrak{T}$ under visibility in ${\cal A}'$.  
In view of Lemma \ref{nov8cf}, $\mathfrak{T}'$ is countable. 

\begin{lemma}\label{may13a}
%\marginpar{may13a}
Assume $({\cal A},{\cal B},\mathfrak{T})$ is a \three-appropriate tuple, and $({\cal A}',{\cal B}',\mathfrak{T}')=\three({\cal A},{\cal B},\mathfrak{T})$. Then:

1. ${\cal B}'$ is a pruning such that, for every unit $\st E$, ${\cal B}'(\st E)\subseteq {\cal B}(\st E)$. 

2. Where ${\cal C}_0,{\cal C}_1,{\cal C}_2,\ldots$ are as defined above, for each $i\geq 0$ we have $\mathbb{F}_{1}^{{\cal A}'{\cal B}'}\subseteq \mathbb{F}_{1}^{{\cal ABC}_i}$.  
\end{lemma}

\begin{proof} Assume the conditions of the lemma. Let ${\cal C}_0,{\cal C}_1,{\cal C}_2,\ldots$ be as defined earlier in this section. 

{\em Clause 1}: Consider an arbitrary unit $\st E$. We want to show that ${\cal B}'(\st E$) is a pruner (so that ${\cal B}'$ is indeed a pruning as promised), and that ${\cal B}'(\st E)\subseteq {\cal B}(\st E)$. 
If $\st E\in\mathfrak{T}$, then ${\cal B}'(\st E)$ is the same as ${\cal B}(\st E)$, and (taking into account that ${\cal B}(\st E)$ is a pruner) we are done.  Now assume $\st E\not\in\mathfrak{T}$. By our definition of ${\cal B}'$, ${\cal B}'(\st E)$ contains $\epsilon$. Next, consider any bitstring $u$ such that $u$ is a prefix of some element of ${\cal B}'(\st E)$. Obviously there is a number $i$ and a bitstring $w$ such that $u\preceq w$ and $w\in {\cal C}_i(\st E)$. But observe from the construction of $\two$ that, where $S={\cal B}(\st E)$, for a certain $w'$ with $w\preceq w'$ and hence of $u\preceq w'$, the strings $[w']^S0$ and $[w']^S1$ are among the ${\cal C}_{i+1}$-pillars of $\st E$, and therefore are among the elements of ${\cal B}'(\st E)$. This is exactly what it takes for   ${\cal B}'(\st E)$ to be a pruner. Next, relying on clause 4 of Lemma \ref{oct10bb}, by induction on $j$, we find that each ${\cal C}_{j}$ ($j\geq 0$) is concordant with $\cal B$, meaning that ${\cal C}_j(\st E)\subseteq {\cal B}(\st E)$.   Since every element of ${\cal B}'(\st E)$ other than $\epsilon$ is an element of ${\cal C}_j(\st E)$ for some $j\geq 0$, we find that  ${\cal B}'(\st E)\subseteq {\cal B}(\st E)$.

{\em Clause 2}: Consider an arbitrary $i\geq 0$. For a contradiction, assume $\mathbb{F}_{1}^{{\cal A}'{\cal B}'}$ is not a subset of $\mathbb{F}_{1}^{{\cal ABC}_i}$. Let $G$ be a unit with $G\in \mathbb{F}_{1}^{{\cal A}'{\cal B}'}$ and 
$G\not\in \mathbb{F}_{1}^{{\cal ABC}_i}$. We may assume that all proper superunits of $G$ are in 
$\mathbb{F}_{1}^{{\cal ABC}_i}$.
By (\ref{apr1a}) and clause 1 of the present lemma, $\mathbb{F}_{1}^{{\cal A}'{\cal B}'}\subseteq \mathbb{F}_{1}^{{\cal AB}}$, so that $G\in \mathbb{F}_{1}^{{\cal AB}}$. It is therefore obvious that the only possible reason why we still have 
$G\not\in\mathbb{F}_{1}^{{\cal ABC}_i}$ is that $G=\tilde{E}[\vec{x},y]$ is a child of some $\cal A$-unresolved unit $\st E=\st\tilde{E}[\vec{x}]$ such that $y$ is not comparable with any of the ${\cal C}_i$-pillars of $\st E$.
This immediately rules out the possibility that $\st E\in \mathfrak{T}$, because then, from the way we defined the resolution ${\cal A}'$ it is clear that $G$ cannot be in $\mathbb{F}_{1}^{{\cal A}'{\cal B}'}$. So, $\st E\not\in \mathfrak{T}$. Notice that then $\st E$ is ${\cal A}'$-unresolved. Therefore, 
since $G\in \mathbb{F}_{1}^{{\cal A}'{\cal B}'}$, we have $y\in \limitset({\cal B}'(\st E))$. The latter implies that there are infinitely many prefixes $w_1\preceq w_2\preceq w_3\preceq\ldots$ of $y$ that are elements of ${\cal B}'(\st E)$. Pick a $w_j$ among these prefixes such that $w_j$ is longer than any of the ${\cal C}_{i}$-pillars of $\st E$. In view of clause 5 of Lemma \ref{oct10bb}, it is obvious that one of the ${\cal C}_i$-pillars of $\st E$ should be a prefix of $w_j$. But then $y$ is comparable with such a pillar, which is a contradiction.
\end{proof}

\begin{lemma}\label{jan17a}
%\marginpar{jan17a}
Assume $({\cal A},{\cal B},\mathfrak{T})$ is a \three-appropriate tuple, $\st E$ is an $\cal A$-unresolved element of $\mathfrak{T}$, and $G$ is an element of $\mathfrak{T}$.   Then, where $({\cal A}',{\cal B}',\mathfrak{T}') = \three ({\cal A},{\cal B},\mathfrak{T})$, we have:    
\begin{enumerate}
  \item   $\st E$ is ${\cal A}'$-resolved, and ${\cal A}'(\st E)\in\limitset({\cal B}(\st E))$.
  \item  $\st E$ does not dominate $G$ in $\mathbb{F}_{1}^{{\cal A}'{\cal B}'}$. 
\end{enumerate}
\end{lemma}

\begin{proof} Assume the conditions of the lemma. Clause 1 is obvious from the way the resolution ${\cal A}'$ is defined. 
To prove clause 2, for a contradiction, assume $\st E$ dominates $G$ in $\mathbb{F}_{1}^{{\cal A}'{\cal B}'}$. 

First, assume the reason for $\st E$ dominating $G$ in $\mathbb{F}_{1}^{{\cal A}'{\cal B}'}$ is that $G$ is a proper subunit of $\st E$. Let $i$ be the number such that, for the corresponding triple $(\st E_i,G_i,\vec{T}_i)$ from the list (\ref{dec28a}) we have $\st E_i=\st E$, $G_i=G$ and $\vec{T}_i$ is empty. From clause 3 of Lemma \ref{oct10bb} we find that $\mathbb{F}_{1}^{{\cal ABC}_{i+1}}$ does not contain $G$. Therefore, in view of clause 2 of Lemma \ref{may13a}, $\mathbb{F}_{1}^{{\cal A}'{\cal B}'}$ does not contain $G$ either. But this is a contradiction because, if $G$ is not in a given unit tree, it cannot be dominated (by whatever) in that tree.

Now assume the reason for $\st E$ dominating $G$ in $\mathbb{F}_{1}^{{\cal A}'{\cal B}'}$ is that there is a $\st E$-over-$G$ domination chain $\vec{D}$ in $\mathbb{F}_{1}^{{\cal A}'{\cal B}'}$. Let $\vec{T}$ be the type of $\vec{D}$. Let $i$ be the number such that, for the corresponding triple $(\st E_i,G_i,\vec{T}_i)$ from (\ref{dec28a}) we have $\st E_i=\st E$, $G_i=G$ and $\vec{T}_i=\vec{T}$. From clause 2 of Lemma \ref{oct10bb} we find that there is no $\st E$-over-$G$ domination chain of type $\vec{T}$ in $\mathbb{F}_{1}^{{\cal ABC}_{i+1}}$. That is, $\vec{D}$ is not a domination chain in $\mathbb{F}_{1}^{{\cal ABC}_{i+1}}$.
Since, on the other hand, $\vec{D}$ is a domination chain in $\mathbb{F}_{1}^{{\cal A}'{\cal B}'}$, we find that there is a unit $H$ in $\vec{D}$ such that $H\in\mathbb{F}_{1}^{{\cal A}'{\cal B}'}$ but not $H\in\mathbb{F}_{1}^{{\cal ABC}_{i+1}}$. However, by clause 2 of Lemma \ref{may13a}, this is impossible.  
\end{proof}

\subsection{The operation \four}

The operation \four\ takes a pair $({\cal A},J)$ such that $\cal A$ is a resolution and $J$ is a politeral unit of $\mathbb{F}_{1}^{\cal A}$, and returns a resolution  
\[{\cal R}=\four ({\cal A},J).\]

To construct such a resolution $\cal R$,  
we define the infinite sequence \[({\cal A}_0,{\cal B}_0,\mathfrak{T}_0),\ ({\cal A}_1,{\cal B}_1,\mathfrak{T}_1),\ ({\cal A}_2,{\cal B}_2,\mathfrak{T}_2),\ \ldots\]
of \three-appropriate tuples as follows. 
\begin{itemize}
  \item ${\cal A}_0={\cal A}$, \ ${\cal B}_0$ is the trivial pruning $\emptyset$, and $\mathfrak{T}_0$ is the closure of the set $\{J\}$ under 
visibility in ${\cal A}$. The countability of $\mathfrak{T}_0$ is guaranteed by Lemma \ref{nov8cf}. 
  \item $({\cal A}_{i+1},{\cal B}_{i+1},\mathfrak{T}_{i+1})=\three({\cal A}_i, {\cal B}_i,\mathfrak{T}_{i})$.
\end{itemize} 

Now, we define $\cal R$ as the union (see Section \ref{ss17n}) of the resolutions ${\cal A}_0,{\cal A}_1,{\cal A}_2,\ldots$. Such a union is well defined because, in view of (\ref{apr1a}), each ${\cal A}_{i+1}$ is an extension of ${\cal A}_i$, and hence the resolutions  ${\cal A}_0,{\cal A}_1,{\cal A}_2,\ldots$ are pairwise consistent. It is also clear that 
%\marginpar{apr1b}
\begin{equation}\label{apr1b}
\mbox{\em $\cal R$ is an extension of $\cal A$, as well as of each ${\cal A}_i$ ($i\geq 0$).}
\end{equation}

For subsequent references, we also define  $\mathbb{T}$ as the union of the sets $\mathfrak{T}_0,\mathfrak{T}_1,\mathfrak{T}_2,\ldots$.
%Since each $\mathfrak{T}_i$ is countable, so is $\mathbb{T}$. 
It is rather obvious that $\mathbb{T}$ is closed under visibility in $\cal R$.  

Henceforth we shall refer to the earlier-stated  assumptions on the parameters ${\cal A}$, $J$ and the definitions of the parameters ${\cal A}_i$, ${\cal B}_i$, $\mathfrak{T}_i$, ${\cal R}$, $\mathbb{T}$ as the {\bf construction of \four}.

\begin{lemma}\label{jan17bnew}
%\marginpar{jan17bnew}
Let all parameters be as in the construction of \four. Further, let $i$ be a natural number, 
 $G$ any element of $\mathfrak{T}_i$, and $\st E$ any  unit that dominates $G$ in 
$\mathbb{F}_{1}^{\cal R}$.  Then we have:    
\begin{enumerate}
  \item $\st E\in \mathbb{T}$;
  \item   $\st E$ is $\cal R$-resolved; 
  \item $\st E$ dominates $G$ in $\mathbb{F}_{1}^{{\cal A}_{i+1}{\cal B}_{i+1}}$.
\end{enumerate}
\end{lemma}

\begin{proof} Assume the conditions of the lemma.

{\em Clause 1}:
If the reason for $\st E$ dominating $G$ in $\mathbb{F}_{1}^{\cal R}$ is that $G$ is a proper subunit of $\st E$, then, by clause 1 of Lemma \ref{jan6cc}, $\st E$ is visible to $G$ in  ${\cal A}_i$. But $\mathfrak{T}_i$ is closed under visibility in ${\cal A}_i$, so $\st E\in\mathfrak{T}_i$.

Now assume the reason for $\st E$ dominating $G$ in $\mathbb{F}_{1}^{\cal R}$ is that there is a $\st E$-over-$G$ domination chain
%\marginpar{may15b}
\begin{equation}\label{may15b}
L_1,\ M_1,\ X_1,\ \ldots,\ L_{n},\ M_{n},\ X_n
\end{equation}
 in 
$\mathbb{F}_{1}^{\cal R}$. Thus, $M_{n}$ drives $G$. 

Let $\st E_1,\ldots,\st E_d$ be all ${\cal A}_i$-unresolved $\st$-superunits of $M_n$ that happen to be subunits  of $X_n$.  Here we assume that $\st E_1$ is a proper superunit of $\st E_2$, $\st E_2$ is a proper superunit of $\st E_3$, etc.  Obviously 
$\st E_1$ ${\cal A}_i$-strictly drives $G$ (through $X_n$) and hence $\st E_1$ is visible to $G$ in  ${\cal A}_i$. This implies that $\st E_1\in \mathfrak{T}_i$, because $\mathfrak{T}_i$ is closed under visibility in ${\cal A}_i$. Therefore, by clause 1 of Lemma \ref{jan17a}, $\st E_1$ is ${\cal A}_{i+1}$-resolved. Then, by (\ref{apr1b}), $\st E_1$ is also $\cal R$-resolved, and the $\cal R$-resolvent of $\st E_1$ is the same as the ${\cal A}_{i+1}$-resolvent of $\st E_1$. But $\st E_2$ has to be a subunit of the $\cal R$-resolvent of $\st E_1$ (otherwise we would have $M_n\not\in \mathbb{F}_{1}^{\cal R}$ and, for this reason, (\ref{may15b}) would not be a domination chain in $\mathbb{F}_{1}^{\cal R}$). So, $\st E_2$ is a subunit of the ${\cal A}_{i+1}$-resolvent of $\st E_1$.
This obviously makes $\st E_2$ visible to $G$ in
${\cal A}_{i+1}$ (namely, $\st E_2$ ${\cal A}_{i+1}$-strictly drives $G$ through $X_n$). By using similar reasoning to the above (and with Lemma \ref{may1a} in mind), we now find that $\st E_2\in\mathfrak{T}_{i+1}$, $\st E_2$ is ${\cal A}_{i+2}$-resolved and $\st E_3$ is a subunit of the ${\cal A}_{i+2}$-resolvent of $\st E_2$, which makes $\st E_3$ visible to $G$ in
${\cal A}_{i+2}$, and so on. 
 This way, we eventually find that  
$M_{n}$ is visible to $G$ in ${\cal A}_{i+d}$ and thus ($M_n\in\mathfrak{T}_{i+d}$ and hence) $M_n\in\mathbb{T}$. 
 By  Lemma \ref{may1a},  $M_{n}$ is also visible to $G$ in $\cal R$, because $\cal R$ is an extension of ${\cal A}_{i+d}$. The visibility of  $M_{n}$ to $G$ in $\cal R$, of course, also implies the visibility if $L_n$ to $G$ in $\cal R$, because $M_n$ and $L_n$ are opposite. Furthermore, by clause 2 of Lemma \ref{jan6cc}, all  superunits of $M_n,L_n$ (including $M_n$ and $L_n$ themselves) 
are visible to $G$ in ${\cal R}$ and therefore are in $\mathbb{T}$, because $\mathbb{T}$ is closed under visibility in ${\cal R}$. 

Now starting from $L_{n}$ instead of $G$ and applying a similar reason, we find that $M_{n-1}$ and $L_{n-1}$, together with all of their superunits, are  visible to $L_n$ in $\cal R$, and hence are elements of $\mathbb{T}$. Further continuing this way, we  eventually find that 
$M_{n-2},L_{n-2},\ldots,M_1,L_1$ are  also in $\mathbb{T}$. To summarize:
%\marginpar{may17a}
\begin{equation}\label{may17a}
\mbox{\em All units of (\ref{may15b}), together with all of their superunits,  are in $\mathbb{T}$.}
\end{equation}

But $\st E$ is a superunit of $M_1$.  Hence, by (\ref{may17a}),  $\st E\in \mathbb{T}$, as desired.  

{\em Clause 2}: According to the already verified clause 1 of the present lemma, 
$\st E\in\mathbb{T}$. This  means nothing but that $\st E\in \mathfrak{T}_j$ for some $j$. If $\st E$ is ${\cal A}_j$-resolved, then it is also ${\cal R}$-resolved because, by (\ref{apr1b}), ${\cal R}$ is an extension of ${\cal A}_j$. And 
if $\st E$ is ${\cal A}_j$-unresolved, then, in view of clause 1 
of Lemma \ref{jan17a}, $\st E$ is ${\cal A}_{j+1}$-resolved, and hence, again by (\ref{apr1b}), it is also $\cal R$-resolved.

{\em Clause 3}:  According to the conditions of our lemma,  $\st E$ dominates $G$ in $\mathbb{F}_{1}^{\cal R}$, which, by the definition of domination, implies that both $\st E$ and $G$ are in $\mathbb{F}_{1}^{\cal R}$.

 First, assume that the reason for $\st E$ dominating $G$ in $\mathbb{F}_{1}^{\cal R}$ is that $G$ is a 
proper subunit of $\st E$. Then, in order to show that $\st E$ dominates $G$ in $\mathbb{F}_{1}^{{\cal A}_{i+1}{\cal B}_{i+1}}$, 
it is sufficient to simply verify that both $\st E$ and $G$ are (not only in $\mathbb{F}_{1}^{\cal R}$ but also) in  $\mathbb{F}_{1}^{{\cal A}_{i+1}{\cal B}_{i+1}}$. 

We first examine $G$. For a contradiction, assume $G\not\in\mathbb{F}_{1}^{{\cal A}_{i+1}{\cal B}_{i+1}}$.
Since $G\in\mathbb{F}_{1}^{\cal R}$ and $\cal R$ is an extension of ${\cal A}_{i+1}$, we have  
$G\in \mathbb{F}_{1}^{{\cal A}_{i+1}}$. Therefore, the only reason for  $G\not\in \mathbb{F}_{1}^{{\cal A}_{i+1}{\cal B}_{i+1}}$ can be that there is an ${\cal A}_{i+1}$-unresolved unit $\st H=\st \tilde{H}[\vec{x}]$ of ($\mathbb{F}_{1}^{\cal R}$ and) $\mathbb{F}_{1}^{{\cal A}_{i+1}}$ and a child 
$\tilde{H}[\vec{x},y]$ of $\st \tilde{H}[\vec{x}]$ such that $G$ is a subunit of $\tilde{H}[\vec{x},y]$ but $\tilde{H}[\vec{x},y]\not\in \mathbb{F}_{1}^{{\cal A}_{i+1}{\cal B}_{i+1}}$, meaning that $y\not\in\limitset({\cal B}_{i+1}(\st H))$. By clause 1 of Lemma \ref{jan6cc}, $\st H$ is visible to $G$ in ${\cal A}_{i+1}$. Therefore, as $\mathfrak{T}_{i+1}$ is closed under visibility in ${\cal A}_{i+1}$, we have $\st H\in\mathfrak{T}_{i+1}$. Hence, by clause 1 of Lemma \ref{jan17a}, $\st H$ is ${\cal A}_{i+2}$-resolved and ${\cal A}_{i+2}(\st H)\in \limitset({\cal B}_{i+1}(\st H))$. But, since ${\cal R}$ is an extension of  ${\cal A}_{i+2}$ and $y= {\cal R}(\st H)$ (otherwise $G$ would not be in $\mathbb{F}_{1}^{\cal R}$), we have ${\cal A}_{i+2}(\st H)=y$. Thus, $y\in \limitset({\cal B}_{i+1}(\st H))$, contradicting  the earlier established $y\not\in \limitset({\cal B}_{i+1}(\st H))$.

We now examine $\st E$, and assume $\st E\not\in\mathbb{F}_{1}^{{\cal A}_{i+1}{\cal B}_{i+1}}$. As in the preceding case,
since $\st E\in\mathbb{F}_{1}^{\cal R}$ and $\cal R$ is an extension of ${\cal A}_{i+1}$, we have  
$\st E\in \mathbb{F}_{1}^{{\cal A}_{i+1}}$. Therefore, the only reason for  $\st E\not\in \mathbb{F}_{1}^{{\cal A}_{i+1}{\cal B}_{i+1}}$ can be that there is an ${\cal A}_{i+1}$-unresolved unit $\st H=\st \tilde{H}[\vec{x}]$ of ($\mathbb{F}_{1}^{\cal R}$ and) $\mathbb{F}_{1}^{{\cal A}_{i+1}}$ and a child 
$\tilde{H}[\vec{x},y]$ of $\st \tilde{H}[\vec{x}]$ such that $\st E$ is a subunit of $\tilde{H}[\vec{x},y]$ but $\tilde{H}[\vec{x},y]\not\in \mathbb{F}_{1}^{{\cal A}_{i+1}{\cal B}_{i+1}}$, meaning that $y\not\in\limitset({\cal B}_{i+1}(\st H))$. By clause 1 of the present lemma, $\st E\in\mathbb{T}$. This means nothing but that $\st E\in\mathfrak{T}_{j}$ for some $j$.   
By clause 1 of Lemma \ref{jan6cc}, $\st H$ is visible to $\st E$ in ${\cal A}_{j}$. Therefore, as $\mathfrak{T}_{j}$ is closed under visibility in ${\cal A}_{j}$, we have $\st H\in\mathfrak{T}_{j}$. We may assume that $j$ is the smallest number with $\st H\in\mathfrak{T}_{j}$. 
Obviously $\st H$ is ${\cal A}_{j}$-unresolved because, in the process of applying $\four$, an ${\cal A}_e$-unresolved $\st$-unit becomes ${\cal A}_{e+1}$-resolved (if and) only if that unit is in $\mathfrak{T}_{e}$.\footnote{And, since $\sti H$ is ${\cal A}_{i+1}$-unresolved, it is ``originally unresolved'', i.e., is ${\cal A}_0$-unresolved.} Therefore, by 
  clause 1 of Lemma \ref{jan17a}, $\st H$ is ${\cal A}_{j+1}$-resolved.  Remember that, on the other hand,   $\st H$ is ${\cal A}_{i+1}$-unresolved. This clearly implies that $i+1<j+1$. By 
  clause 1 of Lemma \ref{jan17a}, we also have ${\cal A}_{j+1}(\st H)\in {\cal B}_{j}(\limitset(\st H))$.   But, since ${\cal R}$ is an extension of  ${\cal A}_{j+1}$ and $y= {\cal R}(\st H)$, we have ${\cal A}_{j+1}(\st H)=y$. Thus, $y\in \limitset({\cal B}_{j}(\st H))$. At the same time, since $i+1\leq j$, in view of clause 1 of Lemma \ref{may13a}, we have ${\cal B}_{j}\subseteq {\cal B}_{i+1}$ and hence $\limitset({\cal B}_{j}(\st H))\subseteq \limitset({\cal B}_{i+1}(\st H))$. Thus, $y\in \limitset({\cal B}_{i+1}(\st H))$. This, however, contradicts the earlier established $y\not\in \limitset({\cal B}_{i+1}(\st H))$.

Now assume the reason for $\st E$ dominating $G$ in $\mathbb{F}_{1}^{\cal R}$ is that there is a $\st E$-over-$G$ domination chain in $\mathbb{F}_{1}^{\cal R}$. Namely, we may assume that (\ref{may15b}) is such a chain. In order to show that $\st E$ also dominates $G$ in $\mathbb{F}_{1}^{{\cal A}_{i+1}}$, it is sufficient to show that both $\st E$ and $G$ are in $\mathbb{F}_{1}^{{\cal A}_{i+1}}$ and that (\ref{may15b}) is a $\st E$-over-$G$ domination chain in $\mathbb{F}_{1}^{{\cal A}_{i+1}}$. In turn, in order to show that (\ref{may15b}) is a $\st E$-over-$G$ domination chain in $\mathbb{F}_{1}^{{\cal A}_{i+1}}$, it is sufficient to show that all of the units $L_1,M_1,\ldots,L_n,M_n$ are in   $\mathbb{F}_{1}^{{\cal A}_{i+1}}$. So, consider any $Z$ with 
\[Z\in\{\st E, G, L_1,M_1,\ldots,L_n,M_n\}.\]
Re-examining our argument in the previous paragraph, we see that it goes through for $Z$ in the role of $\st E$ as long as $Z$  is in $\mathbb{T}$. But $Z$ is indeed in $\mathbb{T}$: if $Z=G$, then the fact  $Z\in\mathbb{T}$  is immediate from the condition $G\in\mathfrak{T}_i$ of the present lemma; otherwise, the fact $Z\in\mathbb{T}$ follows from (\ref{may17a}).  
\end{proof}

\begin{lemma}\label{jan17b}
%\marginpar{jan17b}
Let all parameters be as in the construction of \four. Further, let $i$ be a natural number, 
 $G$ any element of $\mathfrak{T}_i$, and $\st E$ any ${\cal A}_i$-unresolved 
unit. Then 
$\st E$ does not dominate $G$ in 
$\mathbb{F}_{1}^{\cal R}$.  
\end{lemma}

\begin{proof} Assume the conditions of the lemma. For a contradiction, also assume that   $\st E$ dominates $G$ in 
$\mathbb{F}_{1}^{\cal R}$. By clause 1 of Lemma \ref{jan17bnew}, we have $\st E\in\mathbb{T}$, meaning that $\st E\in\mathfrak{T}_j$ for some $j$. We may assume that $j$ is the smallest number with $\st E\in\mathfrak{T}_j$. Obviously 
$\st E$ is ${\cal A}_j$-unresolved because, as observed in the proof of Lemma \ref{jan17bnew}, in the construction of \four, an  ${\cal A}_e$-unresolved $\st$-unit becomes ${\cal A}_{e+1}$-resolved only when the unit is in $\mathfrak{T}_e$.  By clause 1 of Lemma \ref{jan17a}, $\st E$ is ${\cal A}_{j+1}$-resolved. This implies that $j+1>i$ (otherwise $\st E$ would be ${\cal A}_i$-resolved). So, since $G\in \mathfrak{T}_i$, we also have  $G\in \mathfrak{T}_j$. 
Now, in view of  clause 2 of Lemma \ref{jan17a}, $\st E$ does not dominate $G$ in $\mathbb{F}_{1}^{{\cal A}_{j+1}{\cal B}_{j+1}}$. This, in turn, by clause 3 of Lemma \ref{jan17bnew}, implies that  $\st E$ does not dominate $G$ 
in $\mathbb{F}_{1}^{\cal R}$, which is a contradiction. 
\end{proof}

\subsection{The resolution $\cal T$ and its totality}
It is well known that ({\bf ZFC} proves that) every set can be well-ordered. So, let $\prec$ be some fixed well-ordering of the set of all politeral units.

With  each politeral unit $J$   we associate a resolution ${\cal T}_J$  defined by 
\[{\cal T}_J=   \four({\cal A}_J,J),\]
where ${\cal A}_J$ is the union of all resolutions ${\cal T}_{J'}$ such that $J'\prec J$. 

Now, we define the ultimate resolution $\cal T$ as the union  of all resolutions ${\cal T}_J$ associated with politeral units.  In view of (\ref{apr1b}), such a union is well defined, because every ${\cal T}_J$ is an extension of each ${\cal T}_{J'}$ with $J'\prec J$. Of course, we also have:
%\marginpar{apr1c}
\begin{equation}\label{apr1c}
\mbox{\em For each politeral unit $J$, ${\cal T}$ is an extension of ${\cal T}_J$.}
\end{equation}

We claim that $\cal T$ is a total resolution, as promised in Lemma \ref{nov5}. Indeed, for a contradiction, assume $\cal T$ is not total. Let $\st E$ be a $\cal T$-unresolved $\st$-unit of $\mathbb{F}_{1}^{\cal T}$. Pick any politeral unit $J$ among the subunits of $\st E$. Let ${\cal A}$ be the union of all resolutions ${\cal T}_{J'}$ with $J'\prec J$. 
 Since, by (\ref{apr1c}), $\cal T$ is an extension of $\cal A$, $\st E$ is $\cal A$-unresolved. Let, along with our present $\cal A$ and $J$, the parameters ${\cal A}_i,{\cal B}_i,\mathfrak{T}_i,{\cal R}$ be as in the construction of \four.  Of course, $J\in\mathfrak{T}_0$ and $\st E$ dominates $J$ in $\mathbb{F}_{1}^{\cal R}$. But then, by clause 2 of Lemma \ref{jan17bnew}, $\st E$ is $\cal R$-resolved  and thus also $\cal T$-resolved. This is a contradiction.

\subsection{Proof of clause (i) of Lemma \ref{nov5}}
 Let $J$ be the $\prec$-smallest of all politeral units of $\mathbb{F}_{1}^{\cal T}$.  
For a contradiction, assume $\st E$ is a unit that dominates the root $\mathbb{F}_0[\hspace{1pt}]$ of $\mathbb{F}_{1}^{\cal T}$ in $\mathbb{F}_{1}^{\cal T}$. Of course, the reason for this domination cannot be that $\mathbb{F}_0[\hspace{1pt}]$ is a proper subunit of $\st E$, because $\mathbb{F}_0[\hspace{1pt}]$ is not a subunit of anything but itself. 
  Let 
\[L_1,\ M_1,\ X_1,\ \ldots,\ L_{n},\ M_{n},\ X_n\]
be a $\st E$-over-$\mathbb{F}_{0}[\hspace{1pt}]$ domination chain in $\mathbb{F}_{1}^{\cal T}$. 
So, $M_{n}$ drives $\mathbb{F}_{0}[\hspace{1pt}]$. Let $i$ be the smallest number among $1,\ldots,n$ such that $M_i$ drives $\mathbb{F}_{0}[\hspace{1pt}]$. Since $M_i$ drives $\mathbb{F}_{0}[\hspace{1pt}]$,  
$M_i$ has no 
$\cost$-superunits.  Therefore, obviously, $M_i$ drives $J$  as, by the way, it drives any other unit. But then, where $Y$ is the smallest common superunit of $M_i$ and $J$,  the following sequence is  a $\st E$-over-$J$ domination chain in  $\mathbb{F}_{1}^{\cal T}$:
\[L_1,\ M_1,\ X_1,\ \ldots,\ L_{i},\ M_{i},\ Y.\]
Thus,
%\marginpar{janua29a}
\begin{equation}\label{janua29a}
\mbox{\em $\st E$ dominates $J$ in  $\mathbb{F}_{1}^{\cal T}$.}
\end{equation}

Let $\cal A$ be the trivial resolution, and let this $\cal A$, along with all other parameters, be as in the construction of \four. Note that  
$\st E$ (as any other $\st$-unit) is ${\cal A}_0$-unresolved, and $J$ is (of course) an element of $\mathfrak{T}_0$. 
Hence, by  Lemma \ref{jan17b},
$\st E$ does not dominate $J$ in $\mathbb{F}_{1}^{{\cal R}}$. Since $\cal T$ is an extension of ${\cal R}$, $\st E$ does not dominate $J$ in $\mathbb{F}_{1}^{{\cal T}}$, either.
This, however, contradicts (\ref{janua29a}).

\subsection{Proof of clause (ii) of Lemma \ref{nov5}}

Consider any units $\st E$ and $\st G$ of $\mathbb{F}_{1}^{\cal T}$. Our goal is to show that they do not dominate each other in $\mathbb{F}_{1}^{\cal T}$. For a contradiction, deny this. 

Since the resolution $\cal T$ is total, both $\st E$ and $\st G$ have to be ${\cal T}_J$-resolved  for some politeral unit $J$. Assume $J$ is the $\prec$-smallest of all politeral units such that at least one of the units $\st E$, $\st G$ is ${\cal T}_J$-resolved. Let $\cal A$  be the union of all resolutions ${\cal T}_{J'}$ with $J'\prec J$.  Let this $\cal A$ ($={\cal A}_0$) and all other parameters be as in the construction of \four. Let $i$ be the smallest number such that one of the units $\st E$, $\st G$  is ${\cal A}_i$-resolved. Namely, we may assume that $\st G$ is the one (or {\em a} one) of the two units that is ${\cal A}_i$-resolved. Note that $i>0$, and that $\st G\in\mathfrak{T}_{i-1}$ (otherwise the ${\cal A}_{i-1}$-unresolved $\st G$ would not have become ${\cal A}_i$-resolved). At the same time, 
$\st E$ is ${\cal A}_{i-1}$-unresolved. 
Therefore, by  Lemma \ref{jan17b}, $\st E$ does not dominate $\st G$ in 
$\mathbb{F}_{1}^{\cal R}$, and hence  $\st E$ does not dominate $\st G$ in 
$\mathbb{F}_{1}^{\cal \cal T}$ either. This is however a contradiction.

\subsection{Proof of clause (iii) of Lemma \ref{nov5}} Throughout this proof, ``dominates'' is to be understood as ``dominates in $\mathbb{F}_{1}^{\cal T}$''. Similarly for ``domination chain''. Assume $\st E$,  $\st G$ and $H$ are units of $\mathbb{F}_{1}^{\cal T}$, $\st E$ dominates $\st G$, and $\st G$ dominates $H$. Our goal is to show that $\st E$ dominates $H$. We need to consider the following cases. 

{\em Case 1}: The reason for $\st E$ dominating   $\st G$ is that there is a $\st E$-over-$\st G$ 
domination chain
%\marginpar{jan4a}
\begin{equation}\label{jan4a}
L_{1},\ M_{1},\ X_1,\ \ldots, \ L_{n},\ M_{n},\ X_n.
\end{equation}

{\em Subcase 1.1}: The reason for $\st G$ dominating $H$ is that there is a $\st G$-over-$H$ domination chain
%\marginpar{jan4b}
\begin{equation}\label{jan4b}
L'_{1},\ M'_{1},\ X'_1,\ \ldots,\  L'_{n'},\ M'_{n'}, \ X'_{n'}.
\end{equation}
For notational convenience, let us agree that 
\[H=L_{n'+1}.\]

{\em Subsubcase 1.1.1}: No $M'_i$ ($1\leq i\leq n'$) is a subunit of $\st E$.  

{\em Subsubsubcase 1.1.1.1}: 
No $M_i$ ($1\leq i\leq n$) drives  any $L'_j$ ($1\leq j\leq n'+1$).
From the fact that (\ref{jan4a}) is a $\st E$-over-$\st G$ 
domination chain, we know that   $M_{n}$  drives $\st G$. From the fact that (\ref{jan4b}) is a $\st G$-over-$H$ 
domination chain,  we also know that $L'_{1}$ is a subunit of $\st G$. Hence, by Lemma \ref{jan3a} and clause 1 of Lemma \ref{jan2a},  $M_{n}$  drives $L'_{1}$. But then, where 
$Y$ is the smallest common superunit of $M_n$ and $L'_1$, the following is clearly  a 
 $\st E$-over-$H$ domination chain: 
\[L_{1}, \ M_{1},\  X_1,\ \ldots,\  L_{n},\ M_{n},\ Y,\  L'_{1},\  M'_{1},\ X'_1,\ \ldots,\  L'_{n'},\ M'_{n'},\ X'_{n'}.\]

{\em Subsubsubcase 1.1.1.2}: Some $M_i$ ($1\leq i\leq n$) drives one of  $L'_j$ ($1\leq j\leq n'+1$). We may assume that $i$ is the smallest number such that $M_i$ drives one of  $L'_j$, and that $j$ is the greatest number such that $M_i$ drives $L'_j$. But note that then, where $Y$ is the smallest common superunit of $M_i$ and $L'_j$, the following is a 
 $\st E$-over-$H$ domination chain: 
\[L_{1}, \ M_{1},\ X_1,\ \ldots,\  L_{i},\ M_{i},\ Y,\ L'_{j},\ M'_{j},\ X'_j,\ \ldots,\  L'_{n'},\ M'_{n'},\ X'_{n'}.\]

{\em Subsubcase 1.1.2}: Some $M'_i$ ($1\leq i\leq n'$) is a subunit of $\st E$. We may assume that $i$ is the greatest such number. 

First, suppose  $M'_i$  drives $\st E$. Let $j$ be the smallest number among $1,\ldots,i$ such that $M'_j$ drives $\st E$, and $Y$ be the smallest common superunit of $M'_j$ and $\st E$.   Then  
\[L'_{1},\ M'_{1}, \ X'_{1},\ \ldots,\  L'_{j},\ M'_{j}, \ Y\] 
is a $\st G$-over-$\st E$ domination chain, meaning that $\st G$ dominates $\st E$. But this is impossible due to the already proven clause (ii) of Lemma \ref{nov5}. 

Now suppose $M'_i$  does not drive $\st E$, meaning that $M'_i$ has a $\cost$-superunit that happens to be a subunit of $\st E$. 
On the other hand, $M'_i$ drives 
$L'_{i+1}$. This obviously implies that $L'_{i+1}$ is a proper subunit of $\st E$. But then we have: (1) if $i=n'$, then  $H$ ($=L'_{n'+1}$) is a proper subunit of $\st E$ and hence $\st E$ dominates $H$; and (2) if $i<n'$, then the following is a $\st E$-over-$H$ domination chain:
\[L'_{i+1},\ M'_{i+1},\ X'_{i+1}\ \ldots,\  L'_{n'},\ M'_{n'}, \ X'_{n'}.\] 

{\em Subcase 1.2}: The reason for $\st G$ dominating $H$ is that the latter is a proper subunit of the former. 

{\em Subsubcase 1.2.1}: None of the units $M_1,\ldots,M_{n-1}$ drives $H$. By Lemma \ref{jan3a}, $\st G$ drives $H$. From the fact that (\ref{jan4a}) is a $\st E$-over-$\st G$ 
domination chain, we also know that $M_n$ drives $\st G$. Hence, by clause 1 of Lemma \ref{jan2a}, $M_n$ drives $H$.  But then, where $Y$ is the smallest common superunit of $M_n$ and $H$, the following is a $\st E$-over-$H$ domination chain:
\[
L_{1},\ M_{1},\ X_1,\ \ldots, \ L_{n},\ M_{n},\ Y.\]

{\em Subsubcase 1.2.2}: One of the units $M_i$ with $1\leq i<n$  drives $H$. We may assume that $i$ is the smallest such number. But then, where $Y$ is the smallest common superunit of $M_i$ and $H$, the following is a $\st E$-over-$H$ domination chain:
\[L_{1},\ M_{1},\ X_1,\ \ldots, \ L_{i},\ M_{i},\ Y.\]

{\em Case 2}:  The reason for $\st E$ dominating   $\st G$ is that the latter is a proper subunit of the former.

{\em Subcase 2.1}: The reason for $\st G$ dominating $H$ is that there is a $\st G$-over-$H$ domination chain. Let (\ref{jan4b}) be such a chain.  
If none of $M'_i$ ($1\leq i\leq n'$) is a subunit of $\st E$, then the same chain is a $\st E$-over-$H$ domination chain. 
And if some $M'_i$ ($1\leq i\leq n'$) is a subunit of $\st E$, then we reason as in Subsubcase 1.1.2 and find that $\st E$ dominates $H$.

{\em Subcase 2.2}: The reason for $\st G$ dominating $H$ is that the latter is a proper subunit of the former. Then $H$ is a proper subunit of $\st E$, and hence $\st E$ dominates $H$.

\section{Hyperformulas and hypercirquents}

In addition to the language of $\fif$, in this section we consider an infinitary propositional language with the logical vocabulary $\{\gneg,\mlc,\mld\}$ and  a continuum of atoms. Atoms (resp. literals, formulas) of this language will be referred to as ``hyperatoms'' (resp. ``hyperliterals'', ``hyperformulas''). 

Namely, every {\bf hyperatom} is a triple $(P,A,B)$, where $P$ is an atom of the language of $\fif$, and $A$ and $B$ are sets of natural numbers (and vice versa: every such triple is a hyperatom). Next, every {\bf hyperliteral} is either $\alpha$ (a positive hyperliteral) of $\gneg \alpha$ (a negative hyperliteral), where $\alpha$ is a hyperatom.   
Two hyperliterals are said to be {\bf opposite} if one is $\alpha$ and the other is $\gneg \alpha$ for some hyperatom $\alpha$. Finally, {\bf hyperformulas} are defined inductively as follows:  
\begin{description}
  \item[(a)] Every hyperatom is a hyperformula.
  \item[(b)] Whenever $\alpha$ is a hyperformula, so is $\gneg \alpha$. 
  \item[(c)] Whenever $S$ is a (not necessarily finite or even countable) set of hyperformulas, so are $\mlc S$ and $\mld S$.
  \item[(d)] Nothing is a hyperformula unless it can be obtained by repeated applications of (a), (b) and (c). 
  \end{description} 

A {\bf subhyperformula} in the context of hyperformulas means the same as a subformula in the context of formulas. 

What we  call a {\bf hypermodel} is nothing but a truth assignment for hyperformulas, i.e., a function $h$ that sends every hyperatom $\alpha$ to one of the values {\em true} or {\em false}, and extends to all hyperformulas in the standard way. When the value of $h$ at a hyperformula $\alpha$ is {\em true}, we say that $\alpha$ is {\bf true under} $h$, or that $h$ {\bf makes $\alpha$ true}. Similarly for ``false''. A hyperformula $\alpha$ is said to be {\bf satisfiable} iff there is a hypermodel that makes it true; otherwise it is {\bf unsatisfiable}. And a hyperformula $\alpha$ is a {\bf tautology}, or is {\bf tautological}, iff every hypermodel makes it true.

A {\bf positive} (resp. {\bf negative}) {\bf occurrence} of a hyperatom in a hyperformula is an occurrence that is in the scope of even (resp. odd) number of occurrences of $\gneg$. Now, we say that a hyperformula $\alpha$ is {\bf binary} iff every hyperatom has at most one positive and at most one negative occurrence in it. A hyperformula $\alpha$ is a {\bf binary tautology} iff it is both binary and tautological. 
 
In the context of the run $\Omega$ fixed in Section \ref{sscs}, with each politeral unit  $L$ we associate a hyperliteral $L^{\circ}$. 
Namely, let $\Theta_L$ be the projection of $\Omega$ on $L$, we let  $A_L$ and $B_L$ be the following sets of natural numbers:
\[\begin{array}{c}
 A_L\ =\ \{a\  |\ \mbox{\em $\Theta_L$ contains the labmove $\top a$}\};\\
 B_L\ =\ \{b\  |\ \mbox{\em $\Theta_L$ contains the labmove $\bot b$}\}. 
\end{array}\]
In other words, $A_L$ is the set of natural numbers enumerated by $\cal E$'s adversary in $L$, and $B_L$ is the set of natural numbers enumerated by $\cal E$. 
 Then, where $P$ is the atom of the language of $\fif$ such that $\tilde{L}=P$ (resp. $\tilde{L}=\gneg P$),   we stipulate that $L^{\circ}$ is the hyperliteral $(P,A,B)$ (resp. $\gneg (P,B,A))$. 

\begin{lemma}\label{may4a}
%\marginpar{may4a}
Any two politeral units $L$ and $M$ are opposite iff so are the corresponding hyperliterals $L^{\circ}$ and $M^{\circ}$. 
\end{lemma}

\begin{proof} Straightforward. \end{proof}

We now turn the unit tree $\mathbb{F}_{1}$ into a hyperformula 
\[\overline{\mathbb{F}_1}\]
as follows. $\mathbb{F}_1$ can be viewed as a tree where every node is labeled with (rather than {\em is}) a (the corresponding) unit. Then $\overline{\mathbb{F}_1}$ is a hyperformula which --- more precisely, whose parse tree --- is obtained from $\mathbb{F}_1$ through replacing every $\mlc$- or $\st$-unit by $\mlc$, every $\mld$- or $\cost$-unit by $\mld$, and every politeral unit $L$ with the hyperliteral $L^{\circ}$.   Whenever a hyperformula $\beta$ (together with all of its subhyperformulas) is obtained in this process through replacing a unit $E$ (together with all of its subunits), we say that 
 $E$ is the {\bf unital origin} of $\beta$, and denote $E$ by \[\uorigin(\beta).\]

The hyperformula 
\[\overline{\mathbb{F}_{1}^{\cal T}}\]
and the function of {\bf unital origin} for its subhyperformulas are defined in the same way, with the only difference that now the starting unit tree for applying the above-described replacements is $\mathbb{F}_{1}^{\cal T}$ rather than $\mathbb{F}_{1}$.  Note that, in view of the totality of the resolution $\cal T$, every conjunctive subhyperformula of $\overline{\mathbb{F}_{1}^{\cal T}}$ originating from a $\st$-unit has a single conjunct. Thus, unlike disjunctive subhyperformulas (namely, those originating from $\cost$-units), all conjunctive subhyperformulas of $\overline{\mathbb{F}_{1}^{\cal T}}$ have only finitely many (specifically, one or two) conjuncts.   

In view of Lemmas \ref{feb24a} and \ref{may4a}, it is clear that 
%\marginpar{apr3a}
\begin{equation}\label{apr3a}
\mbox{\em $\overline{\mathbb{F}_1}$ is binary, and hence so is $\overline{\mathbb{F}_{1}^{\cal T}}$.}
\end{equation}

In turn, from (\ref{apr3a}), it is clear that different occurrences of subhyperformulas in $\overline{\mathbb{F}_{1}}$ or $\overline{\mathbb{F}_{1}^{\cal T}}$ are always different as hyperformulas; therefore, there is no need to maintain a terminological distinction between ``subhyperformula'' and  ``osubhyperformula''.

\begin{lemma}\label{feb26a}
%\marginpar{feb26a}
$\overline{\mathbb{F}_1}$ is not tautological. 
\end{lemma}

The rest of this section is devoted to a proof of Lemma \ref{feb26a}. For a contradiction, assume 
%\marginpar{apr2b}
\begin{equation}\label{apr2b} \mbox{\em $\overline{\mathbb{F}_1}$ is tautological.}
\end{equation}

Observe that $\overline{\mathbb{F}_{1}^{\cal T}}$ is the result of deleting in $\overline{\mathbb{F}_{1}}$ some conjuncts of some conjunctive subhyperformulas. The result of such a deletion obviously cannot destroy tautologicity. So, (\ref{apr2b}) immediately implies that  
%\marginpar{apr2a}
\begin{equation}\label{apr2a}
\mbox{\em $\overline{\mathbb{F}_{1}^{\cal T}}$ is tautological.}
\end{equation}

We now convert $\overline{\mathbb{F}_{1}^{\cal T}}$ to its disjunctive normal form, which we denote by 
\(\mathbb{F}_2.\)
Namely, $\mathbb{F}_2$ is obtained from  $\overline{\mathbb{F}_{1}^{\cal T}}$ in the standard way by repeatedly applying distributivity and changing every conjunction of disjunctions to an equivalent disjunction of conjunctions. It is not hard to see that such a hyperformula $\mathbb{F}_2$ has a continuum of disjuncts (unless $\mathbb{F}_1$ did not contain $\cost$), where, however, each disjunct only has {\em finitely many} conjuncts. So, $\mathbb{F}_2$ is in fact a disjunction of ordinary, ``finitary'' formulas of classical logic. $\mathbb{F}_2$ is equivalent (in the standard classical sense) to $\overline{\mathbb{F}_{1}^{\cal T}}$ and hence, just like the latter, is a tautology. The tautologicity of $\mathbb{F}_2$ means that $\gneg \mathbb{F}_{2}$ is unsatisfiable. Applying the standard DeMorgan conversions to $\gneg \mathbb{F}_2$, we turn it into an equivalent --- and hence also unsatisfiable --- hyperformula \(\mathbb{F}_3,\) which is a (probably uncountably long) conjunction of finite disjunctions of hyperliterals. Let $S$ be the set of all conjuncts of $\mathbb{F}_3$. Thus, $S$ is unsatisfiable. Then, by the compactness theorem,\footnote{The compactness theorem is known to hold not only for countable languages, but also for languages with uncountable sets of atoms.} there is a {\em finite} subset $S'$ of $S$ which is unsatisfiable. This, in turn, means that there is an unsatisfiable hyperformula $\mathbb{F}'_3$ which results from $\mathbb{F}_3$ by deleting all but finitely many conjuncts. This, in turn, implies that there is a tautological hyperformula $\mathbb{F}'_2$ which results from $\mathbb{F}_2$ by deleting all but finitely many disjuncts. Applying distributivity to $\mathbb{F}'_2$ in the opposite direction (opposite to the direction used when obtaining $\mathbb{F}_2$ from $\overline{\mathbb{F}_{1}^{\cal T}}$), we get a tautological hyperformula 
\[\mathbb{F}_4\]
 which is the result of deleting in $\overline{\mathbb{F}_{1}^{\cal T}}$ all but finitely many disjuncts in all disjunctive subhyperformulas.  For simplicity, we may (and will) assume that 
%\marginpar{apr9b}
\begin{equation}\label{apr9b}
\mbox{\em every disjunctive subhyperformula of $\mathbb{F}_4$ has at least one disjunct}
\end{equation}
(otherwise, if all of the disjuncts of a disjunctive subhyperformula of $\overline{\mathbb{F}_{1}^{\cal T}}$ have been  deleted, restore one --- arbitrary --- disjunct). Similarly, we may (and will) assume that 
%\marginpar{apr9bb}
\begin{equation}\label{apr9bb}
\mbox{\em whenever $\beta$ is a subhyperformula of $\mathbb{F}_4$ such that $\uorigin({\beta})$ is a $\mld$-unit, $\beta$ has two disjuncts.}
\end{equation}

Thus, $\mathbb{F}_4$ is a tautology of ordinary, ``finitary'' classical propositional logic. In view of (\ref{apr3a}), it is also binary. 

The earlier defined function $\uorigin$ of unital origin extends from subhyperformulas of $\overline{\mathbb{F}_{1}^{\cal T}}$ to those of $\mathbb{F}_{4}$ in an obvious way.  Namely, as noted just a while ago,  every   subhyperformula $\beta$ of $\mathbb{F}_4$ is the result of deleting some 
disjuncts in some subdisjunctions of some subhyperformula $\beta'$ of $\overline{\mathbb{F}_{1}^{\cal T}}$.  Then the unital origin of $\beta$ is the same as that of $\beta'$.

We are now going to show that 
%\marginpar{apr3b}
\begin{equation}\label{apr3b}
\fif\vdash \mathbb{F}_0.
\end{equation}
The above statement directly  contradicts our original assumption (\ref{feb17a}) that $\fif\not\vdash \mathbb{F}_0$. And, as we remember, obtaining a contradiction (from assumption (\ref{apr2b})) was the goal of the present proof. Thus, the only remaining task within our proof of Lemma \ref{feb26a} now is to verify (\ref{apr3b}).  

Let us say that a cirquent $C$ is {\bf derivable} (in $\fif$) from a cirquent $C'$ iff there is a sequence $C_1,\ldots,C_n$ of cirquents, called a {\bf derivation} of $C$ from $C'$, such that $C_1=C'$, $C_n=C$ and, for each $i\in\{1,\ldots,n-1\}$, $C_{i+1}$ follows from $C_{i}$ by one of the rules of $\fif$. 

We construct a $\fif$-proof of $\mathbb{F}_0$ bottom-up, starting from $\mathbb{F}_0$ and moving from conclusions to premises. The whole construction consists of two main parts. During the first part, we apply the below-described procedure $\first$, which takes (starts with) $\mathbb{F}_0$ --- more precisely, the cirquent $\mathbb{F}_{0}^{\clubsuit}$ --- and constructs a derivation of it from a certain cirquent $\mathbb{D}$ all of whose oformulas are literals (Lemma \ref{apr6b}). During the second part (Lemma \ref{apr8c}), we continue our bottom-up construction of a proof from $\mathbb{D}$ and hit an axiom, meaning that $\mathbb{D}$ is provable in $\fif$ and hence so is $\mathbb{F}_0$.

Here  we slightly expand the earlier-introduced formalism of cirquents through allowing to prefix some (possibly all, possibly none) $\cost$-oformulas of a cirquent by a $\surd$ (``check''). A $\cost$-oformula of the form $\surd \cost E$ will be said to be {\bf checked}, and all other oformulas said to be {\bf unchecked}. The presence or absence of the prefix $\surd$ has no effect on the applicability of rules: from the perspective of (the rules of) $\fif$, an oformula $\surd \cost E$ is treated as if it simply was $\cost E$.

At every step of our description of the procedure $\first$, $C$ stands for the ``current cirquent'', for which 
(as long as possible) we need to find a premise. $C$ is initialized to $\mathbb{F}_{0}^{\clubsuit}$. 

Next, $\first$ maintains a one-to-one mapping that sends each  oformula $E$ of the ``current cirquent'' $C$ to some subhyperformula $\alpha$ of $\mathbb{F}_4$, with such an $\alpha$ said to be the {\bf $C$-image  of $E$}, or the {\bf image of $E$ in $C$}, or simply the {\bf image of $E$} when $C$ is fixed or clear from the context, and denoted by \[\image^C(E).\] The image of the (only) oformula $\mathbb{F}_0$ of the initial cirquent $\mathbb{F}_{0}^{\clubsuit}$ is the hyperformula $\mathbb{F}_4$ itself. The mapping will be maintained in such a way that --- as will be easily seen from our description of the procedure\footnote{And hence we will not bother to explicitly verify it.} --- we always have:
\begin{conditions}\label{conds}  
%\marginpar{conds}
For any unchecked oformula $H$ of $C$, the $\mathbb{F}_0$-origin of $\uorigin(\image^C(H))$ is $H$; similarly, 
for any checked oformula $\surd\cost H$ of $C$, the $\mathbb{F}_0$-origin of $\uorigin(\image^C(\surd\cost H))$ is $H$.
Consequently: 
\begin{description}
  \item[(i)] where $L$ is a literal, the image of the oformula $L$ or $\surd\cost L$ is a hyperliteral (namely, a hyperliteral $\alpha$ such that the $\mathbb{F}_0$-origin of $\uorigin(\alpha)$ is $L$);
  \item[(ii)] the image of an oformula of the form $E\mlc G$ or $\surd\cost (E\mlc G)$ is of the form $\mlc\{\alpha,\beta\}$;
  \item[(iii)] the image of an oformula of the form $E\mld G$ or $\surd\cost (E\mld G)$ is of the form $\mld\{\alpha,\beta\}$;
  \item[(iv)] the image of an oformula of the form $\st E$ or $\surd\cost\st E$ is of the form $\mlc\{\alpha\}$;
  \item[(v)] the image of an oformula of the form $\cost E$ or $\surd\cost\cost E$ is of the form $\mld\{\alpha_1,\ldots,\alpha_n\}$ for some $n\geq 1$.
\end{description}
\end{conditions}

We refer to the (only) overgroup of the initial cirquent $\mathbb{F}_{0}^{\clubsuit}$ as the {\bf master overgroup}. Every new  (``non-master'')  overgroup that will emerge during the procedure will be {\bf labeled} with some $\st$-unit of $\mathbb{F}_{1}^{\cal T}$. The rules (in the bottom-up view) applied during the procedure $\first$ never destroy, merge or split overgroups,\footnote{Even if those rules may modify the order, quantity or contents of overgroups.} and every given overgroup, together with its label (as well as the status of being or not being the master overgroup) is  automatically inherited without changes in all subsequent/new cirquents. So, when $O$ is an overgroup of a given cirquent, we will terminologically treat the corresponding --- inherited from $O$ in an obvious sense of ``inheritance'' hardly requiring any formal explanation --- overgroup of a subsequent cirquent as ``the same overgroup'', and refer to the latter using the same name $O$. Similarly in the case of undergroups or oformulas whenever the corresponding inheritance is obvious.  

$\first$ proceeds by consecutively performing the following {\bf stages} 1 through 4 over and over again until none of the four stages results in any changes to the ``current cirquent'' $C$. Every stage will involve zero or more {\bf steps}, with each step changing the value of $C$ to that of one of its legitimate premises through applying (bottom-up) one of the rules of $\fif$.  The images of the oformulas not affected/modified by a given application of a rule (by a given step) are implicitly assumed to remain unchanged when moving from conclusion to premise.\vspace{10pt} 

{\em Stage 1}. Keep applying to $C$ Conjunction Introduction and Disjunction Introduction as long as possible.  As a result, we  end up with a cirquent that has no oformulas (this does not extend to osubformulas of oformulas though!) of the form $E\mlc G$ or $E\mld G$. 
Whenever $\mld\{\alpha,\beta\}$ was the image of an oformula $E\mlc G$ (see Condition \ref{conds}(ii)) of the cirquent and $E\mlc G$ was split into $E$ and $G$ as a result of applying Conjunction Introduction, the image of $E$ in the new cirquent is stipulated to be $\alpha$ and the image of $G$ is stipulated to be $\beta$. Similarly, whenever $\mld\{\alpha,\beta\}$ was the image of an oformula $E\mld G$ (see Condition \ref{conds}(iii)) of the cirquent and $E\mld G$ was split into $E$ and $G$ as a result of applying Disjunction Introduction, the image of $E$ in the new cirquent is $\alpha$ and the image of $G$ is $\beta$.

{\em Stage 2}. Assume $\st E$ is an oformula of $C$, and $\mlc\{\alpha\}$ is its image (see Condition \ref{conds}(iv)). Turn $\st E$ into $E$ using Recurrence Introduction. The newly emerged overgroup gets labeled with the unital origin of $\mlc\{\alpha\}$, and the image of the newly emerged oformula $E$ becomes $\alpha$. Repeat such a step as long as possible.

{\em Stage 3}.  Assume $\cost E$ is an (unchecked) oformula of $C$, and $\mld\{\alpha_1,\ldots,\alpha_n\}$ is its image (see Condition \ref{conds}(v)).  Using Contraction $n-1$ times, generate $n$ copies of $\surd \cost E$.\footnote{If here $n=1$, no rule is applied, and simply $\cost E$ is replaced by $\surd \cost E$.} For each $i$ with $1\leq i\leq n$, the image of the $i$'th copy of $\surd\cost E$ will be $\alpha_i$. Repeat such a step as long as possible.

{\em Stage 4}. Assume $\surd \cost E$ is a (checked) oformula of $C$, and $\alpha$ is its image. Further assume that, for every subhyperformula $\beta$ of $\mathbb{F}_4$, whenever $\uorigin(\beta)$ dominates $\uorigin(\alpha)$ in $\mathbb{F}_{1}^{\cal T}$, $\uorigin(\beta)$ is the label of one of the overgroups of $C$. Then, using Corecurrence Introduction, turn $\surd \cost E$ into $E$, and include $E$ (in addition to the old overgroups already containing $\surd \cost E$) in exactly those overgroups that are labeled by some unit $\st H$ such that $\st H$ dominates $\uorigin(\alpha)$ in $\mathbb{F}_{1}^{\cal T}$. The image of $E$ is stipulated to remain $\alpha$.  Repeat such a step as long as possible.\hspace{10pt}

\begin{lemma}\label{nov10a}
%\marginpar{nov10a}
Assume $C$ is the ``current cirquent'' at a given step of applying $\first$, and $E$ is an unchecked oformula of $C$. 
Then, with ``overgroup'' meaning one of $C$ and  ``dominates'' meaning ``dominates in $\mathbb{F}_{1}^{\cal T}$'', we have:

(a) $E$ is in the master overgroup.

(b) Whenever $E$ is in some non-master overgroup $O$, the label of $O$ dominates $\uorigin(\image^C(E))$.

(c) Whenever $\alpha$ is a subhyperformula of $\mathbb{F}_4$ and $\uorigin(\alpha)$ dominates $\uorigin(\image^C(E))$, there is an $\uorigin(\alpha)$-labeled overgroup $O$ such that $E$ is in $O$.

\end{lemma}

\begin{proof} Assume the conditions of the lemma.

{\em Clause (a)}: At the very beginning of the work of $\first$, the oformula $\mathbb{F}_0$ is included in the master overgroup. This inclusion is automatically inherited by all later-emerged oformulas of all later-emerged cirquents.

{\em Clause (b)}: Assume $O$ is a non-master overgroup of $C$ containing $E$. $C$ cannot be the initial cirquent $\mathbb{F}_{0}^{\clubsuit}$, because the only overgroup of the latter is the master overgroup. Let $C'$ be the cirquent immediately preceding $C$ (immediately below $C$) in the bottom-up derivation constructed by $\first$.

If the transition from $C'$ (the conclusion) to $C$ (the premise) did not modify $E$, then it is clear that $O$ already existed in $C'$ and contained $E$. Further,  
the image of $E$ in $C$ is the same as in $C'$.  Also, as always, the label of $O$ in $C$ is the same as in $C'$. By the induction hypothesis,\footnote{Where induction is on the number of the steps of $\first$ preceding the step at which $C$ was generated.} the label of $O$ dominates $\uorigin(\image^{C'}(E))$, and thus dominates $\uorigin(\image^C(E))$.    

Now assume the transition from $C'$ to $C$ modified $E$. 

One possibility is that $C'$ contained $E\mld G$ (or $G\mld E$, or $E\mlc G$, or $G\mlc E$, but these cases are similar) and $E$ emerged in $C$ during Stage 1 as a result of splitting this oformula into its two components $E$ and $G$. Note that $O$ already existed in $C'$ and contained $E\mld G$. So, by the induction hypothesis, the label of $O$ dominates 
$\uorigin(\image^{C'}(E\mld G))$. But $\uorigin(\image^{C}(E))$ is a  subunit of $\uorigin(\image^{C'}(E\mld G))$. Hence, by Lemma \ref{dec29a}, the label of $O$ also dominates $\uorigin(\image^{C}(E))$.

Another possibility is that $C'$ contained $\st E$, and $E$ emerged in $C$ during Stage 2 as a result of deleting the prefix $\st$. If $C'$ already contained $O$, the case is similar to the previous one. Otherwise, the label of $O$ is nothing but $\uorigin(\image^{C'}(\st E))$. But $\uorigin(\image^{C}(E))$ is a proper subunit of the latter and, by the definition of domination, is dominated by it. 

The final possibility is that $C'$ contained $\surd\cost E$, and $E$ emerged in $C$ during Stage 4 as a result of deleting the prefix $\surd\cost$. Note that $O$ already existed in $C'$. If $O$ contained $\surd\cost E$ in $C'$, then, by the induction hypothesis, the label of $O$ dominates $\uorigin(\image^{C'}(\surd\cost E))$. But $\image^{C'}(\surd\cost E)=\image^{C}(E)$. Thus, the label of $O$ dominates  $\uorigin(\image^{C}(E))$. And if $O$ did not contain $\surd\cost E$ in $C'$, then the label of $O$ dominates $\uorigin(\image^{C}(E))$, because otherwise $E$ would not have been included in  $O$ when transferring from $C'$ to $C$ according to the prescriptions of Stage 4. 

{\em Clause (c)}: Assume $\alpha$ is a subhyperformula of $\mathbb{F}_4$ and $\uorigin(\alpha)$ dominates $\uorigin(\image^C(E))$. In the initial cirquent $\mathbb{F}_{0}^{\clubsuit}$, the only oformula is $\mathbb{F}_0$ and the unital origin of its image is the root of $\mathbb{F}_{1}^{\cal T}$. By clause (i) of Lemma \ref{nov5}, the latter is not dominated by anything. So, $C$ cannot be the initial cirquent $\mathbb{F}_{0}^{\clubsuit}$.  
Let $C'$ be the cirquent immediately preceding $C$ in the derivation constructed by $\first$. 

If the transition from $C'$ (the conclusion) to $C$ (the premise) did not modify $E$, then, by the induction hypothesis, $C'$ already had an $\uorigin(\alpha)$-labeled overgroup $O$ such that $O$ contained $E$. This situation is then automatically inherited by $C$ from $C'$. 
  
Now assume the transition from $C'$ to $C$ modified $E$.

One possibility is that $C'$ contained $E\mld G$ (or $G\mld E$, or $E\mlc G$, or $G\mlc E$, but these cases are similar), and $E$ emerged in $C$ during Stage 1 as a result of splitting this oformula into its two components $E$ and $G$. In view of clause (a) of Lemma \ref{apr6a},   $\uorigin(\alpha)$ dominates $\uorigin(\image^C(E\mld G))$. Hence, by the induction hypothesis, an $\uorigin(\alpha)$-labeled overgroup $O$ existed in $C'$ and it included $E\mld G$. The same overgroup and the corresponding inclusion are 
inherited by $C$ and its oformula $E$.

Another possibility is that $C'$ contained $\st E$, and $E$ emerged in $C$ during Stage 2 as a result of deleting the prefix $\st$. In view of clause (b) of Lemma \ref{apr6a},  $\uorigin(\alpha)$ is either the same as $\uorigin(\image^{C'}(\st E))$ or dominates  $\uorigin(\image^{C'}(\st E))$. In the former case, the new overgroup emerged during the transition from $C'$ to $C$ is $\uorigin(\alpha)$-labeled and it includes $E$. In the latter case, by the induction hypothesis, an $\uorigin(\alpha)$-labeled overgroup $O$ existed in $C'$ and it included $\st E$. The same overgroup and the corresponding inclusion are 
inherited by $C$ and its oformula $E$.  

The final possibility is that $C'$ contained $\surd\cost E$, and $E$ emerged in $C$ during Stage 4 as a result of deleting the prefix $\surd\cost$. But in this case the conditions of Step 4 immediately guarantee that clause (c) of the present lemma is satisfied. 
\end{proof}

\begin{lemma}\label{nov10b}
%\marginpar{nov10b}
Assume $C$ is the ``current cirquent'' at some given step of applying $\first$, $L$ and $M$ are literals that happen to be oformulas of $C$,\footnote{Note that being an oformula of $C$ means being strictly ``on the surface'' of $C$. Namely, proper osubformulas of oformulas of $C$ do not count as oformulas of $C$.} and the unital origins of their $C$-images are opposite. Then $L$ and $M$ are opposite (one is the negation of the other), and they are contained in exactly the same overgroups of $C$.\end{lemma}

\begin{proof} Assume the conditions of the lemma. Below, ``dominates'' should be understood as 11dominates in $\mathbb{F}_{\cal T}$. That $L$ and $M$ are opposite is immediate from Condition \ref{conds}. Next, Lemma \ref{nov10c} implies that $\uorigin(\image^C(L))$ and $\uorigin(\image^C(M))$ are dominated by exactly the same $\st$-units of $\mathbb{F}_{1}^{\cal T}$. Let $S$ be the set of all $\st$-units  of $\mathbb{F}_{1}^{\cal T}$ that dominate
$\uorigin(\image^C(L))$ (and/or $\uorigin(\image^C(M))$) and that also happen to be the unital origins of some subhyperformulas  of $\mathbb{F}_4$. Further, let $S'$ be the (sub)set of  overgroups of $C$ consisting of the master overgroup and those overgroups whose labels are in $S$. Then, according to Lemma \ref{nov10a}, $S'$ is exactly the set of the overgroups of $C$ in which $L$ is contained, and the same holds for $M$.  
\end{proof}

What we below call a {\bf hypercirquent} is defined in exactly the same way as a cirquent, with the only difference that while cirquents are formula-based, hypercirquents are hyperformula-based. Correspondingly, ``{\bf ohyperformula}'' in the context of hypercirquents means the same as ``oformula'' in the context of cirquents.

We say that a hypercirquent $D$  is {\bf binary} iff so are all of its ohyperformulas and, in addition, for every hyperatom $\alpha$, there is at most one ohyperformula in $D$ in which $\alpha$ has a positive occurrence, and at most one ohyperformula in $D$ in which $\alpha$ has a negative occurrence. 
Next, we say that $D$ is {\bf tautological} iff, for every undergroup $U$ of $D$, the disjunction of all ohyperformulas of $U$ is a tautology.  Finally, we say that $D$ is a {\bf binary tautology} iff it is both binary and tautological. 

We agree that, whenever $C$ is a cirquent generated by $\first$, the {\bf image} of $C$, denoted by 
\[\image(C),\]
means the hypercirquent resulting from $C$ through replacing each oformula $E$ by $\image^C(E)$, otherwise retaining all arcs, groups and group labels.

The following two lemmas  can be verified by a routine induction on the number of cirquents generated by $\first$ earlier than $C$, in the same style as in the proof of Lemma \ref{nov10a}.\footnote{In the case of Lemma \ref{apr8b}, the basis of induction will be provided by that already observed fact that $\mathbb{F}_4$ is a binary tautology.} Hence we state them without proofs: 

\begin{lemma}\label{apr8b}
%\marginpar{apr8b}
Assume $C$ is the ``current cirquent'' at some given step of applying $\first$. Then $\image(C)$ is a binary tautology.  
\end{lemma}

\begin{lemma}\label{apr8a}
%\marginpar{apr8a}
Assume $C$ is the ``current cirquent'' at some given step of applying $\first$, and $\st H$ is the unital origin of some subhyperformula $\alpha$ of $\mathbb{F}_4$. Then either $\alpha$ is a subhyperformula of some ohyperformula of $\image(C)$,  or there is a $\st H$-labeled overgroup in $C$. 
\end{lemma}

In what follows, 
\[\mathbb{D}\]
stands for the final cirquent produced by the procedure $\first$. Since $\first$ always moves from a conclusion to a legitimate premise, we have: 
%\marginpar{apr9d}
\begin{equation}\label{apr9d}
\mbox{\em $\mathbb{F}_{0}^{\clubsuit}$ is derivable in $\fif$ from $\mathbb{D}$.}
\end{equation}

 \begin{lemma}\label{apr6b}
%\marginpar{apr6b}
Every oformula of $\mathbb{D}$ is a literal, and hence every ohyperformula of $\image(\mathbb{D})$ is a hyperliteral.
\end{lemma}

\begin{proof} Note that, if every oformula of $\mathbb{D}$ is a literal, then, in view of Condition \ref{conds}(i), every ohyperformula of $\image(\mathbb{D})$ is indeed a hyperliteral. So, we only need to verify that all oformulas of $\mathbb{D}$ are literals.   For a contradiction, deny this. Throughout this proof, domination means domination in $\mathbb{F}_{1}^{\cal T}$.  

$\mathbb{D}$ cannot have an oformula of the form $E\mld F$ or $E \mlc F$, because the corresponding Disjunction Introduction or Conjunction Introduction is always applicable, and thus such a formula would have been split into $E$ and $F$ during Stage 1 of \first. Next, there can be no oformulas of the form $\st F$ in $\mathbb{D}$ either,   as Stage 2 would immediately apply Recurrence Introduction and turn it into $F$. Similarly, $\mathbb{D}$ cannot have 
oformulas of the form $\cost F$, because such oformulas would have been modified during Stage 3. 

So, $\mathbb{D}$ must have some  oformulas of the form $\surd\cost F$.  Let $S$ be the set of all hyperformulas $\gamma$ such that $\gamma$ is a subhyperformula of some  ohyperformulas of 
$\image(\mathbb{D})$ and $\uorigin(\gamma)$ is a $\st$-unit.   In view of Lemma \ref{apr8a},  $S$ can be seen to be nonempty, for otherwise all oformulas of the form $\surd \cost F$ would have been eliminated during Stage 4 of $\first$. Next, in view of 
clauses (ii) and (iii) of Lemma \ref{nov5} (namely, the transitive and asymmetric --- and hence irreflexive --- properties of the relation of domination),  it is easy to see that there is an element $\alpha$ of $S$ whose unital origin is not dominated  by the unital origins of any 
other elements of $S$. Such an $\alpha$ obviously should be a subhyperformula of $\image^{\mathbb{D}}(\surd\cost E)$ for some oformula 
$\surd \cost E$ of $\mathbb{D}$  (because there are no $\mld$-, $\mlc$- or $\st$-oformulas in $\mathbb{D}$). Note that $\uorigin(\alpha)$ is a subunit of $\uorigin(\image^{\mathbb{D}}(\surd\cost E))$.
  But then, in view of clause (ii) of Lemma \ref{nov5},  $\uorigin(\image^{\mathbb{D}}(\surd \cost E))$ is not dominated by $\uorigin(\alpha)$. Nor is it dominated by the unital origin of any other element of $S$, because then, by Lemma \ref{dec29a}, so would be 
$\uorigin(\alpha)$. With Lemma \ref{apr8a} in mind, the observation we have just made means that, for every subhyperformula $\beta$ of $\mathbb{F}_4$, whenever $\uorigin(\beta)$ dominates $\uorigin(\image(\surd \cost E))$, $\uorigin(\beta)$ is the label of one of the overgroups of $\mathbb{D}$.  But, if so, $\surd\cost E$ cannot be an oformula of $\mathbb{D}$, because Stage 4 of $\first$ would have turned it into $E$. This is a contradiction. 
\end{proof}

\begin{lemma}\label{apr8c}
%\marginpar{apr8c}
$\fif\vdash \mathbb{D}$. 
\end{lemma}

\begin{proof} For any atom $P$, whenever $P$ (resp. $\gneg P$) is an oformula of $\mathbb{D}$ and $\alpha$ is a hyperatom with $\alpha=\image^{\mathbb{D}}(P)$ (resp. $\gneg\alpha=\image^{\mathbb{D}}(\gneg P)$), we call the atom $P$ the {\bf preimage} of  $\alpha$. In view of  Condition 
 \ref{conds}(i), the concept of preimage is well defined on the atoms found in $\mathbb{D}$: every hyperatom occurring in (some ohyperformula of) $\image({\mathbb{D}})$ has a unique preimage.
Now note that  
 $\mathbb{D}$ is  a {\em substitutional instance} of $\image({\mathbb{D}})$, in the sense that $\mathbb{D}$ is the result of replacing in $\image({\mathbb{D}})$ every (occurrence of every) hyperatom $\alpha$ by the  preimage of $\alpha$. 

For our present purposes, we ignore the irrelevant technicality that $\image({\mathbb{D}})$ is a hypercirquent rather than a cirquent, and treat it as an ordinary cirquent to which the ordinary rules of $\fif$ can be applied.\footnote{Of course, the rules of Contraction, Recurrence Introduction and Corecurrence Introduction automatically become redundant/inapplicable, because hyperformulas cannot contain $\sti$ or $\costi$.} In what follows, we show that $\fif\vdash \image({\mathbb{D}})$. In view of Condition \ref{conds}, this immediately implies that $\fif\vdash\mathbb{D}$, because  
any $\fif$-proof $T$ of $\image({\mathbb{D}})$ automatically turns into a $\fif$-proof of $\mathbb{D}$ once, in all cirquents of $T$, we replace every occurrence of every hyperatom by the preimage of that atom.

We agree that, in our description of a $\fif$-proof of $\image({\mathbb{D}})$, $C$ always stands for the ``current hypercirquent''. Also, in what follows, we do not explicitly mention Exchange even though the latter may have to be occasionally used  to make applications of the other rules possible. 

According to Lemma \ref{apr8b}, $\image(\mathbb{D})$ is a (binary) tautology. In view of Lemma \ref{apr6b}, 
the tautologicity of $\image(\mathbb{D})$ obviously implies that every undergroup of $\image(\mathbb{D})$ contains a pair of opposite hyperliterals. We choose one such pair for each undergroup of $\image(\mathbb{D})$. Then, starting from $C=\image(\mathbb{D})$ and moving in the ``from conclusion to premise'' direction, we apply a series of Weakenings to $C$ and delete from each undergroup all ohyperformulas except the chosen pair. 

Thus, now every undergroup $U$ of (the resulting value of) $C$ has exactly two ohyperformulas $\alpha$ and $\gneg\alpha$, where $\alpha$ is 
an atom. It is obvious that $C$ remains a binary tautology. The binarity of $C$, in turn, implies that, whenever two undergroups $U_1$ and $U_2$ 
share an ohyperformula,  they share {\em both} ohyperformulas --- that is, $U_1$ and $U_2$ have identical contents.  As our next step, using a 
series of Undergroup Duplications, from each set of identical-content undergroups, we eliminate all but one undergroup. Now, no two undergroups 
share ohyperformulas, and every undergroup $U$ (as before) contains exactly two ohyperformulas $\alpha_U$ and $\gneg\alpha_U$, where $\alpha_U$ 
is a hyperatom. Note also that, in view of Lemma \ref{nov10b}, every such pair $\alpha_U,\gneg\alpha_U$ of ohyperformulas is included in exactly the 
same overgroups. 

The content of 
every overgroup $O$ of  the resulting cirquent $C$ is $\{\alpha_{U_1},\gneg\alpha_{U_1},\ldots,\alpha_{U_n},\gneg\alpha_{U_n}\}$ for some undergroups 
$U_1,\ldots,U_n$. Now, applying a series of Mergings, we split every such overgroup $O$ into $n$ overgroups, whose contents are 
$\{\alpha_{U_1},\gneg\alpha_{U_1}\},\ldots,\{\alpha_{U_n},\gneg\alpha_{U_n}\}$. Finally, 
 using a series of Overgroup Duplications, from each set of identical-content overgroups, we eliminate all but one overgroup.  This way we end up with an axiom. 
\end{proof}

As an immediate corollary of (\ref{apr9d}) and Lemma \ref{apr8c}, we find that $\fif\vdash \mathbb{F}_{0}^{\clubsuit}$. Claim (\ref{apr3b}) is thus proven, and so is Lemma \ref{feb26a}. 

\section{Uniform completeness}\label{notice}
%\marginpar{notice}

In the following lemma, $\cal H$ and $e$ are the HPM and the valuation fixed in Section \ref{sscs}.

\begin{lemma}\label{apr8d}
%\marginpar{apr8d}
There is a constant enumeration interpretation $^*$ such that $\cal H$  does not win $\mathbb{F}_{0}^{*}$. Namely, $\cal H$  loses this game against $\cal E$ on $e$. 
\end{lemma}

\begin{proof} According to Lemma \ref{feb26a}, $\overline{\mathbb{F}_1}$ is not tautological. This means there is a hypermodel $h$ such that $h$ makes $\overline{\mathbb{F}_1}$ false. We now define the interpretation $^*$ by stipulating that, for every atom $P$, $P^*$ is the constant enumeration game such that, for any legal run $\Gamma$ of $P^*$, $\win{P^*}{}\seq{\Gamma}=\top$ iff 
$h$ makes the hyperliteral $(P,A,B)$ true, where $A$ and $B$ are the following sets of natural numbers:
\[\begin{array}{c}
 A\ =\ \{a\  |\ \mbox{\em $\Gamma$ contains the labmove $\top a$}\};\\
 B\ =\ \{b\  |\ \mbox{\em $\Gamma$ contains the labmove $\bot b$}\}. 
\end{array}\]
 It is not hard to see that $\Omega$ --- the ${\cal H}$ vs.\hspace{-1pt} ${\cal E}$-run on $e$ fixed in Section \ref{sscs} --- is a $\bot$-won run of $\mathbb{F}_{0}^{*}$. Thus, $\cal H$ loses  $\mathbb{F}_{0}^{*}$ against $\cal E$ on $e$.
\end{proof}

Now, remembering that $\mathbb{F}_0$ was an arbitrary (even if fixed) $\fif$-unprovable formula, Lemma \ref{apr8d} means nothing but that  the implication $(ii)\Rightarrow (i)$ (that is, $\gneg (i)\Rightarrow\gneg(ii)$) of Theorem  6 of  \cite{taming1}, in the strong form of clause (b) of that theorem, holds.

\begin{remark}\label{apr15a}
%\marginpar{apr15a} 
For good reasons (see, for instance, the end of Section 9 of \cite{lbcs}), CoL takes no interest in non-effective strategies. It would however be a pity to leave it unobserved that we have never really relied on the fact that the strategy of $\cal E$'s adversary was algorithmic. Namely, our entire argument and Lemma \ref{apr8d} in particular would just as well go through if $\cal H$ was an HPM with an arbitrary oracle,\footnote{See Section 18 of \cite{Jap03} for a discussion of HPMs with oracles.}  and hence its strategy was not-necessarily-effective.  Thus, whenever $\fif$ does not prove a formula $F$, 
there is simply no strategy --- whether effective or not --- that wins $F^*$ for every constant (and enumeration) interpretation $^*$; in other words, $F$ has no uniform solution even if we do not require that uniform solutions be effective. 
\end{remark}

\section{Multiform completeness}\label{sss}
%\marginpar{sss}

A transition from completeness with respect to uniform validity (``uniform completeness'') to completeness with respect to multiform validity (``multiform completeness'') is already rather standard in the literature on CoL. The diagonalization-style idea that we are going to employ below has been used in  \cite{Japtocl1,Japtocl2,Japtcs,Japjsl,Japtcs2,Propint,Japseq,Japtogl} for a number of other deductive systems for CoL in multiform completeness proofs.  

Let 
%\marginpar{list}
\begin{equation}\label{list}
{\cal H}_0,\ {\cal H}_1,\ {\cal H}_2,\ \ldots
\end{equation}
 be a fixed list of all HPMs (say, arranged in the lexicographic order of their descriptions). 

Next, we pick a variable $x$, and let 
%\marginpar{lista}
\begin{equation}\label{lista}
e_0,\ e_1,\ e_2,\ \ldots
\end{equation}
be the valuations such that, for each $n\geq 0$, $e_n$ assigns $n$ to $x$ and assigns $0$ to all other variables.

Remember that, even though $\cal H$ and $e$ have been fixed so far, they were an arbitrary HPM and an arbitrary valuation. This means that $\cal H$ and $e$ just as well could have been treated --- as we are going to do from now on --- as a variable ranging over HPMs and a variable ranging over valuations, respectively.  Everything we have proven about the fixed $\cal H$ and $e$ automatically goes through for $\cal H$ and $e$ as  variables. Namely,  considering only ({\em HPM,valuation}) pairs of the form $({\cal H}_n,e_n)$,  
Lemma \ref{apr8d} can now be re-written as follows:

\begin{lemma}\label{apr8e}
%\marginpar{apr8e}
For every natural number $n$ there  is a constant enumeration interpretation $^{*_{n}}$ --- which we choose and fix for the purposes of this section --- such that the HPM ${\cal H}_n$ from the list (\ref{list}) does not win the game $\mathbb{F}_{0}^{*_{n}}$ against our EPM $\cal E$ on the valuation $e_n$ from the list (\ref{lista}). 
\end{lemma}

We now define an enumeration interpretation $^\dagger$ such that $P^\dagger$  (any atom $P$) is a unary game that depends on no variables other than $x$. Namely, we let $P^\dagger$ be the game such that, for any valuation $e$, $e[P^\dagger]=P^{*_n}$, where $n$ is the value assigned to the variable $x$ by $e$,     and $^{*_n}$ is the interpretation fixed in Lemma \ref{apr8e} for $n$. Of course, $^\dagger$ is a unary interpretation, as promised in clause (c) of Theorem  6 of  \cite{taming1}. To complete our proof of that clause, we need to show that 
no HPM wins $\mathbb{F}_{0}^{\dagger}$. 

But indeed, consider an arbitrary HPM $\cal H$. Since all HPMs are on the list (\ref{list}), we have ${\cal H}={\cal H}_n$ for some $n$ (fix the latter). 
Next, consider the valuation $e_n$ from the list (\ref{lista}). By our choice of $^\dagger$, for every atom $P$ we have $e_n[P^\dagger]=P^{*_n}$. Clearly this extends from atoms to compound formulas as well, so 
$e_n[\mathbb{F}_{0}^{\dagger}]=\mathbb{F}_{0}^{*_n}$. By Lemma \ref{apr8e}, ${\cal H}_n$ --- i.e. $\cal H$ --- does not win $\mathbb{F}_{0}^{*_n}$ against $\cal E$ on $e_n$. 
Thus, $\cal H$ does not win $e_n[\mathbb{F}_{0}^{\dagger}]$ against $\cal E$ on $e_n$. But not winning $e_n[\mathbb{F}_{0}^{\dagger}]$ against $\cal E$ on $e_n$ 
obviously means the same as not winning $\mathbb{F}_{0}^{\dagger}$ against $\cal E$ on $e_n$. 
 Thus, $\cal H$ does not win $\mathbb{F}_{0}^{\dagger}$, as desired.

\end{document}